\documentclass{llncs}
\pdfoutput=1

\usepackage{algorithm}
\usepackage{makeidx}
\usepackage[noend]{algpseudocode}
\usepackage{amsfonts}
\usepackage{amsmath}
\usepackage{amssymb}
\usepackage{ifthen}
\usepackage{microtype}
\usepackage{refcount}
\usepackage{tikz}
\usetikzlibrary{arrows}
\usetikzlibrary{patterns}
\usetikzlibrary{shapes}
\usetikzlibrary{decorations}
\usepackage{url}
\usepackage{enumitem}
\setlist{noitemsep, nolistsep, topsep=0pt}

\newif\iflncs

\lncsfalse

\iflncs
 \usepackage{hyperref}
 \newcommand{\arxivappendix}[2]{%
 \href{http://arxiv.org/pdf/1512.00174v2.pdf\#nameddest=#2}{Appendix~#1}}
\else
\fi

\usepackage[textsize=small]{todonotes}

\algnewcommand\algorithmicforeach{\textbf{for each}}
\algdef{S}[FOR]{ForEach}[1]{\algorithmicforeach\ #1\ \algorithmicdo}%
\algnewcommand{\SimpleIf}[2]{\State\algorithmicif\ #1\ \algorithmicthen\ #2}
\algnewcommand{\SimpleElsIf}[2]{\State\algorithmicelse\ \algorithmicif\ #1\ \algorithmicthen\ #2}

\let\leq\leqslant
\let\geq\geqslant
\let\setminus\smallsetminus
\let\epsi\varepsilon
\let\rho\varrho
\let\implies\Rightarrow

\newcommand{\brac}[1]{{\left(#1\right)}}

\newcommand{\Oh}[1]{O\brac{#1}}

\newcommand{\classfont}{\mathsf}
\newcommand{\NP}{\ensuremath{\classfont{NP}}}
\newcommand{\problemfont}{\textsc}
\newcommand{\threesat}{\problemfont{3Sat}}
\newcommand{\pmthreesat}{\problemfont{PlanarMonotone3Sat}}
\newcommand{\threepart}{\problemfont{3Partition}}

\makeatletter
\newcommand\ie{i.e\@ifnextchar.{}{.\@}}
\newcommand\etc{etc\@ifnextchar.{}{.\@}}
\newcommand\etal{et~al\@ifnextchar.{}{.\@}}
\makeatother

\makeatletter
\let\rc@refused\refused
\makeatother

\newcounter{hackcount}
\newcommand{\hackLcounter}[1]{\setcounter{hackcount}{\value{lemma}}\setcounterref{lemma}{#1}\addtocounter{lemma}{-1}}
\newcommand{\unhackLcounter}{\setcounter{lemma}{\value{hackcount}}}

\newcommand{\hackTcounter}[1]{\setcounter{hackcount}{\value{theorem}}\setcounterref{theorem}{#1}\addtocounter{theorem}{-1}}
\newcommand{\unhackTcounter}{\setcounter{theorem}{\value{hackcount}}}

\newcommand{\subcase}[1]{\noindent\textbf{#1}}

\begin{document}

%************************************
\title{The Partial Visibility Representation Extension Problem%
\thanks{%
This work was partially supported by ESF project EUROGIGA GraDR and preliminary ideas were formed during HOMONOLO 2014.
G.\,Gu{\'s}piel was partially supported by the MNiSW grant DI2013 000443.
G.\,Gutowski was partially supported by the Polish National Science Center grant UMO-2011/03/D/ST6/01370.
T.\,Krawczyk was partially supported by the Polish National Science Center grant UMO-2015/17/B/ST6/01873.
G.\,Liotta was partially supported by the MIUR project AMANDA ``Algorithmics for MAssive and Networked DAta'', prot. 2012C4E3KT\_001. 
\iflncs
The full version of this paper is available on arXiv:1512.00174~\cite{ARXIV}. 
Note: whenever we refer to sections in the \emph{Appendix} we mean the appendix of 
\href{http://arxiv.org/pdf/1512.00174v2.pdf}{arXiv:1512.00174v2}
\else
Appears in the Proceedings of the 24th International Symposium on Graph Drawing and Network Visualization (GD 2016).
\fi
}
}

\author{Steven~Chaplick\inst{1} \and
Grzegorz Gu\'spiel\inst{2} \and
Grzegorz Gutowski\inst{2} \and \\
Tomasz Krawczyk\inst{2} \and
Giuseppe~Liotta\inst{3}}
\authorrunning{S.~Chaplick, G.~Gu\'spiel, G.~Gutowski, T.~Krawczyk, and G.~Liotta}
\institute{
Lehrstuhl f\"ur Informatik I, Universit\"at W\"urzburg, Germany%
\\ \email{steven.chaplick@uni-wuerzburg.de}
\and
Theoretical Computer Science Department, Faculty of Mathematics and Computer Science, Jagiellonian University, Krak\'ow, Poland%
\\ \email{\{guspiel,gutowski,krawczyk\}@tcs.uj.edu.pl}
\and
Dipartimento di Ingegneria, Universit{\`a} degli Studi di Perugia, Italy%
\\ \email{giuseppe.liotta@unipg.it}
}

%\subjclass{I.3.5 Computational Geometry and Object Modeling}
%\keywords{Visibility Representations, Partial Layout}%

\maketitle

\begin{abstract}
For a graph $G$, a function $\psi$ is called a \emph{bar visibility representation} of $G$ when for each vertex $v \in V(G)$, $\psi(v)$ is a horizontal line segment (\emph{bar}) and $uv \in E(G)$ iff there is an unobstructed, vertical, $\epsi$-wide line of sight between $\psi(u)$ and $\psi(v)$.
Graphs admitting such representations are well understood (via simple characterizations) and recognizable in linear time.
For a directed graph $G$, a bar visibility representation $\psi$ of $G$, additionally, for each directed edge $(u,v)$ of $G$, puts the bar $\psi(u)$ strictly below the bar $\psi(v)$.
We study a generalization of the recognition problem where a function $\psi'$ defined on a subset $V'$ of $V(G)$ is given and the question is whether there is a bar visibility representation $\psi$ of $G$ with $\psi|V' = \psi'$.
We show that for undirected graphs this problem together with closely related problems are \NP-complete, but for certain cases involving directed graphs it is solvable in polynomial time.
\end{abstract}

\section{Introduction}
\label{sec:introduction}
The concept of a visibility representation of a graph is a classic one in computational geometry and graph drawing and the first studies on this concept date back to the early days of these fields (see, e.g.~\cite{TamassiaT86,Wismath85} and~\cite{GhoshG13} for a recent survey).
In the most general setting, a visibility representation of a graph is defined as a collection of disjoint sets from an Euclidean space such that the vertices are bijectively mapped to the sets and the edges correspond to unobstructed lines of sight between two such sets.
Many different classes of visibility representations have been studied via restricting the space (e.g., to be the plane), the sets (e.g., to be points or line segments) and/or the lines of sight (e.g., to be non-crossing or axis-parallel).
In this work we focus on a classic visibility representation setting in which the sets are horizontal line segments (\emph{bars}) in the plane and the lines of sight are vertical.
As such, whenever we refer to a visibility representation, we mean one of this type.
The study of such representations was inspired by the problems in VLSI design and was conducted by different authors~\cite{DuchetHLM83,LuccioMW87,OttenW78} under variations of the notion of visibility.
Tamassia and Tollis~\cite{TamassiaT86} gave an elegant unification of different definitions and we follow their approach.

A \emph{horizontal bar} is an open, non-degenerate segment parallel to the $x$-axis of the coordinate plane.
For a set $\Gamma$ of pairwise disjoint horizontal bars, a \emph{visibility ray} between two bars $a$ and $b$ in $\Gamma$ is a vertical closed segment spanned between bars $a$ and $b$ that intersects $a$, $b$, and no other bar.
A \emph{visibility gap} between two bars $a$ and $b$ in $\Gamma$ is an axis aligned, non-degenerate open rectangle spanned between bars $a$ and $b$ that intersects no other bar.

For a graph $G$, a \emph{visibility representation} $\psi$ is a function that assigns a distinct horizontal bar to each vertex such that these bars are pairwise disjoint and satisfy additional visibility constraints.
There are three standard visibility models:
\begin{itemize}
\item {\em Weak visibility.}
  In this model, for each edge $\{u,v\}$ of $G$, there is a visibility ray between $\psi(u)$ and $\psi(v)$ in $\psi(V(G))$.
\item {\em Strong visibility.}
  In this model, two vertices $u$, $v$ of $G$ are adjacent if and only if there is a visibility ray between $\psi(u)$ and $\psi(v)$ in $\psi(V(G))$.
\item {\em Bar visibility.}
  In this model, two vertices $u$, $v$ of $G$ are adjacent if and only if there is a visibility gap between $\psi(u)$ and $\psi(v)$ in $\psi(V(G))$.
\end{itemize}
The bar visibility model is also known as the $\epsi$-visibility model in the literature.

A graph that admits a visibility representation in any of these models is a planar graph, but the converse does not hold in general.
Tamassia and Tollis~\cite{TamassiaT86} characterized the graphs that admit a visibility representation in these models as follows.
A graph admits a weak visibility representation if and only if it is planar.
A graph admits a bar visibility representation if and only if it has a planar embedding with all cut-points on the boundary of the outer face.
For both of these models, Tamassia and Tollis~\cite{TamassiaT86} presented linear time algorithms for the recognition of representable graphs, and for constructing the appropriate visibility representations.
The situation is different for the strong visibility model.
Although the planar graphs admitting a strong visibility representation are characterized in~\cite{TamassiaT86} (via strong $st$-numberings),
Andreae~\cite{Andreae92} proved that the recognition of such graphs is \NP-complete.
Summing up, from a computational point of view, the problems of recognizing graphs that admit visibility representations and of constructing such representations are well understood.

Recently, a lot of attention has been paid to the question of extending partial representations of graphs.
In this setting a representation of some vertices of the graph is already fixed and the task is to find a representation of the whole graph that extends the given partial representation. %(see, e.g.~\cite{BlasiusR13,ChaplickFK13,KlavikKKW12,KlavikKORSSV14,KlavikKOSV13,KlavikKOS15} for papers that study computational aspects of extending partial representations of geometric intersection graphs).
Problems of this kind are often encountered in graph drawing and are sometimes computationally harder than testing for existence of an unconstrained drawing.
The problem of extending partial drawings of planar graphs is a good illustration of this phenomenon.
On the one hand, by F\'ary's theorem, every planar graph can be drawn in the plane so that each vertex is represented as a point, and edges are pairwise non-crossing, straight-line segments joining the corresponding points.
Moreover, such a drawing can be constructed in linear time.
On the other hand, testing whether a partial drawing of this kind (\ie{}, an assignment of points to some of the vertices) can be extended to a straight-line drawing of the whole graph is \NP-hard~\cite{Patrignani06}.
However, an analogous problem in the model that allows the edges to be drawn as arbitrary curves instead of straight-line segments has a linear-time solution~\cite{AngeliniBFJKPR15}.
A similar phenomenon occurs when we consider contact representations of planar graphs.
Every planar graph is representable as a disk contact graph or a triangle contact graph.
Every bipartite planar graph is representable as a contact graph of horizontal and vertical segments in the plane.
Although such representations can be constructed in polynomial time, the problems of extending partial representations of these kinds are \NP-hard~\cite{ChaplickDKMS14}.

In this paper we initiate the study of extending partial visibility representations of graphs.
From a practical point of view, it may be worth recalling that visibility representations are not only an appealing way of drawing graphs, but they are also typically used as an intermediate step towards constructing visualizations of networks in which all edges are oriented in a common direction and some vertices are aligned (for example to highlight critical activities in a PERT diagram).
Visibility representations are also used to construct orthogonal drawings with at most two bends per edge.
See, e.g.~\cite{BattistaETT99} for more details about these applications.
The partial representation extension problem that we study in this paper occurs, for example, when we want to use visibility representations to incrementally draw a large network and we want to preserve the user's mental map in a visual exploration that adds a few vertices and edges per time.

Both for weak visibility and for strong visibility, the partial representation extension problems are easily found to be \NP-hard.
For weak visibility, the hardness follows from results on contact representations by Chaplick \etal{}~\cite{ChaplickDKMS14}.
For strong visibility, it follows trivially from results by Andreae~\cite{Andreae92}.
Our contribution is the study of the partial representation extension problem for bar visibility.
Hence, the central problem for this paper is the following:

\smallskip
\noindent \underline{Bar Visibility Representation Extension}:\smallskip \\
\noindent{\bf Input:}
 $(G, \psi')$; $G$ is a graph; $\psi'$ is a map assigning bars to a $V' \subseteq V(G)$.\\
\noindent{\bf Question:}
Does $G$ admit a bar visibility representation $\psi$ with $\psi|V' = \psi'$?\smallskip\\
One of our results is the following.

\begin{theorem}\label{th:var-emb-hardness-general}
The Bar Visibility Representation Extension Problem is \NP-complete.
\end{theorem}
The proof is a standard reduction from \pmthreesat{} problem, which is known to be \NP-complete thanks to de Berg and Khosravi~\cite{BergK10}.
The reduction uses gadgets that simulate logic gates and constructs a planar boolean circuit that encodes the given formula.
Theorem~\ref{th:var-emb-hardness-general} is proven in%
{
\iflncs
\arxivappendix{D}{section.A.4}.%
\else
Appendix~\ref{app:hardness}.%
\fi
}
We investigate a few natural modifications of the problem.
Most notably, we study the version of the problem for directed graphs.
We provide some efficient algorithms for extension problems in this setting.
%Observe that
A visibility representation induces a natural orientation on edges of the graph -- each edge is oriented from the lower bar to the upper one.
This leads to the definition of a visibility representation for a directed graph.
The function $\psi$ is a representation of a digraph $G$ if, additionally to satisfying visibility constraints, for each directed edge $(u,v)$ of $G$, the bar $\psi(u)$ is strictly below the bar $\psi(v)$.
Note that a planar digraph that admits a visibility representation also admits an \emph{upward planar drawing}~(see e.g., \cite{GargT95}), that is, a drawing in which the edges are represented as non-crossing $y$-monotone curves. 

A \emph{planar $st$-graph} is a planar acyclic digraph with exactly one source $s$ and exactly one sink $t$ which admits a planar embedding such that $s$ and $t$ are on the outer face.
Di Battista and Tamassia~\cite{BattistaT88} proved that a planar digraph admits an upward planar drawing if and only if it is a subgraph of a planar $st$-graph if and only if
it admits a weak visibility representation.
Garg and Tamassia~\cite{GargT01} showed that the recognition of planar digraphs that admit an upward planar drawing is \NP-complete.
It follows that the recognition of planar digraphs that admit a weak visibility representation is \NP-complete, and so is the corresponding partial representation extension problem.
Nevertheless, as is shown in Lemma~\ref{lem:planar_digraphs_bar_visibility} (see%
{
\iflncs
\arxivappendix{A}{section.A.1}%
\else
Appendix~\ref{app:digraphs}% 
\fi
}
for the proof), the situation might be different for bar visibility.
\begin{lemma}
\label{lem:planar_digraphs_bar_visibility}
Let $st(G)$ be a graph constructed from a planar digraph $G$ by adding two vertices $s$ and $t$, the edge $(s,t)$, an edge $(s,v)$ for each source vertex $v$ of $G$, and an edge $(v,t)$ for each sink vertex $v$ of $G$.
A planar digraph $G$ admits a bar visibility representation if and only if the graph $st(G)$ is a planar $st$-graph.
\end{lemma}
%Since planar $st$-graphs can be recognized in linear time, planar digraphs that admit a bar visibility representation are also recognizable in linear time.
As planar $st$-graphs can be recognized in linear time, the same is true for planar digraphs that admit a bar visibility representation. % are also recognizable in linear time.
The natural problem that arises is the following:

\smallskip
\noindent \underline{Bar Visibility Representation Extension for Digraphs}:\smallskip\\
\noindent{\bf Input:}
 $(G, \psi')$; $G$ is a digraph; $\psi'$ is a map assigning bars to a $V' \subseteq V(G)$.\\
\noindent{\bf Question:}
Does $G$ admit a bar visibility representation $\psi$ with $\psi|V' = \psi'$?\smallskip\\
Although we do not provide a solution for this problem, we present an efficient algorithm for an important variant.
A bar visibility representation $\psi$ of a directed graph $G$ is called \emph{rectangular} if $\psi$ has a unique bar $\psi(s)$ with the lowest $y$-coordinate, a unique bar $\psi(t)$ with the highest $y$-coordinate, $\psi(s)$ and $\psi(t)$ span the same $x$-interval, and all other bars are inside the rectangle spanned between $\psi(s)$ and $\psi(t)$.
See Figure~\ref{fig:large} for an example of a rectangular bar visibility representation of a planar $st$-graph.

\begin{figure*}
  \begin{tikzpicture}[>=latex]
    \definecolor{light-gray}{gray}{0.60}
    \definecolor{areafill-gray}{gray}{0.90}
    %\tikzstyle{every node}=[circle,fill=black]
    \begin{scope}[xscale=0.5, yscale=0.35]
      \begin{tiny} {
        \thinmuskip=0.5mu
        \medmuskip=0.5mu plus 0.5mu minus 0.5mu
        \thickmuskip=0.5mu plus 0.5mu minus 0.5mu
        
        \begin{scope}[shift={(0,0)}]
          \node[] (s) at (5,-0.25) {$s$};
          \node[] (1) at (0,2.5) {$1$};
          \node[] (2) at (2,1.5) {$2$};
          \node[] (3) at (2,4.0) {$3$};
          \node[] (4) at (3.5,2.75) {$4$};
          \node[] (5) at (3.5,6.5) {$5$};
          \node[] (6) at (4.5,8.5) {$6$};
          \node[] (7) at (6.5,2) {$7$};
          \node[] (8) at (5,4) {$8$};
          \node[] (9) at (7.5,4) {$9$};
          \node[] (10) at (6,6) {$10$};
          \node[] (11) at (7.5,1) {$11$};
          \node[] (12) at (10,2) {$12$};
          \node[] (13) at (10,8) {$13$};
          \node[] (14) at (6,9.5) {$14$};
          \node[] (t) at (6,11.5) {$t$};

          \draw (s) edge[->, bend left=20] (1);
          \draw (s) edge[->] (2);
          \draw (s) edge[->] (5);
          \draw (s) edge[->] (7);
          \draw (s) edge[->] (11);
          \draw (s) edge[->, bend right=20] (12);
          \draw (1) edge[->, bend left] (5);
          \draw (2) edge[->] (3);
          \draw (2) edge[->] (4);
          \draw (3) edge[->] (5);
          \draw (4) edge[->] (5);
          \draw (5) edge[->] (6);
          \draw (5) edge[->, bend left] (t);
          \draw (6) edge[->] (14);
          \draw (7) edge[->] (8);
          \draw (7) edge[->] (9);
          \draw (8) edge[->] (10);
          \draw (9) edge[->] (10);
          \draw (10) edge[->] (13);
          \draw (10) edge[->] (14);
          \draw (11) edge[->] (12);
          \draw (11) edge[->] (13);
          \draw (12) edge[->] (13);
          \draw (13) edge[->] (14);
          \draw (13) edge[->] (t);
          \draw (14) edge[->] (t);
        \end{scope}

        \begin{scope}[shift={(13,0.75)}]
          \draw[fill=light-gray, color=light-gray] (2,0) rectangle (3,6);
          \draw[fill=light-gray, color=light-gray] (5,0) rectangle (6,6);
          \node at (5.5,10.4) {$\psi(t)$};
          \draw[thick, (-)] (0,10)--(11,10);
          \node at (5.5,9.4) {$\psi(14)$};
          \draw[thick, (-)] (1,9)--(10,9);
          \node at (9,8.4) {$\psi(13)$};
          \draw[thick, (-)] (7,8)--(11,8);
          \node at (3.5,8.4) {$\psi(6)$};
          \draw[thick, (-)] (1,8)--(6,8);
          \node at (7,7.4) {$\psi(10)$};
          \draw[thick, (-)] (6,7)--(8,7);
          \node at (3,6.4) {$\psi(5)$};
          \draw[thick, (-)] (0,6)--(6,6);
          \node at (10,5.4) {$\psi(12)$};
          \draw[thick, (-)] (9,5)--(11,5);
          \node at (7.6,4.5) {$\psi(9)$};
          \draw[thick, (-)] (7,4)--(8,4);
          \node at (6.4,4.5) {$\psi(8)$};
          \draw[thick, (-)] (6,4)--(7,4);
          \node at (4.6,4.5) {$\psi(4)$};
          \draw[thick, (-)] (4,4)--(5,4);
          \node at (3.4,4.5) {$\psi(3)$};
          \draw[thick, (-)] (3,4)--(4,4);
          \node at (9,3.4) {$\psi(11)$};
          \draw[thick, (-)] (8,3)--(10,3);
          \node at (1,2.4) {$\psi(1)$};
          \draw[thick, (-)] (0,2)--(2,2);
          \node at (7,1.4) {$\psi(7)$};
          \draw[thick, (-)] (6,1)--(8,1);
          \node at (4,1.4) {$\psi(2)$};
          \draw[thick, (-)] (3,1)--(5,1);
          \node at (5.5,0.4) {$\psi(s)$};
          \draw[thick, (-)] (0,0)--(11,0);
          
        \end{scope}
        }
      \end{tiny}
    \end{scope}
  \end{tikzpicture}
  \caption{A planar $st$-graph $G$ and a rectangular bar visibility representation $\psi$ of $G$.
  }
  \label{fig:large}
\end{figure*}
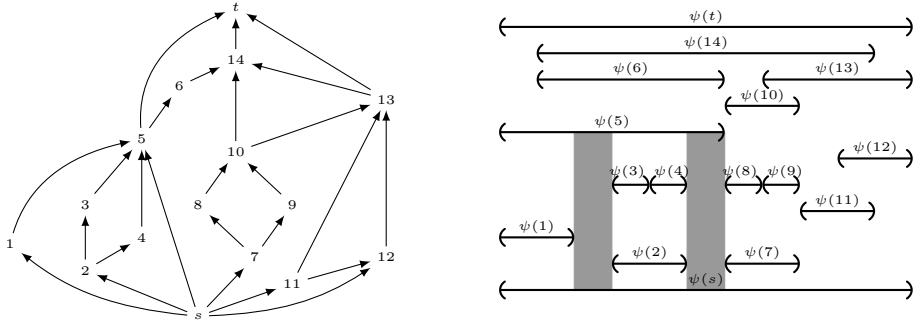

Tamassia and Tollis~\cite{TamassiaT86} showed that a planar digraph $G$ admits a rectangular bar visibility representation if and only if $G$ is a planar $st$-graph.
In Section~\ref{sec:rectangular} we give an efficient algorithm for the following problem:

\smallskip
\noindent \underline{Rectangular Bar Visibility Representation Extension for $st$-graphs}:\smallskip\\
\noindent{\bf Input:}
$(G, \psi')$; $G$ is a planar $st$-graph; $\psi'$ is a map assigning bars to a $V' \subseteq V(G)$.\\
\noindent{\bf Question:}
Does $G$ admit a rectangular bar visibility representation $\psi$ with $\psi|V' = \psi'$?\smallskip\\
The main result in this paper is the following.

\begin{theorem}\label{th:main-plane-st}
The Rectangular Bar Visibility Representation Extension Problem for an $st$-graph with $n$ vertices can be solved in $\Oh{n\log^2{n}}$ time.
\end{theorem}
Our algorithm exploits the correspondence between bar visibility representations and $st$-orientations of planar graphs, and utilizes the $SPQR$-decomposition.

The rest of the paper is organized as follows.
%\noindent\textbf{Outline:}
Section~\ref{sec:preliminaries} contains the necessary definitions and description of the necessary tools.
Section~\ref{sec:rectangular} contains the general ideas for the proof of Theorem~\ref{th:main-plane-st}.
The omitted parts of the proof are reported in% 
{
\iflncs
\arxivappendix{C}{section.A.3}%
\else
Appendix~\ref{app:rectangular}% 
\fi
}
together with some figures illustrating the ideas behind the proofs.
Section~\ref{sec:conclusion} mentions further results from the full version and open problems. 

\section{Preliminaries}
\label{sec:preliminaries}

For a horizontal bar $a$, functions $y(a)$, $l(a)$, $r(a)$ give respectively the $y$-coordinate of $a$, the $x$-coordinate of the left end of $a$, and the $x$-coordinate of the right end of $a$.
For any bounded object $Q$ in the plane, we use functions $X(Q)$ and $Y(Q)$ to denote the smallest possible, possibly degenerate, closed interval containing the projection of $Q$ on the $x$-, and on the $y$-axis respectively.
We denote the left end of $X(Q)$ by $l(Q)$ and the right end of $X(Q)$ by $r(Q)$.
Let $a$ and $b$ be two horizontal bars with $y(a) < y(b)$.
We say that $Q$ is \emph{spanned between $a$ and $b$} if $X(Q) \subseteq X(a)$, $X(Q) \subseteq X(b)$, and $Y(Q) = [y(a),y(b)]$.

For a graph $G$, we often describe the visibility representation $\psi$ by providing the values of functions $y_\psi = y(\psi(v))$, $l_\psi = l(\psi(v))$, $r_\psi = r(\psi(v))$ for any vertex $v$ of $G$.
We drop the subscripts when the representation is known from the context.

Let $G$ be a planar $st$-graph.
An \emph{$st$-embedding} of $G$ is any planar embedding with $s$ and $t$ on the boundary of the outer face.
A planar $st$-graph together with an $st$-embedding is called a \emph{plane $st$-graph}.
Vertices $s$ and $t$ of a planar (plane) $st$-graph are called the \emph{poles} of $G$.
We abuse notation and we use the term \emph{planar (plane) $uv$-graph} to mean a planar (plane) $st$-graph with poles $u$ and $v$.
An \emph{inner vertex} of $G$ is a vertex of $G$ other than the poles of $G$.
A real valued function $\xi$ from $V(G)$ is an \emph{$st$-valuation} of $G$ if for each edge $(u, v)$ we have $\xi(u) < \xi(v)$.

Tamassia and Tollis~\cite{TamassiaT86} showed that the following properties hold for any plane $st$-graph:
\begin{enumerate}
  \item \label{prop:st_inner_face} For every inner face $f$, the boundary of $f$ consists of two directed paths with a common origin and a common destination.
  \item \label{prop:st_outer_face} The boundary of the outer face consists of two directed paths, with a common origin $s$ and a common destination $t$.
  \item \label{prop:st_bipolar} For every inner vertex $v$, edges from $v$ (to $v$) are consecutive around $v$.
\end{enumerate}

Let $G$ be a plane $st$-graph.
We introduce two objects associated with the outer face of $G$: the \emph{left outer face} $s^*$ and the \emph{right outer face} $t^*$.
Properties~\eqref{prop:st_inner_face}-\eqref{prop:st_bipolar} allow us to introduce the following standard notions:
\emph{left/right face} of an edge and a vertex,
\emph{left/right path} of a face,
and the \emph{dual} $G^*$ of $G$ -- a planar $st$-graph with vertex set consisting of inner faces of $G$, $s^*$, and $t^*$.
For two faces $f$ and $g$ in $V(G^*)$ we say that \emph{$f$ is to the left of $g$}, and that \emph{$g$ is to the right of $f$}, if there is a directed path from $f$ to $g$ in $G^*$.
See%
{
\iflncs
\arxivappendix{B.2}{subsection.A.2.2}%
\else
Appendix~\ref{app:plane-st-graphs}%
\fi
} 
for the precise definitions which follow the standard definitions given by Tamassia and Tollis~\cite{TamassiaT86}.

\section{Rectangular bar visibility representations of $st$-graphs}
\label{sec:rectangular}

In this section we provide an efficient algorithm that solves the rectangular bar visibility representation extension problem for $st$-graphs.
Our algorithm employs a specific version of the $SPQR$-decomposition that allows us to describe all $st$-embeddings of a planar $st$-graph. 
See%
{
\iflncs
\arxivappendix{B.1}{subsection.A.2.1}%
\else
Appendix~\ref{app:SPQR-trees}%
\fi
} 
for the exact definition which follows the one given by Di Battista and Tamassia~\cite{BattistaT96}.
In particular, an $SPQR$-tree $T$ of a planar $st$-graph $G$ consists of nodes of four different types: $S$ for \emph{series nodes}, $P$ for \emph{parallel nodes}, $Q$ for \emph{edge nodes}, and $R$ for \emph{rigid nodes}.
Each node $\mu$ represents a \emph{pertinent graph} $G_\mu$, a subgraph of $G$ which is an $st$-graph with poles $s_\mu$ and $t_\mu$.
Additionally, $\mu$ has an associated directed multigraph $skel(\mu)$ called the \emph{skeleton} of $\mu$.
The only difference between our definition of the $SPQR$-tree and the one given in~\cite{BattistaT96} is that we do not add an additional edge between the poles of the skeleton of a node.
Our definition ensures that we have a one-to-one correspondence between the edges of $skel(\mu)$ and the children of $\mu$ in $T$. In Section~\ref{sec:rectangular_properties}, we use the $SPQR$-tree $T$ of $G$ to describe how a rectangular bar visibility representation is composed of rectangular bar visibility representations of the pertinent graphs of $T$.

The skeleton of a rigid node has only two $st$-embeddings, one being the flip of the other around the poles of the node.
The skeleton of a parallel node with $k$ children has $k!$ $st$-embeddings, one for every permutation of the edges of the skeleton.
The skeleton of a series node or an edge node has only one $st$-embedding.

Section~\ref{sec:rectangular_properties} presents structural properties of bar visibility representations in relation to an $SPQR$-decomposition. 
In Section~\ref{sec:rectangular_algorithm_slow} we present an algorithm that solves this extension problem in quadratic time.
In%
{
\iflncs
\arxivappendix{C.6}{subsection.A.3.6}%
\else
Appendix~\ref{app:rectangular_algorithm_fast}%
\fi
} 
we give a refined algorithm that works in $\Oh{n\log^2n}$ time for an $st$-graph with $n$ vertices.

\subsection{Structural properties}
\label{sec:rectangular_properties}

Let $\Gamma$ be a collection of pairwise disjoint bars.
For a pair of bars $a$, $b$ in $\Gamma$ with $y(a) < y(b)$ let the \emph{set of visibility rectangles} $R(a,b)$ be
the interior of the set of points $(x,y)$ in $\mathbb{R}^2$ where:
\begin{enumerate}
\item $a$ is the first bar in $\Gamma$ on a vertical line downwards from $(x,y)$,
\item $b$ is the first bar in $\Gamma$ on a vertical line upwards from $(x,y)$.
\end{enumerate}
Figure \ref{fig:large} shows (shaded area) the set of visibility rectangles $R(s,5)$. 
Note that there is a visibility gap between $a$ and $b$ in $\Gamma$ iff $R(a,b)$ is non-empty.
%Additionally,
If $R(a,b)$ is non-empty, then it is a union of pairwise disjoint open rectangles spanned between $a$ and $b$.

Let $G$ be a planar $st$-graph and let $T$ be the $SPQR$-tree for $G$.
Let $\psi$ be a rectangular bar visibility representation of $G$.
For every node $\mu$ of $T$ we define the set $B_{\psi}(\mu)$, called the \emph{bounding box of $\mu$ with respect to $\psi$}, as the closure of the following union:
$$\bigcup \left\{ R(\psi(u),\psi(v)): (u,v) \text{ is an edge of the pertinent digraph $G_\mu$} \right\}\text{.}$$
If $\psi$ is clear from the context, then the set $B_\psi(\mu)$ is denoted by $B(\mu)$ and is called the \emph{bounding box of $\mu$}.
Let $B(\psi) = X(\psi(V(G))) \times Y(\psi(V(G)))$ be the minimal closed axis-aligned rectangle that contains the representation $\psi$.
It follows that:
\begin{enumerate}
  \item $B(\psi) = B_\psi(\mu)$, where $\mu$ is the root of $T$,
  \item each point in $B(\psi)$ is in the closure of at least one set of visibility rectangles $R(\psi(u),\psi(v))$ for some edge $(u,v)$ of $G$,
  \item each point in $B(\psi)$ is in at most one set of visibility rectangles.
\end{enumerate}
The following two lemmas describe basic properties of a bounding box.

\begin{lemma}[Q-Tiling Lemma]
\label{lem:tiling-lemma-Q-node}
Let $\mu$ be a $Q$-node in $T$ corresponding to an edge $(u,v)$ of $G$.
For any rectangular bar visibility representation $\psi$ of $G$ we have:
\begin{enumerate}
\item $B(\mu)$ is a union of pairwise disjoint rectangles spanned between $\psi(u)$ and $\psi(v)$.
\item If $B(\mu)$ is not a single rectangle, then the parent $\lambda$ of $\mu$ in $T$ is a $P$-node, and $u$, $v$ are the poles of the pertinent digraph $G_{\lambda}$.
\end{enumerate}
\end{lemma}

The Basic Tiling Lemma presented below describes the relation between the bounding box of an inner node $\mu$ and the bounding boxes of the children of $\mu$ in any rectangular bar visibility representation of $G$.
The next lemma justifies the name \emph{bounding box} for $B(\mu)$.

\begin{lemma}[Basic Tiling Lemma]
\label{lem:basic-tiling-lemma}
Let $\mu$ be an inner node in $T$ with children $\mu_1,\ldots,\mu_k$, $k \geq 2$.
For a rectangular bar visibility representation $\psi$ of $G$ we have:
\begin{enumerate}
 \item \label{lem:basic-tiling-lemma-containment} $\psi(v) \subseteq B(\mu)$ for every inner vertex $v$ of $G_\mu$.
 \item \label{lem:basic-tiling-lemma-rectangle} $B(\mu)$ is a rectangle that is spanned between $\psi(s_\mu)$  and $\psi(t_\mu)$.
 \item \label{lem:basic-tiling-lemma-tiling} The sets $B(\mu_1), \ldots, B(\mu_k)$ tile the rectangle $B(\mu)$, \ie{},
 $B(\mu_1), \ldots, B(\mu_k)$ cover $B(\mu)$ and the interiors of $B(\mu_1), \ldots, B(\mu_k)$ are pairwise disjoint.
\end{enumerate}
\end{lemma}

In the next three lemmas we specialize the Basic Tiling Lemma depending on whether $\mu$ is a $P$-node, an $S$-node, or an $R$-node.
These lemmas allow us to describe all tilings of $B(\mu)$ by bounding boxes of $\mu$'s children.
For Lemmas~\ref{lem:tiling-lemma-P-node},~\ref{lem:tiling-lemma-S-node}, and~\ref{lem:tiling-lemma-R-node} we let $\mu_1, \ldots, \mu_k$ be $\mu$'s children.
The next lemma follows from the Basic Tiling Lemma and the Q-Tiling Lemma.
\begin{lemma}[P-Tiling Lemma]
\label{lem:tiling-lemma-P-node}
Let $\mu$ be a $P$-node.
For any rectangular bar visibility representation $\psi$ of $G$ we have:
\begin{enumerate}
  \item If $(s_\mu, t_\mu)$ is not an edge of $G$, then the sets $B(\mu_1),\ldots,B(\mu_k)$ are rectangles spanned between $\psi(s_\mu)$ and $\psi(t_\mu)$.
  \item If $(s_\mu, t_\mu)$ is an edge of $G$, then $\mu$ has exactly one child that is a $Q$-node, say $\mu_1$, and:
 \begin{itemize}
  \item For $i=2,\ldots,k$, $B(\mu_i)$ is a rectangle spanned between $\psi(s_\mu)$ and $\psi(t_\mu)$.
  \item $B(\mu_1) \neq \emptyset$ is a union of rectangles spanned between $\psi(s_\mu)$ and $\psi(t_\mu)$.
 \end{itemize}
\end{enumerate}
\end{lemma}

When $\mu$ is an $S$-node or an $R$-node, then there is no edge $(s_\mu,t_\mu)$.
By the Q-Tiling Lemma and by the Basic Tiling Lemma, each set $B(\mu_i)$ is a rectangle that is spanned between the bars representing the poles of $G_{\mu_i}$.

\begin{lemma}[S-Tiling Lemma]
\label{lem:tiling-lemma-S-node}
Let $\mu$ be an $S$-node.
Let $c_1,\ldots,c_{k-1}$ be the cut-vertices of $G_\mu$ encountered in this order on a path from $s_\mu$ to $t_\mu$.
Let $c_0 = s_\mu$, and $c_k = t_\mu$.
For any rectangular bar visibility representation $\psi$ of $G$, for every $i=1,\ldots,k-1$, we have $X(\psi(c_i)) = X(B(\mu))$.
For every $i=1,\ldots,k$, $B(\mu_i)$ is spanned between $\psi(c_{i-1})$ and $\psi(c_i)$ and $X(B(\mu_i)) = X(B(\mu))$.
\end{lemma}

The R-Tiling Lemma should describe all possible tilings of the bounding box of an $R$-node $\mu$ that appear in all representations of $G$.
Since there is a one-to-one correspondence between the edges of $skel(\mu)$ and the children of $\mu$, we abuse notation and write $B(u,v)$ to denote the bounding box of the child of $\mu$ that corresponds to the edge $(u,v)$. % of $skel(\mu)$.
By the Basic Tilling Lemma, $B(u,v)$ is spanned between the bars representing $u$ and $v$.

Suppose that $\psi$ is a representation of $G$.
The tiling $\tau = (B_{\psi}(\mu_1), \ldots, B_{\psi}(\mu_k))$ of $B_{\psi}(\mu)$ determines a triple $(\mathcal{E}, \xi, \chi)$, where:
 $\mathcal{E}$ is an $s_{\mu}t_{\mu}$-embedding of $skel(\mu)$,
 $\xi$ is an $st$-valuation of $\mathcal{E}$,
 and $\chi$ is an $st$-valuation of $\mathcal{E}^*$,
that are defined as follows.
Consider the following planar drawing of the $st$-graph $skel(\mu)$.
Draw every vertex $u$ in the middle of $\psi(u)$,
and every edge $e=(u,v)$ as a curve that starts in the middle of $\psi(u)$, goes a little above $\psi(u)$ towards the rectangle $B_\psi(u,v)$,
goes inside $B_\psi(u,v)$ towards $\psi(v)$, and a little below $\psi(v)$ to the middle of $\psi(v)$.
This way we obtain a plane $st$-graph $\mathcal{E}$, which is an $st$-embedding of $skel(\mu)$.
The $st$-valuation $\xi$ of $\mathcal{E}$ is just the restriction of $y _\psi$ to the vertices from $skel(\mu)$, \ie{}, $\xi = y_{\psi}|V(skel(\mu))$.
To define the $st$-valuation $\chi$ of $\mathcal{E}^*$ we use the following lemma.

\begin{lemma}[Face Condition]\label{lem:face-condition}\
\begin{enumerate}
\item
\label{lem:fc_right_face}
Let $f$ be a face in $V(\mathcal{E}^*)$ different than $t^*$, and let $v_0,v_1,\ldots,v_p$ be the right path of $f$.
There is a vertical line $L_r(f)$ that contains the left endpoints of $\psi(v_1), \ldots, \psi(v_{p-1})$ and the left sides of $B_\psi(v_0,v_1), \ldots, B_\psi(v_{p-1},v_p)$.
\item
\label{lem:fc_left_face}
Let $f$ be a face in $V(\mathcal{E}^*)$ different than $s^*$, and let $u_0,u_1,\ldots,u_m$ be the left path of $f$.
There is a vertical line $L_l(f)$ that contains the right endpoints of $\psi(u_1), \ldots, \psi(u_{q-1})$ and the right sides of $B_\psi(u_0,u_1), \ldots, B_\psi(u_{q-1},u_q)$.
\item
\label{lem:fc_face}
If $f$ is an inner face of $\mathcal{E}$ then $L_l(f)=L_r(f)$.
\end{enumerate}
\end{lemma}

The above lemma allows us to introduce the notion of a \emph{splitting line} for every face $f$ in $V(\mathcal{E}^*)$;
namely, it is: the line $L_l(f)=L_r(f)$ if $f$ is an inner face of $\mathcal{E}$, $L_r(f)$ if $f$ is the left outer face of $\mathcal{E}$, and $L_l(f)$ if $f$ is the right outer face of $\mathcal{E}$.
Now, let $\chi(f)$ be the $x$-coordinate of the splitting line for a face $f$ in $V(\mathcal{E}^*)$.
To show that $\chi(f)$ is an $st$-valuation of $\mathcal{E}^*$, note that for any edge $(f,g)$ of $\mathcal{E}^*$ there is an edge $(u,v)$ of $\mathcal{E}$ that has $f$ on the left side and $g$ on the right side.
It follows that $\chi(f) = l(B_\psi(u,v)) < r(B_\psi(u,v)) = \chi(g)$, proving the claim.

The representation $\psi$ of $G$ determines the triple $(\mathcal{E}, \xi, \chi)$.
Note that any other representation with the same tiling $\tau = (B_{\psi}(\mu_1), \ldots ,B_\psi(\mu_k))$ of $B(\mu)$ gives the same triple.
To emphasize that the triple $(\mathcal{E}, \xi, \chi)$ is determined by tiling $\tau$, we write $(\mathcal{E}_\tau, \xi_\tau, \chi_\tau)$.

Now, assume that $\mathcal{E}$ is an $st$-embedding of $skel(\mu)$, $\xi$ is an $st$-valuation of $\mathcal{E}$, and $\chi$ is an $st$-valuation of the dual of $\mathcal{E}$.
Consider the function $\phi$ that assigns to every vertex $v$ of $skel(\mu)$ the bar $\phi(u)$ defined as follows: $y_\phi(v) = \xi(v)$, $l_\phi(v) = \chi(\text{left face of } v)$, $r_\phi(v) = \chi(\text{right face of } v)$.
Firstly, Tamassia and Tollis~\cite{TamassiaT86} showed that $\phi$ is a bar visibility representation of $skel(\mu)$ and that for $\tau = (B_\phi(\mu_1), \ldots, B_\phi(\mu_k))$, we have $(\mathcal{E}_\tau, \xi_\tau, \chi_\tau) = (\mathcal{E}, \xi, \chi)$.
Secondly, there is a representation $\psi$ of $G$ that agrees with $\tau$ on $skel(\mu)$, \ie{}, such that $\tau = (B_\psi(\mu_1), \ldots, B_\psi(\mu_k))$.
To construct such a representation, we take any representation $\psi$ of $G$, translate and scale all bars in $\psi$ to get $B_\psi(\mu) = B_\phi(\mu)$, and represent the pertinent digraphs $G_{\mu_1}, \ldots, G_{\mu_k}$ so that the bounding box of $\mu_i$ coincides with $B_\phi(\mu_i)$ for $i=1,\ldots,k$.
This leads to the next lemma.
\begin{lemma}[R-Tiling Lemma]
\label{lem:tiling-lemma-R-node}
Let $\mu$ be an $R$-node.
There is a bijection between the set $\{(B_\psi(\mu_1), \ldots, B_\psi(\mu_k)):\psi$ is a rectangular bar visibility representation of $G\}$ of all possible tilings of the bounding box of $\mu$ by the bounding boxes of $\mu_1,\ldots,\mu_k$ in all representations of $G$, and the set $\{(\mathcal{E}, \xi, \chi):\mathcal{E}$ is an $st$-embedding of $skel(\mu)$, $\xi$ is an $st$-valuation of $\mathcal{E}$, $\chi$ is an $st$-valuation of the dual of $\mathcal{E}\}$.
\end{lemma}

\subsection{Algorithm}
\label{sec:rectangular_algorithm_slow}

Let $G$ be an $n$-vertex planar $st$-graph and let $\psi'$ be a partial representation of $G$ with the set $V'$ of fixed vertices.
We present a quadratic time algorithm that tests if there exists a rectangular bar visibility representation $\psi$ of $G$ that extends $\psi'$.
If such a representation exists, the algorithm can construct it in the same time.

In the first step, our algorithm calculates $y_\psi$.
Namely, the algorithm checks whether $y_{\psi'}:V' \to \mathbb{R}$ is extendable to an $st$-valuation of $G$.
When such an extension does not exist, the algorithm rejects the instance $(G, \psi')$;
otherwise any extension of $y_{\psi'}$ can be used as $y_\psi$.
The next lemma verifies this step's correctness.
%The correctness of this step is verified by the following lemma.
\begin{lemma}
\label{lem:y-st-valuation}
Let $\psi$ be a rectangular bar visibility representation of $G$ that extends~$\psi'$.
\begin{enumerate}
  \item \label{claim:st-valuation} The function $y_{\psi}$ is an $st$-valuation of $G$ that extends $y_{\psi'}$,
  \item \label{claim:st-valuation-any} If $y$ is an $st$-valuation of $G$ that extends $y_{\psi'}$, then a function $\phi$ that sends every vertex $v$ of $G$ into a bar so that $y_\phi(v) = y(v)$, $l_\phi(v) = l_{\psi}(v)$, $r_\phi(v) = r_{\psi}(v)$ is also a rectangular bar visibility representation of $G$ that extends $\psi'$.
\end{enumerate}
\end{lemma}
Clearly, checking whether $y_{\psi'}$ is extendable to an $st$-valuation of $G$, and constructing such an extension can be done in $\Oh{n}$ time.
In the second step, the algorithm computes the $SPQR$-tree $T$ for $G$, which also takes linear time.

Before we describe the last step in our algorithm, we need some preparation.
For an inner node $\mu$ in $T$ we define the sets $V'(\mu)$ and $C(\mu)$ as follows:
$$\begin{array}{rcl}
  V'(\mu) & = & \text{the set of fixed vertices in $V(G_{\mu}) \setminus \{s_\mu, t_\mu\}$,}\\
  C(\mu) & = & \left\{
\begin{array}{l}
\text{$\emptyset$, if $V'(\mu) = \emptyset$;} \\
\text{the smallest closed rectangle containing $\psi'(u)$ for all $ u \in V'(\mu)$, otherwise.}
\end{array}
\right.\\
\end{array}
$$
The set $C(\mu)$ is called the \emph{core} of $\mu$.
For a node $\mu$ whose core is empty, our algorithm can represent $G_\mu$ in any rectangle spanned between the poles of $G_\mu$.
Thus, we focus our attention on nodes whose core is non-empty.

Assume that $\mu$ is a node whose core is non-empty.
We describe the `possible shapes' the bounding box of $\mu$ might have in a representation of $G$ that extends $\psi'$.
The bounding box of $\mu$ is a rectangle that is spanned between the bars corresponding to the poles of $G_\mu$.
By the Basic Tiling Lemma, if $C(\mu)$ is non-empty then $B(\mu)$ contains $C(\mu)$.
For our algorithm it is important to distinguish whether the left (right) side of $B(\mu)$ contains the left (right) side of $C(\mu)$.
This criterion leads to four types of representations of $\mu$ with respect to the core of $\mu$.

The main idea of the algorithm is to decide for each inner node $\mu$ whose core is non-empty, which of the four types of representation of $\mu$ are possible and which are not.
The algorithm traverses the tree bottom-up and for each node and each type of representation it tries to construct the appropriate tiling using the information about possible representations of its children.
The types chosen for different children need to fit together to obtain a tiling of the parent node.
In what follows, we present our approach in more detail.

Let $\mu$ be an inner node in $T$.
Fix $\phi' = \psi'|V'(\mu)$.
Function $\phi'$ gives a partial representation of the pertinent digraph $G_\mu$ obtained by restricting $\psi'$ to the inner vertices of $G_\mu$.
Let $x,x'$ be two real values.
A rectangular bar visibility representation $\phi$ of $G_\mu$ is called an
\emph{$[x,x']$-representation of $\mu$} if $\phi$ extends $\phi'$ and $X(\phi(s_\mu)) = X(\phi(t_\mu)) = [x,x']$.
We say that an $[x,x']$-representation of $\mu$ is:
\begin{itemize}
\item \emph{left-loose, right-loose} (\emph{LL}), when $x < l(C(\mu))$ and $x' > r(C(\mu))$,
\item \emph{left-loose, right-fixed} (\emph{LF}), when $x < l(C(\mu))$ and $x' = r(C(\mu))$,
\item \emph{left-fixed, right-loose} (\emph{FL}), when $x = l(C(\mu))$ and $x' > r(C(\mu))$,
\item \emph{left-fixed, right-fixed} (\emph{FF}), when $x = l(C(\mu))$ and $x' = r(C(\mu))$.
\end{itemize}

The next lemma justifies this categorization of representations.
It says that if a representation of a given type exists, then every representation of the same type is also realisable.
\begin{lemma}[Stretching Lemma]\label{lem:stretching} Let $\mu$ be an inner node whose core is non-empty.
 If $\mu$ has an LL-representation, then $\mu$ has an $[x,x']$-representation for any $x < l(C(\mu))$ and any $x' > r(C(\mu))$.
 If $\mu$ has an LF-representation, then $\mu$ has an $[x,x']$-representation for any $x < l(C(\mu))$ and $x' = r(C(\mu))$.
 If $\mu$ has an FL-representation, then $\mu$ has an $[x,x']$-representation for $x = l(C(\mu))$ and any $x' > r(C(\mu))$.
\end{lemma}

The main task of the algorithm is to verify which representations are feasible for nodes that have non-empty cores.
We assume that:
 $\mu$ is an inner node whose core is non-empty;
 $\mu_1,\ldots,\mu_k$ are the children of $\mu$, $k \geq 2$;
 $\lambda_1,\ldots,\lambda_{k'}$ are the children of $\mu$ with $C(\lambda_i) \neq \emptyset$, $0 \leq k' \leq k$;
 $\theta(\lambda_i)$ is the set of feasible types of representations for $\lambda_i$, $\theta(\lambda_i) \subseteq \{LL,LF,FL,FF\}$.
We process the tree bottom-up and assume that $\theta(\lambda_i)$ is already computed and non-empty.

Let $x$ and $x'$ be two real numbers such that $x \leq l(C(\mu))$ and $x' \geq r(C(\mu))$.
We provide an algorithm that tests whether an $[x,x']$-representation of $\mu$ exists.
We use it to find feasible types for $\mu$ by calling it $4$ times with appropriate values of $x$ and $x'$.
While searching for an $[x,x']$-representation of $\mu$ our algorithm tries to tile the rectangle $[x,x'] \times [y(s_\mu),y(t_\mu)]$ with $B(\mu_1), \ldots, B(\mu_k)$.
The tiling procedure is determined by the type of $\mu$.
Note that as the core of a $Q$-node is empty, the algorithm splits into three cases: $\mu$ is an $S$-node, a $P$-node, and an $R$-node.
The pseudocode for the algorithms is given in%
{
\iflncs
\arxivappendix{C.3}{subsection.A.3.3}.%
\else
Appendix~\ref{app:pseudocode}.%
\fi
}

\subcase{Case S. $\mu$ is an $S$-node.}
In this case we
%Algorithm~\ref{alg:s-node} attempts
attempt to align the left and right side of the bounding box of each child $\lambda$ of $\mu$ to $x$ and $x'$ respectively.
For example, if the core of $\lambda$ is strictly contained in $[x,x']$, then $\lambda$ must have an LL-representation. The other cases follow similarly.
We also must set the $x$-intervals of the bars of the cut vertices of $G_\mu$ to $[x,x']$.
%The correctness of this approach
%Algorithm~\ref{alg:s-node}
%follows directly from the S-Tiling Lemma and the Stretching Lemma.
The S-Tiling Lemma and the Stretching Lemma imply the correctness of this approach.

\subcase{Case P. $\mu$ is a $P$-node.}
In this case we attempt to tile the rectangle $[x,x'] \times [y(s_\mu),y(t_\mu)]$ by placing the bounding boxes of the children of $\mu$ side by side from left to right.
The order of children whose cores are non-empty is determined by the position of those cores.
We sort $\lambda_1,\ldots,\lambda_{k'}$ by the left ends of their cores.
Let $l_i = l(C(\lambda_i))$ and $r_i = r(C(\lambda_i))$, $r_0=x$, $l_{k'+1}=x'$, and without loss of generality $l_1 < \ldots < l_{k'}$.

We need to find enough space to place the bounding boxes of children whose cores are empty.
Additionally, if $(s_\mu,t_\mu)$ is an edge of $G$, then we need to leave at least one visibility gap in the tiling for that edge.
Otherwise, if $(s_\mu,t_\mu)$ is not an edge of $G$, we need to close all the gaps in the tiling.
A more detailed description of the algorithm follows.

If there are $\lambda_i,\lambda_{i+1}$ such that the interior of the set $X(C(\lambda_i)) \cap X(C(\lambda_{i+1}))$ is non-empty, then we prove that there is no $[x,x']$-representation of $G_\mu$.
Indeed, by the P-Tiling Lemma and by $C(\lambda_i) \subseteq B(\lambda_i)$, the interior of $B(\lambda_i) \cap B(\lambda_{i+1})$ is non-empty and hence tiling of $B(\mu)$ with $B(\mu_1),\ldots,B(\mu_k)$ is not possible.
Additionally, if $r(C(\lambda_i)) = l(C(\lambda_{i+1}))$, then neither a right-loose representation of $\lambda_{i}$ nor a left-loose representation of $\lambda_{i+1}$ can be used, so we delete such types of representations from $\theta(\lambda_i)$ and $\theta(\lambda_{i+1})$.
If that leaves some $\theta(\lambda_i)$ empty, then an $[x,x']$-representation of $\mu$ does not exist.
These checks take $\Oh{k'}$ time.

Let $Q_i = [r_{i}, l_{i+1}] \times [y(s_\mu), y(t_\mu)]$ for $i \in [0,k']$.
We say that $Q_{i}$ is an \emph{open gap} (after $\lambda_{i}$, before $\lambda_{i+1}$) if $Q_i$ has non-empty interior.
In particular, if $x = r_0 < l_1$ ($r_{k'} < l_{k'+1} = x'$) then there is an open gap before $\lambda_1$ (after $\lambda_{k'}$).
On the one hand, if there is an edge $(s_\mu,t_\mu)$ or there is at least one $\mu_i$ whose core is empty then we need at least one open gap to construct an $[x,x']$-representation.
On the other hand, if $(s_\mu,t_\mu)$ is not an edge of $G$ then we need to close all the gaps in the tiling.
There are two ways to close the gaps.
Firstly, the representation of each child node whose core is empty can be placed so that it closes a gap.
The second way is to use loose representations for children nodes $\lambda_1,\ldots,\lambda_{k'}$.

Suppose that $c$ is a function that assigns to every $\lambda_i$ a feasible type of representation from the set $\theta(\lambda_i)$.
Whenever $c(\lambda_{i})$ is right-loose or $c(\lambda_{i+1})$ is left-loose, we can stretch the representation of $\lambda_i$ or $\lambda_{i+1}$, so that it closes the gap $Q_i$.
We describe a simple greedy approach to close the maximum number of gaps in this way.
We processes the $\lambda_i$'s from left to right and for each one:
we close both adjacent gaps if we can (\ie{} $LL \in \theta(\lambda_i)$); otherwise, we prefer to close the left gap if it is not yet closed rather than the right gap.
This is optimal by a simple greedy exchange argument.

If there are still $g>0$ open gaps left and $(s_\mu, t_\mu)$ is not an edge of $G$, then each open gap needs to be closed by placing in this gap a representation of one or more of the children whose core is empty.
Thus, it is enough to check that $k-k' \geq g$.
The correctness of the described algorithm follows by the P-Tiling Lemma, and the Stretching Lemma.

\subcase{Case R. $\mu$ is an $R$-node.}
The detailed discussion of this case is reported in%
{
\iflncs
\arxivappendix{C.4}{subsection.A.3.4}.%
\else
Appendix~\ref{app:rnode}.%
\fi
}
Here, we sketch our approach.
By the R-Tiling Lemma, the set of possible tilings of $B(\mu)$ by $B(\mu_1), \ldots, B(\mu_k)$
is in correspondence with the triples $(\mathcal{E}, \xi, \chi)$, where $\mathcal{E}$ is a planar embedding of $skel(\mu)$, $\xi$ is an $st$-valuation of $\mathcal{E}$, and $\chi$ is an $st$-valuation of $\mathcal{E}^*$.
To find an appropriate tiling of $B(\mu)$ (that yields an $[x,x']$-representation of $\mu$) we search through the set of such triples.
Since $\mu$ is a rigid node, there are only two $st$-embeddings of $skel(\mu)$ and we consider both of them separately.
Let $\mathcal{E}$ be one of these planar embeddings.
Since the $y$-coordinate for each vertex of $G$ is already fixed, the $st$-valuation $\xi$ is given by the $y$-coordinates of the vertices from $skel(\mu)$.
It remains to find an $st$-valuation $\chi$ of $\mathcal{E}^*$, \ie{}, to determine the $x$-coordinate of the splitting line for every face.

We claim that the existence of an $st$-valuation $\chi$ is equivalent to checking the satisfiability of a carefully designed 2-CNF formula.
For every child $\lambda$ of $\mu$ whose core is non-empty, we introduce two boolean variables that indicate which type (LL, LF, FL, FF) of representation is used for $\lambda$.
Additionally, for every inner face $f$ of $\mathcal{E}$ we introduce two boolean variables: the first (the second) indicates if the splitting line of $f$ is set to the leftmost (rightmost) possible position determined by the bounding boxes of nodes on the left (right) path of $f$.
Now, using those variables, we can express that:
feasible representations of the children nodes are used,
splitting line of a face $f$ agrees with the choice of representation for the nodes on the boundary of $f$ (see Face Condition Lemma),
the choice of splitting lines gives an $st$-valuation of $\mathcal{E}^*$.

In%
{
\iflncs
\arxivappendix{C.4}{subsection.A.3.4}%
\else
Appendix~\ref{app:rnode}%
\fi
} 
we present a formula construction that uses a quadratic number of clauses and results in a quadratic time algorithm.
In%
{
\iflncs
\arxivappendix{C.6}{subsection.A.3.6},%
\else
Appendix~\ref{app:rectangular_algorithm_fast},%
\fi
} 
we present a different, less direct, approach that constructs smaller formulas for $R$-nodes and leads to the $\Oh{n\log^2{n}}$ time algorithm.
Therefore, Lemma~\ref{lem:y-st-valuation}, and the discussion of cases S, P, and R, together with the results in%
{
\iflncs
\arxivappendix{C}{section.A.3}%
\else
Appendix~\ref{app:rectangular}%
\fi
} 
imply Theorem~\ref{th:main-plane-st}.

\section{Concluding Remarks and Open Problems}
\label{sec:conclusion}
We considered the representation extension problem for bar visibility representations and provided an efficient algorithm for st-graphs and showed NP-completeness for planar graphs. 
An important variant of bar visibility representations is when all bars used in the representation have integral coordinates, i.e., \emph{grid representations}. 
Any visibility representation can be easily modified into a grid representation.
However, this transformation does not preserve coordinates of the given vertex bars.  Indeed, we can show (in%
{
\iflncs
\arxivappendix{D.3}{subsection.A.4.3}%
\else
Appendix~\ref{app:hardness_grid}%
\fi
} 
that the (Rectangular) Bar Visibility Representation Extension problem is NP-hard on series-parallel st-graphs when one desires a grid representation. 

We conclude with two natural, interesting open problems.
The first one is to decide if there exists a polynomial time algorithm that checks whether a partial representation of a directed planar graph is extendable to a bar visibility representation of the whole graph.
Although we show an efficient algorithm for an important case of planar $st$-graphs, it seems that some additional ideas are needed to resolve this problem in general.
The second one is to decide if there is an efficient algorithm for recognition of digraphs admitting strong visibility representation, and for the corresponding partial representation extension problem.

\bibliographystyle{splncs03}

\bibliography{PVRE}

\iflncs
\else
\newpage

\appendix

\section{Planar digraphs that admit bar visibility representations}\label{app:digraphs}

\hackLcounter{lem:planar_digraphs_bar_visibility}
\begin{lemma}
Let $st(G)$ be a graph constructed from a planar digraph $G$ by adding two vertices $s$ and $t$, the edge $(s,t)$, edges $(s,v)$ for each source vertex $v$ of $G$, and edges $(v,t)$ for each sink vertex $v$ of $G$.
A planar digraph $G$ admits a bar visibility representation if and only if the graph $st(G)$ is a planar $st$-graph.
\end{lemma}
\unhackLcounter

\begin{proof}%[Proof of Lemma~\ref{lem:planar_digraphs_bar_visibility}]
Suppose that $st(G)$ is a planar $st$-graph.
Tamassia and Tollis~\cite{TamassiaT86} showed that $st(G)$ has a rectangular bar visibility representation $\psi$ with the bottom-most bar $\psi(s)$ and the top-most bar $\psi(t)$.
Clearly, $\psi|V(G)$ is a bar visibility representation for $G$.

Conversely, assume that $\psi$ is a bar visibility representation of $G$ and $\Gamma$ is the image of $\psi$.
We claim that some bars in $\Gamma$ can be extended, so that the new set of bars represents the same graph and have the following property:
For every vertical line $L$ intersecting a bar in $\Gamma$, if we traverse $L$ upwards:
\begin{itemize}
 \item the first encountered bar from $\Gamma$ represents a source of $G$,
 \item the last encountered bar from $\Gamma$ represents a sink of $G$.
\end{itemize}
In other words, we can transform every bar representation of $G$ so that the bars representing the sources of $G$
are the only bars visible from below and the bars representing the sinks of $G$ are the only bars visible from above.

Let $X(\psi) = X(\bigcup \psi(V(G)))$.
To get the above condition for $\psi$, consider the following procedure.
For any bar $b$ in $\psi(V(G))$, we extend $b$ as much as possible without introducing new visibilities and without increasing $X(\psi)$.
The procedure finishes, when no bar in the representation can be further extended.
Suppose this procedure completes and some non-source bar $b_v$ is visible from below.
Clearly, there is a neighbour $u$ of $v$ with $b_u$ below $b_v$ where $b_u$ can be extended without introducing new visibilities.

Suppose that $\psi$ satisfies the above condition.
Clearly, we can define two bars $\psi(s)$ and $\psi(t)$ such that $X(s)=X(t)$,
$X(\psi) \subsetneq X(s)$, and $y_{\psi}(s) < y_{\psi}(v) < y_{\psi}(t)$ for every vertex $v$ of $G$.
This extension of $\psi$ is a rectangular bar representation of $st(G)$.
It follows that $st(G)$ is a planar $st$-graph.
\end{proof}

\section{Preliminaries}
\label{app:preliminaries}
%In this section we define $SPQR$-trees and the notions related to plane $st$-graphs.

\subsection{$SPQR$-decomposition}
\label{app:SPQR-trees}
Our algorithm employs a specific version of $SPQR$-decomposition that allows us to describe all $st$-embeddings of a planar $st$-graph.
We follow Di Battista and Tamassia~\cite{BattistaT96}, who were the first to define such $SPQR$-trees, and to prove the properties presented below.

Let $G$ be a planar $st$-graph.
A \emph{cut-vertex} of $G$ is a vertex whose removal disconnects $G$.
A \emph{separation pair} of $G$ is a pair of vertices whose removal disconnects $G$.
A \emph{split pair} of $G$ is either a separation pair or a pair of adjacent vertices.
A \emph{split component} of a split pair $\{u,v\} $ is either an edge $(u,v)$ or a maximal subgraph $C$ of $G$ such that $C$ is a planar $uv$-graph and $\{u,v\}$ is not a split pair of $C$.
A \emph{maximal split pair} $\{u,v\}$ of $G$ is a split pair such that there is no other split pair $\{u',v'\}$ where $\{u,v\}$ is contained in some split component of $\{u',v'\}$.

An \emph{$SPQR$-tree} $T$ for a planar $st$-graph $G$ is a recursive decomposition of $G$ with respect to the split pairs of $G$.
$T$ is a rooted tree whose nodes are of four types: $S$ for \emph{series nodes}, $P$ for \emph{parallel nodes}, $Q$ for \emph{edge nodes}, and $R$ for \emph{rigid nodes}.
Each node $\mu$ of $T$ represents an $st$-graph (a subgraph of $G$) called the \emph{pertinent digraph} of $\mu$ and denoted by $G_{\mu}$.
We use $s_\mu$ and $t_\mu$ to denote the poles of $G_\mu$: $s_\mu$ is the source of $G_\mu$, and $t_\mu$ is the sink of $G_\mu$.
The pertinent digraph of the root node of $T$ is $G$.
Each node $\mu$ of $T$ has an associated directed multigraph $skel(\mu)$ called the \emph{skeleton} of $\mu$.
If $\mu$ is not the root of the tree, then let $\lambda$ be the parent of $\mu$ in $T$.
The node $\mu$ is associated with an edge of the skeleton of $\lambda$, called the \emph{virtual edge} of $\mu$, which connects the poles of $G_{\mu}$ and represents $G_{\mu}$ in $skel(\lambda)$.
The tree $T$ is defined recursively as follows.
\begin{itemize}

\item \emph{Trivial case.} If $G$ consists of a single edge $(s,t)$, then $T$ is simply a $Q$-node $\mu$.
The skeleton $skel(\mu)$ is $G$.

\item \emph{Series case.} If $G$ is a chain of biconnected components $G_1, \dots, G_k$ for some $k \geq 2$ and $c_1, \dots, c_{k-1}$ are the cut-vertices encountered in this order on any path from $s$ to $t$, then the root of $T$ is an $S$-node $\mu$ with children $\mu_1, \dots, \mu_k$.
Let $c_0 = s$ and $c_k = t$.
The skeleton $skel(\mu)$ is the directed path $c_0,\ldots,c_k$.
The pertinent digraph of $\mu_i$ is $G_i$, and edge $(c_{i-1},c_i)$ of $skel(\mu)$ is the virtual edge of $\mu_i$.

\item \emph{Parallel case.} If $\{s,t\}$ is a split pair of $G$ with split components $G_1, \dots, G_k$ for some $k \geq 2$, then the root of $T$ is a $P$-node $\mu$ with children  $\mu_1,\ldots, \mu_k$.
The skeleton $skel(\mu)$ has $k$ parallel edges $(s,t)$: $e_1,\ldots,e_k$.
The pertinent digraph of $\mu_i$ is $G_i$, and edge $e_i$ of $skel(\mu)$ is the virtual edge of $\mu_i$.

\item \emph{Rigid case.} If none of the above applies, let $\{s_1,t_1\}, \dots, \{s_k,t_k\}$ for some $k \geq 2$ be the maximal split pairs of $G$.
For $i=1, \ldots, k$, let $G_i$ be the union of all split components of $\{s_i,t_i\}$.
The root of $T$ is an $R$-node $\mu$ with children $\mu_1, \ldots, \mu_k$.
The skeleton $skel(\mu)$ is obtained from $G$ by replacing each subgraph $G_i$ with an edge $e_i=(s_i,t_i)$.
The pertinent digraph of $\mu_i$ is $G_i$, and edge $e_i$ of $skel(\mu)$ is the virtual edge of $\mu_i$.

\end{itemize}

\begin{figure*}
  \begin{tikzpicture}[>=latex]
    \definecolor{light-gray}{gray}{0.60}
    \definecolor{areafill-gray}{gray}{0.90}
    %\tikzstyle{every node}=[circle,fill=black]
    \begin{scope}[xscale=0.45, yscale=0.45]
      \begin{tiny} {
        \thinmuskip=0.5mu
        \medmuskip=0.5mu plus 0.5mu minus 0.5mu
        \thickmuskip=0.5mu plus 0.5mu minus 0.5mu
        
        \begin{scope}[shift={(0,0)}]
          \draw[fill=white, stroke=black] (0,0) rectangle (10,7);
          \draw[fill=white, stroke=black] (7,6) rectangle (10,7);
          \node[align=center] at (8.5,6.5) {$R$-node};
          \begin{scope}[shift={(0.5,0.5)}]%9x6
            \node[] (s) at (3,0) {$s$};
            \node[] (5) at (0,2) {$5$};
            \node[] (10) at (6,1) {$10$};
            \node[] (13) at (9,2) {$13$};
            \node[] (14) at (3,4) {$14$};
            \node[] (t) at (3,6) {$t$};

            \draw (s) edge[->,dashed] (5);
            \draw (s) edge[->,dashed] (10);
            \draw (s) edge[->,dashed,bend right] (13);
            \draw (5) edge[->,dashed] (14);
            \draw (5) edge[->] (t);
            \draw (10) edge[->] (13);
            \draw (10) edge[->] (14);
            \draw (13) edge[->] (14);
            \draw (13) edge[->] (t);
            \draw (14) edge[->] (t);
          \end{scope}
        \end{scope}

        \draw[thick, ->] (0,3.5)--(-1,3.5);
        \draw[thick, ->] (10,3.5)--(11,3.5);
        \draw[thick, ->] (1.5,0)--(1.5,-1);
        \draw[thick, ->] (8.5,0)--(8.5,-1);
        \draw[thick, ->] (-1.5,-3.5)--(-2.5,-3.5);
        \draw[thick, ->] (11.5,-3.5)--(12.5,-3.5);
        \draw[thick, ->] (1.5,-6)--(1.5,-7);
        \draw[thick, ->] (-1.5,-9.5)--(-2.5,-9.5);
        \draw[thick, ->] (15.5,-6)--(15.5,-7);
        \draw[thick, ->] (12.5,-6)--(11.5,-7);
        \draw[thick, ->] (-5.5,-12)--(-5.5,-13);
        \draw[thick, ->] (-2.5,-12)--(-1.5,-13);

        \begin{scope}[shift={(-1.5,-6)}]
          \draw[fill=white, stroke=black] (0,0) rectangle (6,5);
          \draw[fill=white, stroke=black] (4,4) rectangle (6,5);
          \node[align=center] at (5,4.5) {$P$-node};
          \begin{scope}[shift={(0.5,0.5)}]%9x6
            \node[] (s) at (2,0) {$s$};
            \node[] (5) at (2,4) {$5$};
            \draw (s) edge[->,dashed,bend left] (5);
            \draw (s) edge[->] (5);
            \draw (s) edge[->,dashed,bend right] (5);
          \end{scope}
        \end{scope}

        \begin{scope}[shift={(5.5,-6)}]
          \draw[fill=white, stroke=black] (0,0) rectangle (6,5);
          \draw[fill=white, stroke=black] (4,4) rectangle (6,5);
          \node[align=center] at (5,4.5) {$S$-node};
          \begin{scope}[shift={(0.5,0.5)}]%9x6
            \node[] (s) at (2,0) {$s$};
            \node[] (7) at (2,2) {$7$};
            \node[] (10) at (2,4) {$10$};
            \draw (s) edge[->] (7);
            \draw (7) edge[->,dashed] (10);
          \end{scope}
        \end{scope}

        \begin{scope}[shift={(11,1)}]
          \draw[fill=white, stroke=black] (0,0) rectangle (6,5);
          \draw[fill=white, stroke=black] (4,4) rectangle (6,5);
          \node[align=center] at (5,4.5) {$R$-node};
          \begin{scope}[shift={(0.5,0.5)}]%9x6
            \node[] (s) at (2,0) {$s$};
            \node[] (11) at (1,1.5) {$11$};
            \node[] (12) at (3,2.5) {$12$};
            \node[] (13) at (2,4) {$13$};
            \draw (s) edge[->] (11);
            \draw (s) edge[->] (12);
            \draw (11) edge[->] (12);
            \draw (11) edge[->] (13);
            \draw (12) edge[->] (13);
          \end{scope}
        \end{scope}

        \begin{scope}[shift={(-7,1)}]
          \draw[fill=white, stroke=black] (0,0) rectangle (6,5);
          \draw[fill=white, stroke=black] (4,4) rectangle (6,5);
          \node[align=center] at (5,4.5) {$S$-node};
          \begin{scope}[shift={(0.5,0.5)}]%9x6
            \node[] (5) at (2,0) {$5$};
            \node[] (6) at (2,2) {$6$};
            \node[] (14) at (2,4) {$14$};
            \draw (5) edge[->] (6);
            \draw (6) edge[->] (14);
          \end{scope}
        \end{scope}

        \begin{scope}[shift={(-8.5,-6)}]
          \draw[fill=white, stroke=black] (0,0) rectangle (6,5);
          \draw[fill=white, stroke=black] (4,4) rectangle (6,5);
          \node[align=center] at (5,4.5) {$S$-node};
          \begin{scope}[shift={(0.5,0.5)}]%9x6
            \node[] (s) at (2,0) {$s$};
            \node[] (1) at (2,2) {$1$};
            \node[] (5) at (2,4) {$5$};
            \draw (s) edge[->] (1);
            \draw (1) edge[->] (5);
          \end{scope}
        \end{scope}

        \begin{scope}[shift={(-1.5,-12)}]
          \draw[fill=white, stroke=black] (0,0) rectangle (6,5);
          \draw[fill=white, stroke=black] (4,4) rectangle (6,5);
          \node[align=center] at (5,4.5) {$S$-node};
          \begin{scope}[shift={(0.5,0.5)}]%9x6
            \node[] (s) at (2,0) {$s$};
            \node[] (2) at (2,2) {$2$};
            \node[] (5) at (2,4) {$5$};
            \draw (s) edge[->] (2);
            \draw (2) edge[->,dashed] (5);
          \end{scope}
        \end{scope}

        \begin{scope}[shift={(12.5,-6)}]
          \draw[fill=white, stroke=black] (0,0) rectangle (6,5);
          \draw[fill=white, stroke=black] (4,4) rectangle (6,5);
          \node[align=center] at (5,4.5) {$P$-node};
          \begin{scope}[shift={(0.5,0.5)}]%9x6
            \node[] (7) at (2,0) {$7$};
            \node[] (10) at (2,4) {$10$};
            \draw (7) edge[->,dashed,bend left] (10);
            \draw (7) edge[->,dashed,bend right] (10);
          \end{scope}
        \end{scope}

        \begin{scope}[shift={(-8.5,-12)}]
          \draw[fill=white, stroke=black] (0,0) rectangle (6,5);
          \draw[fill=white, stroke=black] (4,4) rectangle (6,5);
          \node[align=center] at (5,4.5) {$P$-node};
          \begin{scope}[shift={(0.5,0.5)}]%9x6
            \node[] (2) at (2,0) {$2$};
            \node[] (5) at (2,4) {$5$};
            \draw (2) edge[->,dashed,bend left] (5);
            \draw (2) edge[->,dashed,bend right] (5);
          \end{scope}
        \end{scope}

        \begin{scope}[shift={(-8.5,-18)}]
          \draw[fill=white, stroke=black] (0,0) rectangle (6,5);
          \draw[fill=white, stroke=black] (4,4) rectangle (6,5);
          \node[align=center] at (5,4.5) {$S$-node};
          \begin{scope}[shift={(0.5,0.5)}]%9x6
            \node[] (2) at (2,0) {$2$};
            \node[] (3) at (2,2) {$3$};
            \node[] (5) at (2,4) {$5$};
            \draw (2) edge[->] (3);
            \draw (3) edge[->] (5);
          \end{scope}
        \end{scope}

        \begin{scope}[shift={(-1.5,-18)}]
          \draw[fill=white, stroke=black] (0,0) rectangle (6,5);
          \draw[fill=white, stroke=black] (4,4) rectangle (6,5);
          \node[align=center] at (5,4.5) {$S$-node};
          \begin{scope}[shift={(0.5,0.5)}]%9x6
            \node[] (2) at (2,0) {$2$};
            \node[] (4) at (2,2) {$4$};
            \node[] (5) at (2,4) {$5$};
            \draw (2) edge[->] (4);
            \draw (4) edge[->] (5);
          \end{scope}
        \end{scope}

        \begin{scope}[shift={(5.5,-12)}]
          \draw[fill=white, stroke=black] (0,0) rectangle (6,5);
          \draw[fill=white, stroke=black] (4,4) rectangle (6,5);
          \node[align=center] at (5,4.5) {$S$-node};
          \begin{scope}[shift={(0.5,0.5)}]%9x6
            \node[] (7) at (2,0) {$7$};
            \node[] (8) at (2,2) {$8$};
            \node[] (10) at (2,4) {$10$};
            \draw (7) edge[->] (8);
            \draw (8) edge[->] (10);
          \end{scope}
        \end{scope}

        \begin{scope}[shift={(12.5,-12)}]
          \draw[fill=white, stroke=black] (0,0) rectangle (6,5);
          \draw[fill=white, stroke=black] (4,4) rectangle (6,5);
          \node[align=center] at (5,4.5) {$S$-node};
          \begin{scope}[shift={(0.5,0.5)}]%9x6
            \node[] (7) at (2,0) {$7$};
            \node[] (9) at (2,2) {$9$};
            \node[] (10) at (2,4) {$10$};
            \draw (7) edge[->] (9);
            \draw (9) edge[->] (10);
          \end{scope}
        \end{scope}

        }
      \end{tiny}
    \end{scope}
  \end{tikzpicture}
  \caption{The $SPQR$-tree for the graph in Figure~\ref{fig:large}.
    The $Q$-nodes (leaves of the tree) have been omitted for clarity.
    For each $S$-, $P$-, and $R$-node, the skeleton is given such that each solid edge corresponds to a $Q$-node child and each dashed edge corresponds to a $S$-, $P$-, or $R$-node child.
  }
  \label{fig:large_spqr}
\end{figure*}
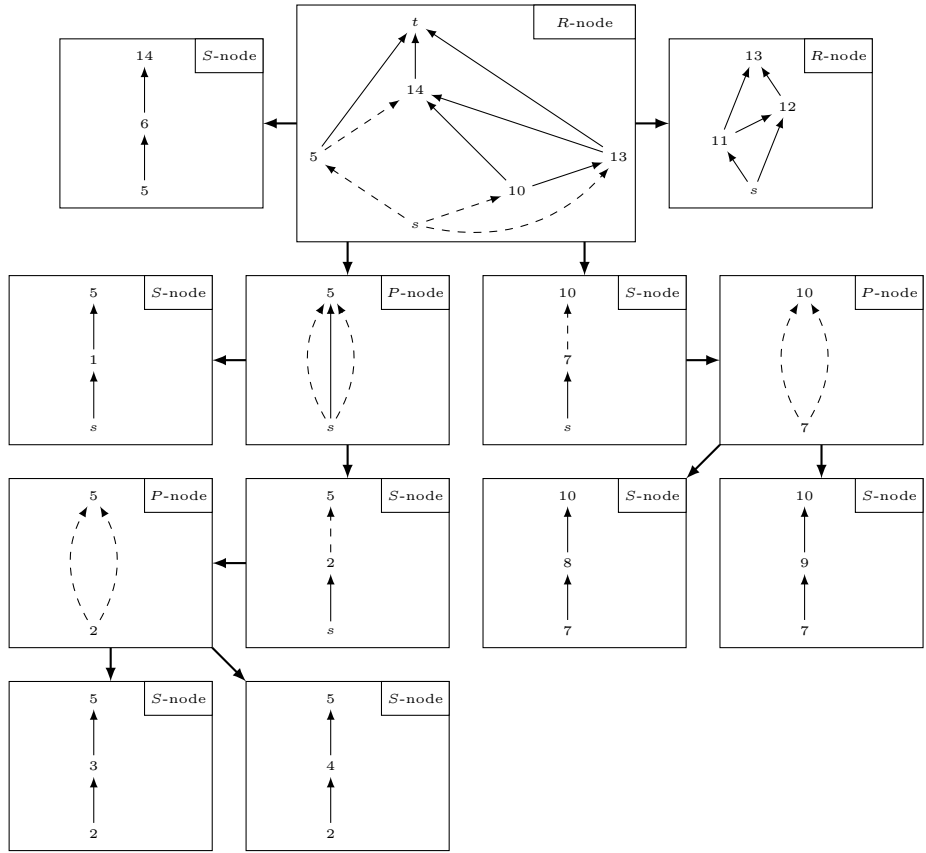

Note also that there is no additional edge between the poles of the skeleton of a series, parallel or rigid node -- this is the only difference in the $SPQR$-tree definition given above and the one given in~\cite{BattistaT96}.
In particular, our definition ensures that we have a one-to-one correspondence between the edges of $skel(\mu)$ and the children of $\mu$.
See Figure~\ref{fig:large_spqr} for an example of an $SPQR$-decomposition of the planar $st$-graph presented in Figure \ref{fig:large}.

Observe that the skeleton of a rigid node has only two $st$-embeddings, one being the flip of the other around the poles of the node.
The skeleton of a parallel node with $k$ children has $k!$ $st$-embeddings, one for every permutation of the edges of $skel(\mu)$.
The skeleton of a series node or a edge node has only one $st$-embedding.

There is a correspondence between $st$-embeddings of an $st$-graph $G$ and $st$-embeddings of the skeletons of $P$-nodes and $R$-nodes in the $SPQR$-tree $T$ for $G$.
Having selected an $st$-embedding of the skeleton of all $P$-nodes and all $R$-nodes, we can construct an embedding of $G$ as follows.
Let $t$ be the root of $T$.
We replace every virtual edge $(u,v)$ in the embedding of $skel(t)$ with the recursively defined embedding of the pertinent digraph of a child of $t$ associated with the edge $(u,v)$.
On the other hand, any $st$-embedding of $G$ determines:
\begin{itemize}
 \item one of the two possible flips of the skeleton of every $R$-node in $T$;
 \item a permutation of the edges in the skeleton of every $P$-node.
\end{itemize}

Di Battista and Tamassia~\cite{BattistaT96} showed that the $SPQR$-tree $T$ for a planar $st$-graph with $n$ vertices has $\Oh{n}$ nodes, that the total number of edges of all skeletons is $\Oh{n}$, and that $T$ can be computed in linear time.

\subsection{Plane $st$-graphs}
\label{app:plane-st-graphs}
%In this section we define the notions realted to plane $st$-graphs.

Tamassia and Tollis~\cite{TamassiaT86} showed that the following properties are satisfied for any plane $st$-graph:
\begin{enumerate}
  \item \label{app_prop:st_inner_face} For every inner face $f$, the boundary of $f$ consists of two directed paths with a common origin and a common destination.
  \item \label{app_prop:st_outer_face} The boundary of the outer face consists of two directed paths, with a common origin $s$ and a common destination $t$.
  \item \label{app_prop:st_bipolar} For every inner vertex $v$, its incoming (outgoing) edges are consecutive around $v$.
\end{enumerate}

%To illustrate the above properties, observe in Figure~\ref{fig:large} two paths on the boundary of the face $(s,2,3,5,1)$, and the alignment of incoming and outgoing edges around vertex $5$.

Let $G$ be a plane $st$-graph.
We introduce two special objects associated with the outer face of $G$: the \emph{left outer face} $s^*$ and the \emph{right outer face} $t^*$.
Let $e = (u,v)$ be an edge of $G$.
The \emph{left face} (\emph{right face}) of $e$ is the face of $G$ that is to the left (right) of $e$ when we traverse $e$ from $u$ to $v$.
If the outer face of $G$ is to the left (right) of $e$ then we say that the left face (right face) of $e$ is $s^*$ ($t^*$).

Using property~\eqref{app_prop:st_inner_face} we can define the \emph{left path} and the \emph{right path} for each inner face of $G$ as follows.
If $f$ is an inner face of $G$ then the left path (right path) of $f$ consists of edges from the boundary of $f$ for which $f$ is the right face (left face).

Using property~\eqref{app_prop:st_outer_face} we can define the left path for $t^*$ and the right path for $s^*$ as follows.
The right path of $s^*$ consists of edges from the boundary of the outer face that have the outer face on their left side.
The left path of $t^*$ consists of edges from the boundary of the outer face that have the outer face on their right side.
The left path for $s^*$ and the right path for $t^*$ are not defined.

Using property~\eqref{app_prop:st_bipolar} we can define the \emph{left face} and the \emph{right face} for each vertex of $G$ as follows.
The left face (right face) of an inner vertex $v$ is the unique face $f$ incident to $v$ such that there are two edges $e_1$ and $e_2$ on the right path (left path) of $f$, where $e_1$ is an incoming edge for $v$ and $e_2$ is an outgoing edge for $v$.
If the left face (right face) of $v$ is the outer face of $G$, we say that the left face (right face) of $u$ is $s^*$ ($t^*$).
We also say that $s^*$ ($t^*$) is the left face (right face) for $s$ and $t$.

Let $G$ be a plane $st$-graph.
Let $F$ be the set of inner faces of $G$ together with $s^*$ and $t^*$.
The \emph{dual} of $G$ is the directed graph $G^*$ with vertex set $F$ and edge set consisting of all pairs $(f,g)$ such that there exists an edge $e$ of $G$ with $f$ being the left face of $e$ and $g$ being the right face of $e$.
Di Battista and Tamassia~\cite{BattistaT88} showed that $G^*$ is a planar $s^*t^*$-graph.

Let $G$ be a plane $st$-graph and let $G^*$ be the dual of $G$.
For two faces $f$ and $g$ in $V(G^*)$ we say that \emph{$f$ is to the left of $g$}, and that \emph{$g$ is to the right of $f$}, if there is a directed path from $f$ to $g$ in $G^*$.

\section{Rectangular bar visibility representations of $st$-graphs}
\label{app:rectangular}
%At the begining of the section we present the proofs of the lemmas showing the correctness of the $\Oh{n^2}$ time algorithm for the
%Rectangular Bar Visibility Representation Extension problem.
%In the second part we prove Theorem \ref{th:main-plane-st}, i.e., we show how to refine our algorithm to work in $\Oh{n \log{n}}$ time.

\subsection{Structural properties}
\label{app:rectangular_properties}
%The goal fo this section is to prove the Tiling Lemmas presented in Subsection~\ref{sec:rectangular_properties} of Section~\ref{sec:rectangular}.

\hackLcounter{lem:tiling-lemma-Q-node}
\begin{lemma}[Q-Tiling Lemma]
Let $\mu$ be a $Q$-node in $T$ that corresponds to an edge $(u,v)$ of $G$.
For any rectangular bar visibility representation $\psi$ of $G$ we have:
\begin{enumerate}
\item $B(\mu)$ is a union of pairwise disjoint rectangles spanned between $\psi(u)$ and $\psi(v)$.
\item If $B(\mu)$ is not a single rectangle, then the parent $\lambda$ of $\mu$ in $T$ is a $P$-node, and $u$, $v$ are the poles of the pertinent digraph $G_{\lambda}$.
\end{enumerate}
\end{lemma}
\unhackLcounter
\begin{proof}%[Proof of Q-Tiling Lemma (Lemma~\ref{lem:tiling-lemma-Q-node})]
The first assertion is obvious.
Suppose that $B(\mu)$ is a union of at least $2$ rectangles.
Let $R_1$ and $R_2$ be the two left-most rectangles of $B(\mu)$.
Consider the rectangle $S$ spanned between $\psi(u)$ and $\psi(v)$ and between the right side of $R_1$ and left side of $R_2$.
There are some bars in $\psi(G)$ that are contained in $S$.
The vertices corresponding to these bars together with $u$ and $v$ form a planar $uv$-graph.
Hence, the split pair $\{u,v\}$ has at least two split components: the edge $(u,v)$ and at least one other component.
Thus, $\lambda$ is a $P$-node with poles $u$ and $v$.
\end{proof}

In Figure~\ref{fig:large_tiling_ps} observe that the set $R(s,5)$ is a union of two rectangles.
Recall the $SPQR$-decomposition presented in Figure~\ref{fig:large_spqr} and that the $Q$-node corresponding to the edge $(s,5)$ is a child of a $P$-node.

\begin{figure*}
  \begin{tikzpicture}[>=latex]
    \definecolor{light-gray}{gray}{0.92}
    \definecolor{dark-gray}{gray}{0.78}
    \definecolor{areafill-gray}{gray}{0.90}
    %\tikzstyle{every node}=[circle,fill=black]
    \begin{scope}[xscale=0.5, yscale=0.5]
      \begin{tiny} {
        \thinmuskip=0.5mu
        \medmuskip=0.5mu plus 0.5mu minus 0.5mu
        \thickmuskip=0.5mu plus 0.5mu minus 0.5mu
        
        \begin{scope}[shift={(0,0)}]
          \node[] (s) at (5,-0.25) {$s$};
          \node[] (1) at (0,2.5) {$1$};
          \node[] (2) at (2,1.5) {$2$};
          \node[] (3) at (2,4.0) {$3$};
          \node[] (4) at (3.5,2.75) {$4$};
          \node[] (5) at (3.5,6.5) {$5$};
          \node[] (6) at (4.5,8.5) {$6$};
          \node[] (7) at (6.5,2) {$7$};
          \node[] (8) at (5,4) {$8$};
          \node[] (9) at (7.5,4) {$9$};
          \node[] (10) at (6,6) {$10$};
          \node[] (11) at (7.5,1) {$11$};
          \node[] (12) at (10,2) {$12$};
          \node[] (13) at (10,8) {$13$};
          \node[] (14) at (6,9.5) {$14$};
          \node[] (t) at (6,11.5) {$t$};

          \draw (s) edge[->, bend left=20, very thick] (1);
          \draw (s) edge[->, very thick] (2);
          \draw (s) edge[->, very thick] (5);
          \draw (s) edge[->, very thick, dashed] (7);
          \draw (s) edge[->] (11);
          \draw (s) edge[->, bend right=20] (12);
          \draw (1) edge[->, bend left, very thick] (5);
          \draw (2) edge[->, very thick] (3);
          \draw (2) edge[->, very thick] (4);
          \draw (3) edge[->, very thick] (5);
          \draw (4) edge[->, very thick] (5);
          \draw (5) edge[->] (6);
          \draw (5) edge[->, bend left] (t);
          \draw (6) edge[->] (14);
          \draw (7) edge[->, very thick, dashed] (8);
          \draw (7) edge[->, very thick, dashed] (9);
          \draw (8) edge[->, very thick, dashed] (10);
          \draw (9) edge[->, very thick, dashed] (10);
          \draw (10) edge[->] (13);
          \draw (10) edge[->] (14);
          \draw (11) edge[->] (12);
          \draw (11) edge[->] (13);
          \draw (12) edge[->] (13);
          \draw (13) edge[->] (14);
          \draw (13) edge[->] (t);
          \draw (14) edge[->] (t);
        \end{scope}

        \begin{scope}[shift={(13,0.75)}]
          \draw[fill=light-gray, color=light-gray] (0,0) rectangle (2,6);
          \draw[fill=dark-gray, color=dark-gray] (2,0) rectangle (3,6);
          \draw[fill=light-gray, color=light-gray] (3,0) rectangle (5,6);
          \draw[fill=dark-gray, color=dark-gray] (5,0) rectangle (6,6);

          \fill[pattern=horizontal lines light gray] (6,0) rectangle (8,1);
          \fill[pattern=horizontal lines gray] (6,1) rectangle (8,7);

          \node at (5.5,10.4) {$\psi(t)$};
          \draw[thick, (-)] (0,10)--(11,10);
          \node at (5.5,9.4) {$\psi(14)$};
          \draw[thick, (-)] (1,9)--(10,9);
          \node at (9,8.4) {$\psi(13)$};
          \draw[thick, (-)] (7,8)--(11,8);
          \node at (3.5,8.4) {$\psi(6)$};
          \draw[thick, (-)] (1,8)--(6,8);
          \node at (7,7.4) {$\psi(10)$};
          \draw[thick, (-)] (6,7)--(8,7);
          \node at (3,6.4) {$\psi(5)$};
          \draw[thick, (-)] (0,6)--(6,6);
          \node at (10,5.4) {$\psi(12)$};
          \draw[thick, (-)] (9,5)--(11,5);
          \node at (7.6,4.5) {$\psi(9)$};
          \draw[thick, (-)] (7,4)--(8,4);
          \node at (6.4,4.5) {$\psi(8)$};
          \draw[thick, (-)] (6,4)--(7,4);
          \node at (4.6,4.5) {$\psi(4)$};
          \draw[thick, (-)] (4,4)--(5,4);
          \node at (3.4,4.5) {$\psi(3)$};
          \draw[thick, (-)] (3,4)--(4,4);
          \node at (9,3.4) {$\psi(11)$};
          \draw[thick, (-)] (8,3)--(10,3);
          \node at (1,2.4) {$\psi(1)$};
          \draw[thick, (-)] (0,2)--(2,2);
          \node at (7,1.4) {$\psi(7)$};
          \draw[thick, (-)] (6,1)--(8,1);
          \node at (4,1.4) {$\psi(2)$};
          \draw[thick, (-)] (3,1)--(5,1);
          \node at (5.5,0.4) {$\psi(s)$};
          \draw[thick, (-)] (0,0)--(11,0);
        \end{scope}
        }
      \end{tiny}
    \end{scope}
  \end{tikzpicture}
  \caption{
    The graph $G$,
      the pertinent digraph of the $P$-node $\mu_1$ (solid thick edges), 
      the pertinent digraph of the $S$-node $\mu_2$ (dashed thick edges), 
      the representation $\psi(G)$,
      the tiling of $\mu_1$ in $\psi(G)$ (solid fill), and
      the tiling of $\mu_2$ in $\psi(G)$ (patterned fill).
  }
  \label{fig:large_tiling_ps}
\end{figure*}
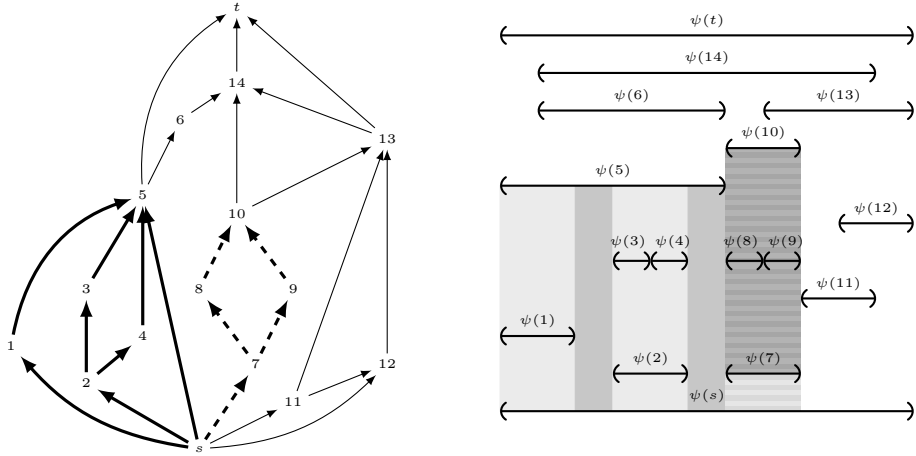

\hackLcounter{lem:basic-tiling-lemma}
\begin{lemma}[Basic Tiling Lemma]
Let $\mu$ be an inner node in $T$ with children $\mu_1,\ldots,\mu_k$, $k \geq 2$.
For any rectangular bar visibility representation $\psi$ of $G$ we have:
\begin{enumerate}
 \item $\psi(v) \subseteq B(\mu)$ for every inner vertex $v$ of $G_\mu$.
 \item $B(\mu)$ is a rectangle that is spanned between $\psi(s_\mu)$  and $\psi(t_\mu)$.
 \item The sets $B(\mu_1), \ldots, B(\mu_k)$ tile the rectangle $B(\mu)$, \ie{},
 $B(\mu_1), \ldots, B(\mu_k)$ cover $B(\mu)$ and the interiors of $B(\mu_1), \ldots, B(\mu_k)$ are pairwise disjoint.
\end{enumerate}
\end{lemma}
\unhackLcounter
\begin{proof}%[Proof of the Basic Tiling Lemma (Lemma~\ref{lem:basic-tiling-lemma})]
  Observe that for an inner vertex $v$ of $G_\mu$, any edge of $G$ incident to $v$ is an edge of $G_\mu$.
  The closures of the sets of visibility rectangles corresponding to all edges incident to $v$ cover $\psi(v)$ and Property~\eqref{lem:basic-tiling-lemma-containment} follows.

  To prove~\eqref{lem:basic-tiling-lemma-rectangle} note that for every inner vertex $v$ of $G_\mu$, the set
  $$S_\mu(v) = X(\psi(v)) \times [y(\psi(s_\mu)), y(\psi(t_\mu))]$$ is a rectangle
  that is spanned between $\psi(s_\mu)$ and $\psi(t_\mu)$ and it is internally disjoint from $\psi(w)$ for any vertex $w$ not in $V(G_\mu)$.
  Otherwise, there would be a visibility gap that would correspond to an edge between an inner vertex of $G_\mu$ and a vertex in $V(G) \setminus V(G_\mu)$.

  If $(s_\mu, t_\mu)$ is not an edge of $G_\mu$, then
  $$B(\mu) = \bigcup \{S_\mu(v): v \text{ is an inner vertex of $G_\mu$}\};$$
  otherwise
  $$B(\mu) = \bigcup \{S_\mu(v): v \text{ is an inner vertex of $G_\mu$}\} \cup R(\psi(s_\mu), \psi(t_\mu)).$$
  In both cases $B(\mu)$ is a rectangle spanned between $\psi(s_\mu)$ and $\psi(t_\mu)$.

  Property~\eqref{lem:basic-tiling-lemma-tiling} follows immediately from the fact that the edges of $G_{\mu_1}, \ldots, G_{\mu_k}$ form a partition of the edges of $G_\mu$.
\end{proof}

Figures~\ref{fig:large_tiling_ps} and~\ref{fig:large_tiling_r} give a graphical presentation of the Basing Tiling Lemma.
Figure ~\ref{fig:large_tiling_ps} shows a possible tiling of a $P$-node and an $S$-node.

\hackLcounter{lem:tiling-lemma-S-node}
\begin{lemma}[S-Tiling Lemma]
Let $\mu$ be an $S$-node.
Let $c_1,\ldots,c_{k-1}$ be the cut-vertices of $G_\mu$ encountered in this order on a path from $s_\mu$ to $t_\mu$.
Let $c_0 = s_\mu$, and $c_k = t_\mu$.
For any rectangular bar visibility representation $\psi$ of $G$, for every $i=1,\ldots,k-1$, we have $X(\psi(c_i)) = X(B(\mu))$.
For every $i=1,\ldots,k$, $B(\mu_i)$ is spanned between $\psi(c_{i-1})$ and $\psi(c_i)$ and $X(B(\mu_i)) = X(B(\mu))$.
\end{lemma}
\unhackLcounter
\begin{proof}%[Proof of the S-Tiling Lemma (Lemma~\ref{lem:tiling-lemma-S-node})]
  Suppose to the contrary, that the bar assigned to cut-vertex $c_i$ is the first one that does not span the whole interval $X(B(\mu))$.
  This creates a gap of visibility between a vertex in the $i$-th biconnected component and a vertex in one of the later components.
  This contradicts $c_i$ being a cut-vertex.
\end{proof}

\hackLcounter{lem:face-condition}
\begin{lemma}[Face Condition]\
\begin{enumerate}
\item
Let $f$ be a face in $V(\mathcal{E}^*)$ different than $t^*$, and let $v_0,v_1,\ldots,v_n$ be the right path of $f$.
There is a vertical line $L_r(f)$ that contains the left endpoints of $\psi(v_1), \ldots, \psi(v_{n-1})$ and the left sides of $B_\psi(v_0,v_1), \ldots, B_\psi(v_{n-1},v_n)$.
\item
Let $f$ be a face in $V(\mathcal{E}^*)$ different than $s^*$, and let $u_0,u_1,\ldots,u_m$ be the left path of $f$.
There is a vertical line $L_l(f)$ that contains the right endpoints of $\psi(u_1), \ldots, \psi(u_{m-1})$ and the right sides of $B_\psi(u_0,u_1), \ldots, B_\psi(u_{m-1},u_m)$.
\item
If $f$ is an inner face of $\mathcal{E}$ then $L_l(f)=L_r(f)$.
\end{enumerate}
\end{lemma}
\unhackLcounter
\begin{proof}%[Proof of the Face Condition Lemma (Lemma~\ref{lem:face-condition})]
To prove~\eqref{lem:fc_right_face} we first show that for every $i=1,\ldots,n-1$, we have
$$l(\psi(v_i)) = l(B_\psi(v_{i-1}, v_{i}))\text{.}$$
By the Basic Tiling Lemma, $B_\psi(v_{i-1}, v_{i})$ is a rectangle spanned between $\psi(v_{i-1})$ and $\psi(v_{i})$.
It follows that $l(\psi(v_i)) \leq l(B_\psi(v_{i-1}, v_{i}))$.
Suppose that $l(\psi(v_i)) < l(B_\psi(v_{i-1}, v_{i}))$.
By the Basic Tiling Lemma again, there is a child $\lambda$ of $\mu$ such that the rectangle $B_\psi(\lambda)$ has its top right corner located at the intersection of the left side of $B_\psi(v_{i-1}, v_{i})$ and $\psi(v_i)$.
Clearly, $\lambda$ corresponds to an edge of $skel(\mu)$ that is in the embedding $\mathcal{E}$ between $(v_{i-1},v_{i})$ and $(v_i,v_{i+1})$ in the clockwise order around $v_i$.
However, there is no such edge in $\mathcal{E}$, a contradiction.
Similarly, for every $i=1,\ldots,n-1$, we have
$$l(\psi(v_i)) = l(B_\psi(v_{i}, v_{i+1}))\text{.}$$
It follows that the left sides of the bounding boxes $B_\psi(v_0,v_1), \ldots, B_\psi(v_{n-1},v_n)$ and the left endpoints of $\psi(v_1),\ldots,\psi(v_{n-1})$
are aligned to the same vertical line $L_r(f)$.

The proof of~\eqref{lem:fc_left_face} is analogous.
Property~\eqref{lem:fc_face} is an immediate consequence of the Basic Tiling Lemma.
\end{proof}

Figure ~\ref{fig:large_tiling_r} shows a possible tiling of an $R$-node and illustrates the Face Condition Lemma.

\begin{figure*}
  \begin{tikzpicture}[>=latex]
    \definecolor{gray-1}{gray}{0.95}
    \definecolor{gray-2}{gray}{0.80}
    \definecolor{gray-3}{gray}{0.65}
    \definecolor{gray-4}{gray}{0.50}
    %\tikzstyle{every node}=[circle,fill=black]
    \begin{scope}[xscale=0.43, yscale=0.5]
      \begin{tiny} {
        \thinmuskip=0.5mu
        \medmuskip=0.5mu plus 0.5mu minus 0.5mu
        \thickmuskip=0.5mu plus 0.5mu minus 0.5mu
        
        \begin{scope}[shift={(18,0.5)}]
          \begin{scope}[scale=0.7]
            \node[] (ss) at (1.5,10) {$s^*$};
            \node[] (ts) at (13.5,10) {$t^*$};
            \node[] (s) at (7,0) {$s$};
            \node[] (5) at (3,8) {$5$};
            \node[] (10) at (7,8) {$10$};
            \node[] (13) at (12,12) {$13$};
            \node[] (14) at (7,16) {$14$};
            \node[] (t) at (7,20) {$t$};

            \node[] at (4.5,14.5) {$f_1$};
            \node[] at (5,8) {$f_2$};
            \node[] at (9,12) {$f_3$};
            \node[] at (8.5,6.5) {$f_4$};
            \node[] at (9,16.5) {$f_5$};

            \draw (s) edge[->] (5);
            \draw (s) edge[->] (10);
            \draw (s) edge[->] (13);
            \draw (5) edge[->, bend left=30] (t);
            \draw (5) edge[->] (14);
            \draw (10) edge[->] (13);
            \draw (10) edge[->] (14);
            \draw (13) edge[->] (14);
            \draw (13) edge[->, bend right=30] (t);
            \draw (14) edge[->] (t);
          \end{scope}
        \end{scope}

        \begin{scope}[shift={(0,0)}]
          \begin{scope}[xscale=1.5, yscale=1.5]
            \draw[fill=gray-1, color=gray-1] (0,0) rectangle (6,6);
            \draw[fill=gray-2, color=gray-2] (6,0) rectangle (8,7);
            \draw[fill=gray-4, color=gray-4] (0,6) rectangle (1,10);
            \draw[fill=gray-4, color=gray-4] (6,7) rectangle (7,9);
            \draw[fill=gray-1, color=gray-1] (7,7) rectangle (8,8);
            \draw[fill=gray-4, color=gray-4] (8,0) rectangle (11,8);
            \draw[fill=gray-3, color=gray-3] (7,8) rectangle (10,9);
            \draw[fill=gray-3, color=gray-3] (1,6) rectangle (6,9);
            \draw[fill=gray-1, color=gray-1] (1,9) rectangle (10,10);
            \draw[fill=gray-2, color=gray-2] (10,8) rectangle (11,10);

            \draw[very thick, densely dashed, -] (0,0)--(0,10);
            \node[rectangle,fill=white,draw] at (0,4) {$\chi(s^*)$};
            \draw[very thick, densely dashed, -] (1,6)--(1,10);
            \node[rectangle,fill=white,draw] at (1,7) {$\chi(f_1)$};
            \draw[very thick, densely dashed, -] (6,0)--(6,9);
            \node[rectangle,fill=white,draw] at (6,2.5) {$\chi(f_2)$};
            \draw[very thick, densely dashed, -] (7,7)--(7,9);
            \node[rectangle,fill=white,draw] at (7,7.5) {$\chi(f_3)$};
            \draw[very thick, densely dashed, -] (8,0)--(8,8);
            \node[rectangle,fill=white,draw] at (8,5.5) {$\chi(f_4)$};
            \draw[very thick, densely dashed, -] (10,8)--(10,10);
            \node[rectangle,fill=white,draw] at (10,9.5) {$\chi(f_5)$};
            \draw[very thick, densely dashed, -] (11,0)--(11,10);
            \node[rectangle,fill=white,draw] at (11,4) {$\chi(t^*)$};

            \node at (5.5,10.27) {$\psi(t)$};
            \draw[thick, (-)] (0,10)--(11,10);
            \node at (5.5,9.27) {$\psi(14)$};
            \draw[thick, (-)] (1,9)--(10,9);
            \node at (9,8.27) {$\psi(13)$};
            \draw[thick, (-)] (7,8)--(11,8);
            %\node at (3.5,8.27) {$\psi(6)$};
            \draw[dotted, (-)] (1,8)--(6,8);
            \node at (7,6.73) {$\psi(10)$};
            \draw[thick, (-)] (6,7)--(8,7);
            \node at (3,6.27) {$\psi(5)$};
            \draw[thick, (-)] (0,6)--(6,6);
            %\node at (10,5.27) {$\psi(12)$};
            \draw[dotted, (-)] (9,5)--(11,5);
            %\node at (7.5,4.27) {$\psi(9)$};
            \draw[dotted, (-)] (7,4)--(8,4);
            %\node at (6.5,4.27) {$\psi(8)$};
            \draw[dotted, (-)] (6,4)--(7,4);
            %\node at (4.5,4.27) {$\psi(4)$};
            \draw[dotted, (-)] (4,4)--(5,4);
            %\node at (3.5,4.27) {$\psi(3)$};
            \draw[dotted, (-)] (3,4)--(4,4);
            %\node at (9,3.27) {$\psi(11)$};
            \draw[dotted, (-)] (8,3)--(10,3);
            %\node at (1,2.27) {$\psi(1)$};
            \draw[dotted, (-)] (0,2)--(2,2);
            %\node at (7,1.27) {$\psi(7)$};
            \draw[dotted, (-)] (6,1)--(8,1);
            %\node at (4,1.27) {$\psi(2)$};
            \draw[dotted, (-)] (3,1)--(5,1);
            \node at (5.5,0.27) {$\psi(s)$};
            \draw[thick, (-)] (0,0)--(11,0);
          \end{scope}
        \end{scope}
        }
      \end{tiny}
    \end{scope}
  \end{tikzpicture}
  \caption{The tiling of the $R$-node $\mu$, the embedding $\mathcal{E}$ of $skel(\mu)$, splitting lines (dashed) for faces of $\mathcal{E}$, and the $st$-valuation $\chi$ of $\mathcal{E}^*$.
  }
  \label{fig:large_tiling_r}
\end{figure*}
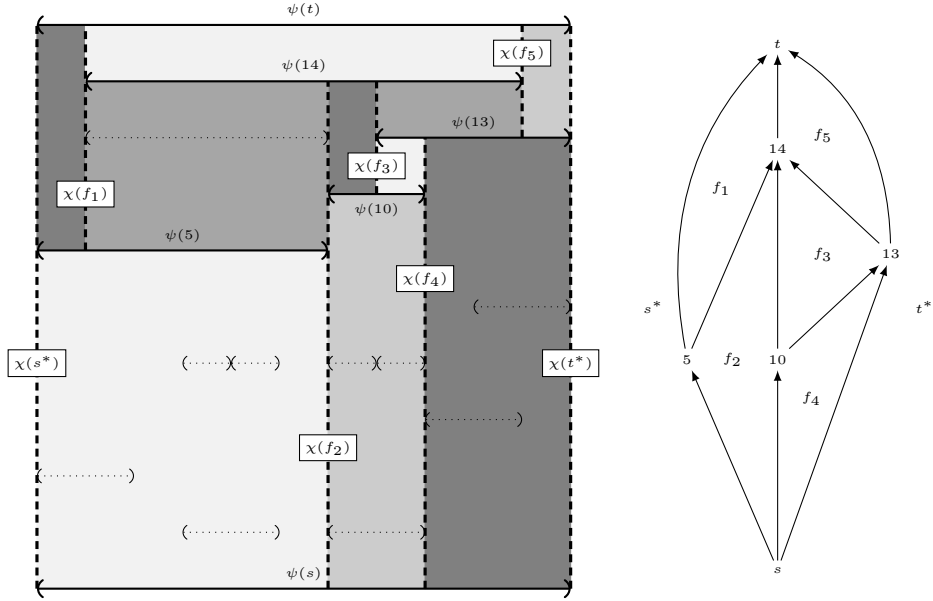

\subsection{Algorithm}
\label{app:rectangular_algorithm_slow}
%In this section we prove lemmas from Subsection~\ref{sec:rectangular_algorithm_slow} of Section~\ref{sec:rectangular}.

\hackLcounter{lem:y-st-valuation}
\begin{lemma}
Let $\psi$ be a rectangular bar visibility representation of $G$ that extends $\psi'$.
\begin{enumerate}
  \item The function $y_{\psi}$ is an $st$-valuation of $G$ that extends $y_{\psi'}$,
  \item If $y$ is an $st$-valuation of $G$ that extends $y_{\psi'}$, then a function $\phi$ that sends every vertex $v$ of $G$ into a bar so that $y_\phi(v) = y(v)$, $l_\phi(v) = l_{\psi}(v)$, $r_\phi(v) = r_{\psi}(v)$ is also a rectangular bar visibility representation of $G$ that extends $\psi'$.
\end{enumerate}
\end{lemma}
\unhackLcounter
\begin{proof}%[Proof of Lemma~\ref{lem:y-st-valuation}]
  The function $y_{\psi}$ extends $y_{\psi'}$, because $\psi$ extends $\psi'$.
  It is an $st$-valuation of $G$ because for an edge $(u,v)$ of $G$, the bar of $u$ is below the bar of $v$.

  For the proof of~\eqref{claim:st-valuation-any}, observe that for each vertex $u$ of $G$ we have $X(\phi(u))=X(\psi(u))$.
  We claim that for any two vertices $u$, $v$ of $G$ such that the interior of $X(\psi(u)) \cap X(\psi(v))$ is non-empty we have that  $y_\psi(u) < y_\psi(v)$ iff $y(u) < y(v)$.
  For the proof of this claim, let $u$ and $v$ be vertices of $G$ such that the interior of $X(\psi(u)) \cap X(\psi(v))$ is non-empty.
  From the fact that $\psi(V(G))$ is a collection of pairwise disjoint bars, it follows that $y_\psi(u) \neq y_{\psi}(v)$.
  Without loss of generality assume that $y_\psi(u) < y_{\psi}(v)$.
  The non-empty interior of $X(\psi(u)) \cap X(\psi(v))$ means that there is a path from $u$ to $v$ in $G$.
  Hence $y(u) < y(v)$ as $y$ is an $st$-valuation of $G$.

  As a consequence we have that $(x_1,x_2) \times (y_\psi(u),y_\psi(v))$ is a visibility gap between bars $\psi(u)$ and $\psi(v)$ in representation $\psi$ iff $(x_1,x_2) \times (y(u),y(v))$ is a visibility gap between $\phi(u)$ and $\phi(v)$ and $\phi$ is a rectangular bar visibility representation of $G$.
\end{proof}

\hackLcounter{lem:stretching}
\begin{lemma}[Stretching Lemma]
Let $\mu$ be an inner node whose core is non-empty.
 If $\mu$ has an LL-representation, then $\mu$ has an $[x,x']$-representation for any $x < l(C(\mu))$ and any $x' > r(C(\mu))$.
 If $\mu$ has an LF-representation, then $\mu$ has an $[x,x']$-representation for any $x < l(C(\mu))$ and $x' = r(C(\mu))$.
 If $\mu$ has an FL-representation, then $\mu$ has an $[x,x']$-representation for $x = l(C(\mu))$ and any $x' > r(C(\mu))$.
\end{lemma}
\unhackLcounter
\begin{proof}%[Proof of the Stretching Lemma (Lemma~\ref{lem:stretching})]
  Let $x_l = l(C(\mu))$.
  Suppose that $\phi$ is some left-loose $[x_1,x']$-representation of $\mu$ with $x_1 < x_l$.
  For any $x_2 < x_l$ we can obtain an $[x_2,x']$-representation of $\mu$ by appropriately stretching the part of the drawing of $\phi$ that is to the left of $x_l$.
  If the representation is right-loose, then we can arbitrarily stretch the part of the drawing that is to the right of $r(C(\mu))$.
\end{proof}

\subsection{Pseudocode}
\label{app:pseudocode}

Section~\ref{sec:rectangular_algorithm_slow} presents an algorithm that tests if there exists a rectangular bar visibility representation that extends a partial representation of an $st$-graph $G$.
The algorithm divides into a few cases for different types of nodes of the $SPQR$-decomposition of $G$.
Algorithm~\ref{alg:s-node} presents pseudocode of the procedure for an $S$-node.
Algorithm~\ref{alg:p-node} presents pseudocode of the procedure for a $P$-node.

\begin{algorithm}
  \caption{Algorithm for series node}\label{alg:s-node}
  \begin{algorithmic}[1]
    \ForEach{cut-vertex $c$ in $G_\mu$}
      \If{$c \in V'(\mu)$ {\bf and} $X(\psi'(c)) \neq [x,x']$}
        \State {\bf return} {\it False}
        \Comment{fixed cut-vertex does not span $[x,x']$}
      \EndIf
    \EndFor
    \For{$i = 1,\ldots,k'$}
      \If{$l(C(\lambda_i)) > x$ {\bf and} $r(C(\lambda_i)) < x'$ {\bf and} $LL \notin \theta(\lambda_i)$}
        \State {\bf return} {\it False}
        \Comment{$\lambda_i$ must stretch on both sides}
      \EndIf
      \If{$l(C(\lambda_i)) > x$ {\bf and} $r(C(\lambda_i)) = x'$ {\bf and} $LF \notin \theta(\lambda_i)$}
        \State {\bf return} {\it False}
        \Comment{$\lambda_i$ must stretch only on the left side}
      \EndIf
      \If{$l(C(\lambda_i)) = x$ {\bf and} $r(C(\lambda_i)) < x'$ {\bf and} $FL \notin \theta(\lambda_i)$}
        \State {\bf return} {\it False}
        \Comment{$\lambda_i$ must stretch only on the right side}
      \EndIf
      \If{$l(C(\lambda_i)) = x$ {\bf and} $r(C(\lambda_i)) = x'$ {\bf and} $FF \notin \theta(\lambda_i)$}
        \State {\bf return} {\it False}
        \Comment{$\lambda_i$ must stretch on neither side}
      \EndIf
    \EndFor
    \State {\bf return} {\it True}
  \end{algorithmic}
\end{algorithm}

\begin{algorithm}
  \caption{Algorithm for parallel node}\label{alg:p-node}
  \begin{algorithmic}[1]
    \State\label{alg:p_sort}{\bf sort} $\lambda_i$'s by the value $l(C(\lambda_i))$
    \For{$i=1,\ldots,k'$}\label{alg:p_init_b}
      \State $l_i \gets l(C(\lambda_i))$, $r_i \gets r(C(\lambda_i))$
      \Comment{left- and right- endpoints of cores}
    \EndFor
    \State $r_0 \gets x$, $l_{k'+1} \gets x'$
    \State\label{alg:p_init_e}$Q \gets \emptyset$ \Comment{set of closed gaps}
    \For{$i=0,\ldots,k'$}\label{alg:p_check_b}
      \SimpleIf{$r_{i} > l_{i+1}$}{{\bf return} {\it False}} \Comment{cores overlap}
      \If{$r_{i} = l_{i+1}$}
        \State $Q \gets Q \cup \{i\}$
        \Comment{$\lambda_{i}$ and $\lambda_{i+1}$ touch}
        \SimpleIf{$i>0$}{$\theta(\lambda_{i}) \gets \theta(\lambda_{i}) \setminus \{FL,LL\}$}
          \Comment Use right-fixed rep.\ of $\lambda_{i}$
        \SimpleIf{$i<k'$}{$\theta(\lambda_{i+1}) \gets \theta(\lambda_{i+1}) \setminus \{LF,LL\}$}
          \Comment Use left-fixed rep.\ of $\lambda_{i+1}$
      \EndIf
    \EndFor\label{alg:p_check_e}
    \If{$\left(k > k' \text{ {\bf or} } (s_\mu,t_\mu) \in E(G_\mu)\right)$ {\bf and} $|Q| = k'+1$}\label{alg:p_onegap}
      \State {\bf return} {\it False}
      \Comment{there is no gap}
    \EndIf
    \For{$i=1,\ldots,k'$}\label{alg:p_greedy_b}
      \SimpleIf{$\theta(\lambda_i) = \emptyset$}{{\bf return} {\it False}}
      \SimpleIf{$LL \in \theta(\lambda_i)$}{$Q \gets Q \cup \{i-1,i\}$}
        \Comment{close both gaps}
      \SimpleElsIf{$i-1 \notin Q$ {\bf and} $LF \in \theta(\lambda_i)$}{$Q \gets Q \cup \{i-1\}$}
        \Comment{close left gap}
      \SimpleElsIf{$FL \in \theta(\lambda_i)$}{$Q \gets Q \cup \{i\}$}
        \Comment{close right gap}
    \EndFor\label{alg:p_greedy_e}
    \State {\bf return} $(s_\mu,t_\mu) \in E(G_\mu)$ {\bf or} $k-k' >= k'+1-|Q|$
    \Comment{can close all gaps}
  \end{algorithmic}
\end{algorithm}

\subsection{Algorithm for $R$-node}
\label{app:rnode}

Section~\ref{sec:rectangular_algorithm_slow} presents an algorithm that tests if there exists a rectangular bar visibility representation that extends a partial representation of an $st$-graph $G$.
In the following, we present a detailed discussion of the algorithm for an $R$-node in the decomposition.

\subcase{Case R. $\mu$ is an $R$-node.}
By the R-Tiling Lemma, the set of possible tilings of $B(\mu)$ by $B(\mu_1), \ldots, B(\mu_k)$
is in correspondence with the triples $(\mathcal{E}, \xi, \chi)$, where $\mathcal{E}$ is a planar embedding of $skel(\mu)$, $\xi$ is an $st$-valuation of $\mathcal{E}$, and $\chi$ is an $st$-valuation of $\mathcal{E}^*$.
To find an appropriate tiling of $B(\mu)$ (that yields an $[x,x']$-representation of $\mu$) we search through the set of such triples.
Since $\mu$ is a rigid node, there are only two planar embeddings of $skel(\mu)$ and we consider both of them separately.
Let $\mathcal{E}$ be one of these planar embeddings.
Since the $y$-coordinate for each vertex of $G$ is already fixed, the $st$-valuation $\xi$ is given by the $y$-coordinates of the vertices from $skel(\mu)$.
Now, it remains to find an $st$-valuation $\chi$ of $\mathcal{E}^*$, \ie{}, to determine the $x$-coordinate of the splitting line for every face.
First, for every face $f$ in $V(\mathcal{E}^*)$ we compute an initial set of possible placements for the splitting line of $f$ by taking into account the partial representation $\phi'$.
If $f$ is an inner face of $\mathcal{E}$, then we have the following restrictions on $\chi(f)$:
\begin{itemize}
\item If $u$ is a fixed vertex on the left path of $f$, then $\chi(f) = r(u)$.
\item If $u$ is a fixed vertex on the right path of $f$, then $\chi(f) = l(u)$.
\item If $\lambda$ is a child of $\mu$ whose core is non-empty, and the edge corresponding to $\lambda$ is on the left path of $f$, then $\chi(f) \geq r(C(\lambda))$.
\item If $\lambda$ is a child of $\mu$ whose core is non-empty, and the edge corresponding to $\lambda$ is on the right path of $f$, then $\chi(f) \leq l(C(\lambda))$.
\end{itemize}
We impose analogous conditions for the faces $s^*$ and $t^*$.

Let $\mathcal{X}'(f)$ be a set of all $\chi(f)$ in $[x,x']$ that satisfy all the above conditions.
If $\mathcal{X}'(f) = \emptyset$ for some face $f$ in $V(\mathcal{E}^*)$ or $x \notin \mathcal{X}'(s^*)$ or $x' \notin \mathcal{X}'(t^*)$ then there is no $[x,x']$-representation of $G_\mu$.
Since we are looking for an $[x,x']$-representation of $G_\mu$, we set $\mathcal{X}'(s^*) = [x,x]$ and $\mathcal{X}'(t^*) = [x',x']$ as the splitting line for $s^*$ ($t^*$) must be set to $x$ ($x'$).

Now, we further restrict the possible values for $\chi(f)$ by taking into account the fact that $\chi$ needs to be an $st$-valuation of $\mathcal{E}^*$.
For every two faces $f$ and $g$ in $V(\mathcal{E}^*)$:
\begin{itemize}
  \item If $g$ is to the left of $f$, then $\chi(f) > l(\mathcal{X}'(g))$.
  \item If $g$ is to the right of $f$, then $\chi(f) < r(\mathcal{X}'(g))$.
\end{itemize}
Let $\mathcal{X}(f)$ be the set of all $\chi(f)$ such that $\chi(f) \in \mathcal{X}'(f)$ and that satisfy the above conditions.
If $\mathcal{X}(f)$ is empty for some face $f$ in $V(\mathcal{E}^*)$, then there is no $[x,x']$-representation of $G_\mu$.
We assume that $\mathcal{X}(f)$ is non-empty for every $f$ in $V(\mathcal{E}^*)$.
One can easily verify the following.
\begin{claim} \
\label{claim:X_sets}
  For every face $f$ in $V(\mathcal{E}^*)$, $\mathcal{X}(f)$ is an interval in $[x,x']$.
  For every two faces $f$ and $g$ such that $f$ is to the left of $g$, we have that:
  \begin{itemize}
    \item $l(\mathcal{X}(f)) \leq l(\mathcal{X}(g))$ and if $l(\mathcal{X}(f)) = l(\mathcal{X}(g))$ then $\mathcal{X}(g)$ is open from the left side.
    \item $r(\mathcal{X}(f)) \leq r(\mathcal{X}(g))$ and if $r(\mathcal{X}(f)) = r(\mathcal{X}(g))$ then $\mathcal{X}(f)$ is open from the right side.
  \end{itemize}
\end{claim}

A face $f$ in $V(\mathcal{E}^*)$ is \emph{determined} if $\mathcal{X}(f)$ is a singleton (\ie{}, the location of the splitting line of $f$ is already fixed); otherwise $f$ is \emph{undetermined}.

In what follows, we construct a 2-CNF formula $\Phi$ that is satisfiable if and only if an $[x,x']$-representation of $\mu$ exists.

\noindent {\bf Variables of $\Phi$.}
For every child $\lambda$ of $\mu$ whose core is non-empty, we introduce two boolean variables: $l_\lambda$ and $r_\lambda$, which have the following interpretation:
\begin{itemize}
  \item The positive (negative) value of variable $l_\lambda$ means that we use a left-loose (left-fixed) representation of node $\lambda$.
\item The positive (negative) value of variable $r_\lambda$ means that we use a right-loose (right-fixed) representation of node $\lambda$.
\end{itemize}
For every inner face $f$ of $\mathcal{E}$ we introduce two boolean variables: $l_f$ and $r_f$, which have the following interpretation:
\begin{itemize}
\item The variable $l_f$ is positive when the splitting line of $f$ is set strictly to the right of $l(\mathcal{X}(f))$.
It is negative when $\chi(f) = l(\mathcal{X}(f))$.
\item The variable $r_f$ is positive when the splitting line of $f$ is set strictly to the left of $r(\mathcal{X}(f))$.
It is negative when $\chi(f) = r(\mathcal{X}(f))$.
\end{itemize}
In particular, if $\mathcal{X}(f)$ is open from the left (right) then $l_f$ ($r_f$) is positive.
When $f$ is determined then $l_f$ and $r_f$ are negative.
For the left outer face $s^*$ and the right outer face $t^*$ we introduce variables $r_{s^*}$ and $l_{t^*}$.
Since $s^*$ and $t^*$ are determined, the corresponding variables are always set to negative.

\noindent {\bf Clauses of $\Phi$.}
We split the clauses of $\Phi$ into four types.

\subcase{Type I.}
These clauses ensure that no infeasible representation of a child of $\mu$ is used.
%              propagate the information about the possible representation types of the children of $\mu$.
\begin{itemize}
\item For every child $\lambda$ with non-empty core and for every type of representation of $\lambda$ which is not feasible, we add a clause that forbids using representation of this type.
For example, when there is no LF-representation of $\lambda$ we add a clause $\lnot (l_\lambda \wedge \lnot r_\lambda)$ (\ie{}, $\lnot l_\lambda \vee r_\lambda$) to $\Phi$.
\end{itemize}

\subcase{Type II.}
These clauses enforce the meaning of variables $l_f$ and $r_f$ for every face $f$ in $V(\mathcal{E}^*)$.
\begin{itemize}
\item For every determined inner face $f$, we add the clauses $(\lnot l_f)$ and $(\lnot r_f)$, and for $s^*$ and $t^*$ we add the clauses
$(\lnot r_{s^*})$ and $(\lnot l_{t^*})$.

\item For every undetermined inner face $f$, we add the clause $(l_f)$ if $\mathcal{X}(f)$ is open from the left side, and $(r_f)$ if $\mathcal{X}(f)$ is open from the right side.
If $\mathcal{X}(f)$ is closed from the left and closed from the right side, we add the clause $(l_f \vee r_f)$ as the splitting line of $f$ cannot be placed in both endpoints of $\mathcal{X}(f)$ simultaneously.
\end{itemize}

\subcase{Type III.}
These clauses enforce a `proper tiling' of every face $f$ in $V(\mathcal{E}^*)$ (see the Face Condition Lemma).
This ensures that the bounding boxes associated with the left path and the right path of $f$ can be aligned to the splitting line of $f$.
\begin{itemize}
\item For every face $f$ and for every node $\lambda$ on the left path of $f$ with a non-empty core:
\begin{itemize}
\item We add the clause $(l_f \implies r_\lambda)$.
\item If $r(C(\lambda)) < l(\mathcal{X}(f))$, we add the clause $(r_\lambda)$.
\item If $r(C(\lambda)) = l(\mathcal{X}(f))$, we add the clause  $(\lnot l_f \implies \lnot r_\lambda)$.
\end{itemize}
\end{itemize}
The clause $(l_f \implies r_\lambda)$ asserts that whenever the splitting line of $f$ is set to the right of $l(\mathcal{X}(f))$,
then a right-loose representation of $\lambda$ is necessary to align the bounding box of $\lambda$ to the splitting line of $f$.
The two remaining clauses have similar meaning.

We add analogous clauses for the nodes whose cores are non-empty and that correspond to the edges from the right path of $f$.

\subcase{Type IV.}
These clauses ensure that the $x$-coordinates of the splitting lines form an $st$-valuation of $\mathcal{E}^*$.
\begin{itemize}
  \item For every pair of faces $f$ and $g$ in $V(\mathcal{E}^*)$ such that $f$ is to the left of $g$ and such that $r(\mathcal{X}(f)) \geq l(\mathcal{X}(g))$, we add the clause $(\lnot r_f \implies l_g)$.
\end{itemize}
Such a clause forbids setting $\chi(f)=r(\mathcal{X}(f))$ and $\chi(g) = l(\mathcal{X}(g))$ -- such an assignment of $\chi(f)$ and $\chi(g)$ would not be a valid $st$-valuation.

\begin{claim}\label{clm:2sat_formula}
Let $\mu$ be an $R$-node and let $\mathcal{E}$ be a planar embedding of the skeleton of $\mu$.
An $[x,x']$-representation of $\mu$ that corresponds to a planar embedding $\mathcal{E}$ exists iff $\Phi$ is satisfiable.
\end{claim}
\begin{proof}
Suppose that $\phi$ is an $[x,x']$-representation of $\mu$.
For every face $f$ in $V(\mathcal{E}^*)$, the splitting line $\chi(f)$ of $f$ in $\phi$ satisfies $\chi(f) \in \mathcal{X}(f)$.
Thus, $\chi(f)$ determines the following assignment for $l_f$ and $r_f$:
$l_f$ is positive iff $\chi(f) > l(\mathcal{X}(f))$,
$r_f$ is positive iff $\chi(f) < r(\mathcal{X}(f))$.
For every child $\lambda$ of $\mu$ such that $C(\lambda) \neq \emptyset$ we set the variables $l_\lambda$ and $r_\lambda$ as follows:
$l_\lambda$ is positive iff $\lambda$ is left-loose in $\phi$,
$r_\lambda$ is positive iff $\lambda$ is right-loose in $\phi$.
One can easily check that this assignment satisfies $\Phi$.

Suppose now that $\Phi$ is satisfiable.
We define an $[x,x']$-representation $\phi$ of $\mu$ by setting a splitting line $\chi(f)$ for every face $f$ in $V(\mathcal{E}^*)$.
To conclude that $\phi$ is an $[x,x']$-representation it is enough to check that:
\begin{enumerate}
  \item \label{pro:chi_st} The function $\chi$ is an $st$-valuation of $\mathcal{E}^*$.
  \item \label{pro:chi_fix} For every $u \in V'(\mu)$ we have
  $l_{\phi}(u) = \chi($left face of $u)$ and $r_{\phi}(u) = \chi($right face of $u)$.
  \item \label{pro:chi_children} For every child $\lambda$ of $\mu$ such that $C(\lambda) \neq \emptyset$ we have that
  $G_\lambda$ has an $[\chi($left face of $\lambda),$ $\chi($right face of $\lambda)]$-representation.
\end{enumerate}
First, we define $\chi(f) = l(\mathcal{X}(f))$ when $l_f$ is negative.
We also set $\chi(f)=r(\mathcal{X}(f))$ when $r_f$ is negative.
Note that this definition is unambiguous as both $l_f$ and $r_f$ are negative only for a determined face $f$ by the satisfiability of the clauses of Type~II.
By Claim~\ref{claim:X_sets} and by the satisfiability of the clauses of Type~IV,
for any two faces $f$ and $g$ in $V(\mathcal{E}^*)$ for which $\chi(f)$ and $\chi(g)$ are already fixed,
we have $\chi(f) < \chi(g)$ whenever $f$ is to the left of $g$.
Notice that $\chi(f)$ is not yet determined for inner faces $f$ for which $l_f$ and $r_f$ are positive.
For such a face $f$, let $\mathcal{X}''(f)$ contain all values $z$ such that $l(\mathcal{X}(f)) < z < r(\mathcal{X}(f))$ and:
\begin{itemize}
\item $\chi(g) < z$ whenever $g$ is a face to the left of $f$  and the value $\chi(g)$ is already fixed, and
\item $z < \chi(h)$ whenever $h$ is a face to the right of $f$ and the value $\chi(h)$ is already fixed.
\end{itemize}
We claim that $\mathcal{X}''(f)$ is an open, non-empty interval. Indeed, if $\mathcal{X}''(f)$ is empty,
then there are faces $g$ and $h$ with the values $\chi(g),\chi(h)$ fixed such that
$g$ is to the left of $h$ and $\chi(g) \geq \chi(h)$, which contradicts our previous observation.
Moreover, for any two faces $f_1,f_2$ such that $f_1$ is to the left of $f_2$ and neither $\chi(f_1)$ nor $\chi(f_2)$ is fixed, we have that $l(\mathcal{X}''(f_1)) \leq l(\mathcal{X}''(f_2))$ and $r(\mathcal{X}''(f_1)) \leq r(\mathcal{X}''(f_2))$.
Thus, for every face $f$ for which both $l_f,r_f$ are positive, we can choose a value $\chi(f)$ from $\mathcal{X}''(f)$ so that $\chi$ is an $st$-valuation of $\mathcal{E}^*$.
We need to check the remaining conditions~\eqref{pro:chi_fix} and~\eqref{pro:chi_children}.
Condition~\eqref{pro:chi_fix} is satisfied since, for a determined face $f$ we have chosen $\chi(f)$ from the singleton $\mathcal{X}(f)$.
Condition~\eqref{pro:chi_children} follows from the satisfiability of the clauses of Type~I and Type~III, and by the Stretching Lemma.
\end{proof}

\subsection{Complexity considerations}
\label{app:rectangular_complexity}

To compute the feasible representation for a node $\mu$ with $k$ children, our algorithm works in $\Oh{k}$ time if $\mu$ is an $S$-node.
Algorithm~\ref{alg:p-node} for a $P$-node $\mu$ needs to sort the children of $\mu$ and thus, it works in $\Oh{k \log k}$ time.
For an $R$-node, the number of clauses of Types I, II and III is $\Oh{k}$.
The number of clauses of Type~IV is $\Oh{k^2}$ and for some graphs, it is quadratic.
Thus, the algorithm works in $\Oh{k^2}$ time for an $R$-node.
Since the number of all edges in all nodes of $T$ is $\Oh{n}$, the whole algorithm works in $\Oh{n^2}$ time.

The bottleneck of the algorithm presented in Section~\ref{sec:rectangular_algorithm_slow} is the number of clauses of Type~IV in the 2-CNF formula constructed for $R$-nodes.
In the presented algorithm we added one clause $(\lnot r_f \implies l_g)$ for any two faces $f$ and $g$ in $V(\mathcal{E}^*)$ such that $f$ is to the left of $g$ and $r(\mathcal{X}(f)) \geq l(\mathcal{X}(g))$.
The number of such pairs of faces can be quadratic.
In the next section we present a different, less direct, approach that uses a smaller number of clauses to express the same set of constraints.

\subsection{Faster algorithm}
\label{app:rectangular_algorithm_fast}

We can treat the $st$-graph $\mathcal{E}^*$ as a planar poset with single minimal and single maximal element.
Using the result by Baker, Fishburn and Roberts~\cite{BakerFR71} we know that such a poset has dimension at most $2$.
Thus, there are two numberings $p$ and $q$ of the vertices of $\mathcal{E}^*$ such that a face $f$ is to the left of a face $g$ iff $p(f) < p(g)$ and $q(f) < q(g)$.
Such numberings correspond to dominance drawings of $st$-graphs and can be computed in linear time~\cite{BattistaTT92}.

For each face $f$, we have two boolean variables $l_f$ and $r_f$, two real values $\lambda_f = l(\mathcal{X}(f))$, $\rho_f = r(\mathcal{X}(f))$, and two integer values $p_f = p(f)$ and $q_f = q(f)$.
We want to introduce a small set of 2-CNF clauses so that we can imply $(\lnot r_f \implies l_g) = (\lnot l_g \implies r_f)$ whenever $p_f < p_g$, $q_f < q_g$, and $\rho_f \geq \lambda_g$.

We compute the set of implications with a sweep line algorithm.
During the course of the algorithm, variables $r_f$ are stored in multiple \emph{persistent} balanced binary search trees with $\rho_f$ as the sorting key.

For an overview of persistent data structures, refer to~\cite{DriscollSST89}.
However, we only need the ideas presented in~\cite{Myers84}, which are summarized in this paragraph.
The tree structure used in our algorithm is a modification of the AVL tree.
A node $\alpha$ of the tree stores a pointer to its left child $left(\alpha)$, a pointer to its right child $right(\alpha)$ and the sorting key $key(\alpha)$.
Note that there are no parent links.
The difference from regular AVL trees is that no node is ever modified.
When the insertion or the balancing strategy needs to modify a node $\alpha$, a copy of $\alpha$ is made instead, creating a new node $\alpha'$ and setting $left(\alpha') = left(\alpha)$, $right(\alpha')=right(\alpha)$ and $key(\alpha')=key(\alpha)$.
If $\alpha$ has a parent $p_{\alpha}$, then the parent is copied too.
The node $p_{\alpha}'$ is a copy of $p_{\alpha}$ with the exception of the pointer to $\alpha$, which is replaced with a pointer to $\alpha'$.
The entire path from $\alpha$ to the root of the tree is copied in this fashion.
This way, each addition to the tree introduces a logarithmic number of new nodes.
After each addition, we get a new root node, that represents the new tree, the old tree is represented by the previous root and all but a logarithmic number of the nodes are shared by both trees.
The graph of old and new nodes together with edges from nodes to their children is an acyclic digraph.

Now, in our algorithm, each tree keeps some set of boolean variables $r_f$ sorted by the value of $\rho_f$.
More specifically, with every vertex $\alpha$ we associate a value $var(\alpha)$.
To add a variable $r_f$ into the tree we add a new node $\alpha$ with $key(\alpha) = \rho_f$ and $var(\alpha) = r_f$.
Additionally, with each node $\alpha$ of the tree we associate a second boolean variable $var'(\alpha)$ and for each child node $\beta$ we add a clause $(var'(\alpha) \implies var'(\beta))$.
Finally, for each node $\alpha$ we add a clause $(var'(\alpha) \implies var(\alpha))$.

To simplify the presentation, assume that we have $n=2^k$ faces.
Let numberings $p$ and $q$ take values $0,\ldots,n-1$.
We construct the 2-CNF formula in the following way.
First, for each interval of integers $[j \cdot 2^i, (j+1) \cdot 2^i), 0 \leq i \leq k, j<\frac{n}{2^i}$ we have one persistent balanced binary search tree.
The tree for the interval $[a,b)$ is going to keep variables $r_f$ for faces $f$ such that $q_f \in [a,b)$.

We process faces, one by one, in order of increasing values of $p_f$.
After processing each face $f$, we add the variable $r_f$ to all trees $[a,b)$ such that $q_f \in [a,b)$.
There are $k+1$ such trees and an addition to each tree takes $\Oh{\log{n}}$ time and introduces $\Oh{\log{n}}$ new boolean variables.

This way, when we process face $g$, then any earlier processed face $f$ satisfies $p_f < p_g$ and no other face satisfies this condition.
Now, for any value $q_g$ we can select a logarithmic size subset $S$ of the trees such that the union of the intervals of the trees is exactly $[0,q_g)$.
The variables $r_f$ stored in these trees are exactly those for faces that satisfy both $p_f < p_g$ and $q_f < q_g$.
This is exactly the set of faces that are to the left of the face $g$.

Each tree in $S$ stores variables $r_f$ sorted by $\rho_f$.
We can execute a binary search for the left-most node with the key no smaller than $\lambda_g$.
During the search, when we descend from an inner node $\alpha$ to the left child or when $\alpha$ is the final node of the search, we add clauses $(\lnot l_g \implies var'(right(\alpha))$ and $(\lnot l_g \implies var(\alpha))$.
The first implication is forwarded over the tree to all nodes in the right subtree.
This way, after completing the search, we have that $\lnot l_g$ implies $r_f$ for all faces $f$ such that $p_f < p_g$, $q_f < q_g$ and $\rho_f \geq \lambda_g$ -- exactly as intended.

The total running time of this procedure is $\Oh{n\log^2{n}}$ and it produces at most that many variables and clauses.
This shows Theorem \ref{th:main-plane-st}.

\section{Hardness results}
\label{app:hardness}
%In this section we prove all \NP-hardness results mention in the introduction.

\subsection{\NP-complete problems}
Our hardness proofs use reductions from the following \NP-complete problems:

\smallskip
\noindent \underline{$\pmthreesat$}:\smallskip\\
\noindent{\bf Input:} A \emph{rectilinear} planar representation of a $\threesat$ formula in which
 each variable is a horizontal segment on the $x$-axis,
 each clause is a horizontal segment above or below the $x$-axis
 with straight-line vertical connections to the variables it includes.
 All positive clauses are above the $x$-axis and all negative clauses are below the $x$-axis.
 There are no clauses including both positive and negative occurrences of variables.
 See Figure~\ref{fig:sat} for an example.\\
\noindent{\bf Question:}
 Is the formula satisfiable?\smallskip\\
\pmthreesat{} is known to be \NP-complete thanks to de Berg and Khosravi~\cite{BergK10}.
\begin{figure}[t]
\begin{tikzpicture}
\definecolor{light-gray}{gray}{0.90}
\tikzstyle{every node}=[inner sep=1pt,fill=light-gray]

\begin{scope}[shift={(0,0)}, xscale=0.63, yscale=1]
\draw[very thick, gray] (-1,0) -- (18,0);

\draw[fill=light-gray] (0,0.25) rectangle (2,-0.25);
\node (x1) at (1,0) {$x_1$};

\draw[fill=light-gray] (3,0.25) rectangle (5,-0.25);
\node (x2) at (4,0) {$x_2$};

\draw[fill=light-gray] (6,0.25) rectangle (8,-0.25);
\node (x3) at (7,0) {$x_3$};

\draw[fill=light-gray] (9,0.25) rectangle (11,-0.25);
\node (x4) at (10,0) {$x_4$};

\draw[fill=light-gray] (12,0.25) rectangle (14,-0.25);
\node (x5) at (13,0) {$x_5$};

\draw[fill=light-gray] (15,0.25) rectangle (17,-0.25);
\node (x6) at (16,0) {$x_6$};

\draw[fill=light-gray] (1.33,1.25) rectangle (6.66,0.75);
\node (c1) at (4,1) {$x_1 \vee x_2 \vee x_3$};
\draw[very thick] (1.33,0.25) -- (1.33,0.75);
\draw[very thick] (6.66,0.25) -- (6.66,0.75);
\draw[very thick] (4,0.25) -- (4,0.75);

\draw[fill=light-gray] (10,1.25) rectangle (15.66,0.75);
\node (c1) at (12.83,1) {$x_4 \vee x_5 \vee x_6$};
\draw[very thick] (10,0.25) -- (10,0.75);
\draw[very thick] (13,0.25) -- (13,0.75);
\draw[very thick] (15.66,0.25) -- (15.66,0.75);

\draw[fill=light-gray] (0.66,2.25) rectangle (16.33,1.75);
\node (c1) at (8.5,2) {$x_1 \vee x_3 \vee x_6$};
\draw[very thick] (0.66,0.25) -- (0.66,1.75);
\draw[very thick] (7.33,0.25) -- (7.33,1.75);
\draw[very thick] (16.33,0.25) -- (16.33,1.75);

\draw[fill=light-gray] (4.33,-0.75) rectangle (9.5,-1.25);
\node (c1) at (6.915,-1) {$\overline{x_2} \vee \overline{x_3} \vee \overline{x_4}$};
\draw[very thick] (4.33,-0.25) -- (4.33,-0.75);
\draw[very thick] (7, -0.25) -- (7,-0.75);
\draw[very thick] (9.5,-0.25) -- (9.5,-0.75);

\draw[fill=light-gray] (1.33,-1.75) rectangle (10,-2.25);
\node (c1) at (5.665,-2) {$\overline{x_1} \vee \overline{x_2} \vee \overline{x_4}$};
\draw[very thick] (1.33,-0.25) -- (1.33,-1.75);
\draw[very thick] (3.66, -0.25) -- (3.66,-1.75);
\draw[very thick] (10,-0.25) -- (10,-1.75);

\draw[fill=light-gray] (0.66,-2.75) rectangle (13,-3.25);
\node (c1) at (6.83,-3) {$\overline{x_1} \vee \overline{x_4} \vee \overline{x_5}$};
\draw[very thick] (0.66,-0.25) -- (0.66,-2.75);
\draw[very thick] (10.5,-0.25) -- (10.5,-2.75);
\draw[very thick] (13, -0.25) -- (13,-2.75);

%\draw[fill=light-gray] (4.6,1.25) rectangle (9.4,0.75);
%\node (c2) at (7,1) {$x_2 \vee x_3 \vee x_4$};
%\draw[very thick] (1.4,0.25) -- (1.4,0.75);
%\draw[very thick] (3.6,0.25) -- (3.6,0.75);

\end{scope}

\end{tikzpicture}
\caption{A $\pmthreesat$ formula with variables $x_1,\ldots,x_6$;
  positive clauses $\{x_1, x_3, x_6\}$, $\{x_1, x_2, x_3 \}$, and $\{x_4, x_5, x_6 \}$;
  and negative clauses  $\{ \overline{x_1}, \overline{x_4}, \overline{x_5} \}$, $\{ \overline{x_1}, \overline{x_2}, \overline{x_4}\}$, and $\{ \overline{x_2}, \overline{x_3}, \overline{x_4}\}$.
} 
\label{fig:sat}
\end{figure}
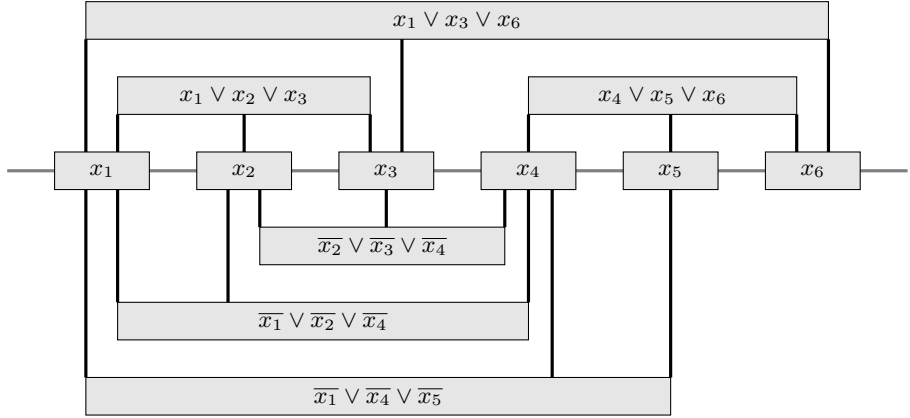

\smallskip
\noindent \underline{\threepart{}}:\smallskip\\
\noindent{\bf Input:}
 A set of positive integers $w, a_1, a_2, \ldots, a_{3m}$ such that for each $i=1,\ldots,3m$, we have $\frac{w}{4} < a_i < \frac{w}{2}$. \\
\noindent{\bf Question:}
 Can $\{a_1, \ldots, a_{3m}\}$ be partitioned into $m$ triples, such that the total sum of each triple is exactly $w$?\smallskip\\
\threepart{} is known to be strongly \NP-complete~\cite{GareyJ79}, \ie{}, the problem remains \NP-complete even when the integers given in the input are encoded in unary.

\subsection{Representations of undirected graphs}
\label{app:hardness_undirected}

\hackTcounter{th:var-emb-hardness-general}
\begin{theorem}
The Bar Visibility Representation Extension Problem is \NP-complete.
\end{theorem}
\unhackTcounter
\begin{proof}%[Proof of Theorem~\ref{th:var-emb-hardness-general}]
  It is clear that the bar visibility representation extension problem is in \NP.
  To prove completeness, we present a reduction from $\pmthreesat$.
  Given a formula $\phi$ we construct a graph $G$ and a partial representation $\psi'$ that is extendable to a representation of the whole $G$ if and only if $\phi$ is satisfiable.
  The reduction constructs a planar boolean circuit that simulates the formula $\phi$.
  The bars assigned to fixed vertices of $G$ create wires and gates of the circuit.
  Unrepresented vertices of $G$ correspond to boolean values transmitted over the wires.
  Our construction uses several boolean gates: a NOT gate, a XOR gate, a special gate which we call a CXOR gate, and an OR gate.

  In the figures illustrating this proof, the red bars denote the fixed vertices of $G$ and the black bars denote the unrepresented vertices.
  A bar may have its left (right) endpoint marked or not depending on whether the bar extends to the left (right) of the figure or not.
  The figures also contain some vertical ranges.
  These ranges are only required for the description of the properties of the gadgets.

  For readability, the figures contain only schemes of previously defined gadgets.
  Whenever a scheme appears in a figure, its area is colored gray.

\begin{figure}[t]
\begin{center}
\begin{tikzpicture}[>=latex]
\definecolor{light-gray}{gray}{0.90}

\begin{scope}[xscale=0.8, yscale=0.6]

\begin{scope}[shift={(3.5,0)}]
  \begin{tiny}
  %\draw[fill=light-gray, color=light-gray] (0,0.5) rectangle (3,2);
  %\draw[fill=light-gray, color=light-gray] (0,3) rectangle (3,4.5);
  
  \node at (1.5,0.25) {$\bot$};
  \draw[thick,color=red,(-)] (0,0.5)--(3,0.5);
  
  \node at (1.5,4.75) {$\top$};
  \draw[thick,color=red,(-)] (0,4.5)--(3,4.5);

  \draw[dotted, (-)] (0.2,3)--(0.2,4.5);
  \draw[dotted, (-)] (0.2,0.5)--(0.2,2);

  \draw[dotted, (-)] (2.8,3)--(2.8,4.5);
  \draw[dotted, (-)] (2.8,0.5)--(2.8,2);
  
  \draw[very thick, -] (0.0,3.75)--(3, 3.75);
  \node at (1.5,4) {$v$};

  \end{tiny}
\end{scope}

\begin{scope}[shift={(8.5,0)}]
  \begin{tiny}
  %\draw[fill=light-gray, color=light-gray] (0,0.5) rectangle (3,2);
  %\draw[fill=light-gray, color=light-gray] (0,3) rectangle (3,4.5);
  
  \node at (1.5,0.25) {$\bot$};
  \draw[thick,color=red,(-)] (0,0.5)--(3,0.5);
  
  \node at (1.5,4.75) {$\top$};
  \draw[thick,color=red,(-)] (0,4.5)--(3,4.5);

  \draw[dotted, (-)] (0.2,3)--(0.2,4.5);
  \draw[dotted, (-)] (0.2,0.5)--(0.2,2);

  \draw[dotted, (-)] (2.8,3)--(2.8,4.5);
  \draw[dotted, (-)] (2.8,0.5)--(2.8,2);
  
  \draw[very thick, -] (0.0,1.25)--(3, 1.25);
  \node at (1.5,1.5) {$v$};

  \end{tiny}
\end{scope}
\end{scope}

\end{tikzpicture}
\caption{Wire transmitting positive value (on the left) and negative value (on the right)}
\label{fig:eps_wire}
\end{center}
\end{figure}
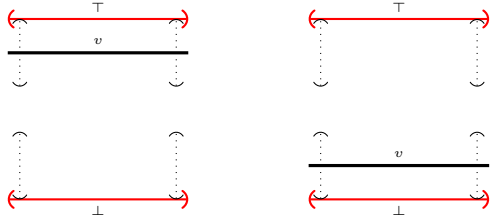

  Each wire, see Figure~\ref{fig:eps_wire}, in the circuit is an empty space between two fixed vertices, $\top$ and $\bot$, whose bars are placed one above the other.
  One unrepresented vertex $v$, that is adjacent to both $\top$ and $\bot$, corresponds to the value transmitted over the wire.
  The construction of the gates ensures that in each wire there are exactly two disjoint rectangular regions in which $v$ can be placed.
  In the figures, the small horizontal lines with a dotted line between them are used to mark those regions and are not part of the construction.
  Placement of $v$ anywhere in the top (bottom) region corresponds to a transmission of a positive (negative) value.
  We show a collection of gadgets for the gates that use such wires as inputs, and outputs.
  Similarly to wires, each gate is bounded from the top and bottom by two bars.
  This way, it is easy to control the visibility between bars from different gates and wires.

\begin{figure*}
\begin{center}
\begin{tikzpicture}[>=latex]
\definecolor{light-gray}{gray}{0.90}

\begin{scope}[xscale=0.7, yscale=0.6]

%%  \tikzstyle{every label}=[rectangle,inner sep=2pt,draw=none,fill=none]
  \begin{scope}[shift={(-7.0,0)}]
    \tikzstyle{every node}=[circle,minimum size=5pt,inner sep=0pt,draw,fill]
    \begin{footnotesize}
    \node[label=below:$\bot$,fill=red] (bot) at (5.0,0.5) {};
    \node[label=right:$y$,fill=black] (y) at (6,2.5) {};
    \node[label=below:$a$,fill=red] (a) at (5,2.0) {};
    \node[label=above:$b$,fill=red] (b) at (5,3) {};
    \node[label=left:$x$,fill=black] (x) at (4,2.5) {};
    \node[label=above:$\top$,fill=red] (top) at (5.0,4.5) {};
    \draw[thick,-] (bot)--(y);
    \draw[thick,-] (y)--(a);
    \draw[thick,-] (y)--(b);
    \draw[thick,-] (a)--(b);
    \draw[thick,-] (a)--(x);
    \draw[thick,-] (b)--(x);
    \draw[thick,-] (x)--(top);
    \draw[thick,-] (x)--(bot);
    \draw[thick,-] (y)--(top);

    \end{footnotesize}
  \end{scope}

\begin{scope}[shift={(0,0)}]

  \begin{tiny}
  \tikzstyle{every node}=[inner sep=2pt,fill=white]
  
  \node at (1.5,4.75) {$\top$};
  \draw[thick,color=red,(-)] (0,4.5)--(3,4.5);

  \node at (1.5,0.25) {$\bot$};
  \draw[thick,color=red,(-)] (0,0.5)--(3,0.5);

  \node at (1.25,2.25) {$a$};
  \draw[thick,color=red,(-)] (0.5,2)--(2,2);

  \node at (1.75,3.25) {$b$};
  \draw[thick,color=red,(-)] (1,3)--(2.5,3);

  \node at (1.25,4) {$x$};
  \draw[thick, -)] (0,3.75)--(2.5, 3.75);
  
  \node at (1.75,1.5) {$y$};
  \draw[thick, (-] (0.5,1.25)--(3, 1.25);

  \end{tiny}
\end{scope}

\begin{scope}[shift={(3.5,0)}, xscale=0.75]
  \begin{tiny}
  \draw[fill=light-gray, color=light-gray] (0,0.5) rectangle (3,4.5);
  
  \draw[thick,color=red,(-)] (0,0.5)--(3,0.5);
  
  \draw[thick,color=red,(-)] (0,4.5)--(3,4.5);

  \node at (1.5,2.5) {$\text{NOT}$};
  
  \draw[dotted, (-)] (0.2,3)--(0.2,4.5);
  \draw[dotted, (-)] (0.2,0.5)--(0.2,2);

  \draw[dotted, (-)] (2.8,3)--(2.8,4.5);
  \draw[dotted, (-)] (2.8,0.5)--(2.8,2);
  
  \draw[very thick, -] (0.0,3.75)--(1.0, 3.75);
  \node at (0.5,4) {$x$};
  
  \draw[very thick, -] (2.0,1.25)--(3, 1.25);
  \node at (2.5,1.5) {$y$};

  \end{tiny}
\end{scope}

\begin{scope}[shift={(6.75,0)}]

  \begin{tiny}
  \tikzstyle{every node}=[inner sep=2pt,fill=white]
  
  \node at (1.5,4.75) {$\top$};
  \draw[thick,color=red,(-)] (0,4.5)--(3,4.5);

  \node at (1.5,0.25) {$\bot$};
  \draw[thick,color=red,(-)] (0,0.5)--(3,0.5);

  \node at (1.25,2.25) {$a$};
  \draw[thick,color=red,(-)] (0.5,2)--(2,2);

  \node at (1.75,3.25) {$b$};
  \draw[thick,color=red,(-)] (1,3)--(2.5,3);

  \node at (1.25,1.5) {$x$};
  \draw[thick, -)] (0,1.25)--(2.5, 1.25);
  
  \node at (1.75,4) {$y$};
  \draw[thick, (-] (0.5,3.75)--(3, 3.75);

  \end{tiny}
\end{scope}

\begin{scope}[shift={(10.25,0)},xscale=0.75]
  \begin{tiny}
  \draw[fill=light-gray, color=light-gray] (0,0.5) rectangle (3,4.5);
  
  \draw[thick,color=red,(-)] (0,0.5)--(3,0.5);
  
  \draw[thick,color=red,(-)] (0,4.5)--(3,4.5);

  \node at (1.5,2.5) {$\text{NOT}$};
  
  \draw[dotted, (-)] (0.2,3)--(0.2,4.5);
  \draw[dotted, (-)] (0.2,0.5)--(0.2,2);

  \draw[dotted, (-)] (2.8,3)--(2.8,4.5);
  \draw[dotted, (-)] (2.8,0.5)--(2.8,2);
  
  \draw[very thick, -] (0.0,1.25)--(1.0, 1.25);
  \node at (0.5,1.5) {$x$};
  
  \draw[very thick, -] (2.0,3.75)--(3, 3.75);
  \node at (2.5,4) {$y$};

  \end{tiny}
\end{scope}
\end{scope}

\end{tikzpicture}
\caption{The NOT gadget depicted by its two possible representations and their schemes}
\label{fig:eps_not}
\end{center}
\end{figure*}
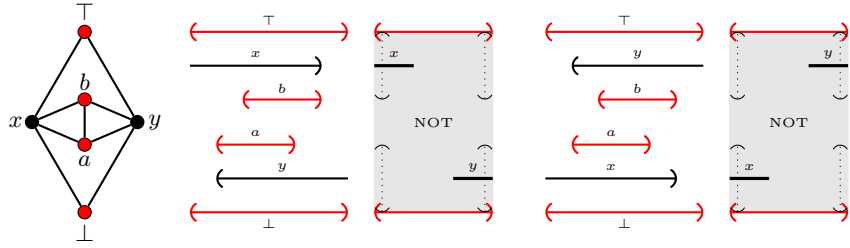

\noindent \textbf{NOT Gadget.}
  Figure~\ref{fig:eps_not} presents a NOT gate and its scheme.
  An unrepresented vertex $x$ ($y$) can transmit the value in the wire that can be placed to the left (right) from the gate.
  Both bars $a$ and $b$ are adjacent to each other, adjacent to vertices $x$ and $y$, and not adjacent to the bounding bars.
  As $a$ and $b$ don't have any other neighbors, the only way to obstruct the visibility gap between $a$, $b$ and bounding bars is to use $x$ and $y$.
  Thus, one of $x$ or $y$ must be placed below $a$ and $b$, and the second must be placed above $a$ and $b$.
  This way we obtain a desired functionality of a NOT gate.
  As the visibility between every two bars does not change across the two representations, the corresponding edges of $G$ are well defined.

\begin{figure*}
\begin{center}
\begin{tikzpicture}[>=latex]
\definecolor{light-gray}{gray}{0.90}

\begin{scope}[xscale=0.7, yscale=0.7]
 
   \begin{scope}[shift={(-7.0,0)}]
    \tikzstyle{every node}=[circle,minimum size=5pt,inner sep=0pt,draw,fill]
    \begin{footnotesize}
    \node[label=above:$\top$,fill=red] (top) at (5.0,6) {};
    \node[label=right:$y_1$,fill=black] (y1) at (6,4.5) {};
    \node[label=above:$b$,fill=red] (b) at (5,4.5) {};
    \node[label=left:$x_1$,fill=black] (x1) at (4,4.5) {};
    \node[label=left:$x_2$,fill=black] (x2) at (4,1.5) {};
    \node[label=below:$a$,fill=red] (a) at (5,1.5) {};
    \node[label=right:$y_2$,fill=black] (y2) at (6,1.5) {};
    \node[label=below:$\bot$,fill=red] (bot) at (5.0,0) {};
    \draw[thick,-] (bot)--(y2);
    \draw[thick,-] (bot)--(x2);
    \draw[thick,-] (a)--(y2);
    \draw[thick,-] (a)--(x2);
    \draw[thick,-] (b)--(x1);
    \draw[thick,-] (b)--(y1);
    \draw[thick,-] (top)--(x1);
    \draw[thick,-] (top)--(y1);
    \draw[thick,-] (x1)--(x2);
    \draw[thick,-] (y1)--(y2);

    \end{footnotesize}
  \end{scope}

\begin{scope}[shift={(0,0)}]
  \begin{tiny}
  \tikzstyle{every node}=[inner sep=2pt,fill=white]
  
  \node at (1.5,-0.25) {$\bot$};
  \draw[thick,color=red,(-)] (0,0)--(3,0);
  \node at (1.5,6.25) {$\top$};
  \draw[thick,color=red,(-)] (0,6)--(3,6);

  \node at (1.5,1.75) {$a$};
  \draw[thick,color=red,(-)] (0.5,1.5)--(2.5,1.5);

  \node at (1.5,4.75) {$b$};
  \draw[thick,color=red,(-)] (0.5,4.5)--(2.5,4.5);
    
  \node at (1.25,2.55) {$x_2$};
  \draw[thick, -)] (0.0,2.3)--(2.5, 2.3);
  
  \node at (1.25,3.95) {$x_1$};
  \draw[thick, -)] (0.0,3.7)--(2.5, 3.7);
  
  \node at (1.75,1.05) {$y_2$};
  \draw[thick, (-] (0.5,0.75)--(3, 0.75);

  \node at (1.75,5.55) {$y_1$};
  \draw[thick, (-] (0.5,5.25)--(3, 5.25);

%  \draw[thick, (-] (0.25,1)--(3, 1);
%  \node (t) at (2.5,1) {$y_2$};
%  \draw[thick, (-] (0.25,3.5)--(3, 3.5);
%  \node (t) at (2.5,3.5) {$y_1$};
    
  \end{tiny}
\end{scope}

\begin{scope}[shift={(3.5,0)}, xscale=0.75]
  \begin{tiny}

  \draw[fill=light-gray, color=light-gray] (0,0) rectangle (3,6);
  
  \draw[thick,color=red,(-)] (0,0)--(3,0);
  \draw[thick,color=red,(-)] (0,6)--(3,6);

  \node at (1.5,3) {$\text{XOR}$};
    
  \node at (0.5,2.55) {$x_2$};
  \draw[very thick, -] (0,2.3)--(1, 2.3);

  \node at (0.5,3.95) {$x_1$};
  \draw[very thick, -] (0.0,3.7)--(1, 3.7);

  \node at (2.5,1) {$y_2$};
  \draw[very thick, -] (2,0.75)--(3, 0.75);
  
  \node at (2.5,5.5) {$y_1$};
  \draw[very thick, -] (2,5.25)--(3, 5.25);
  
  \draw[dotted, (-)] (0.2,0)--(0.2,1.5);
  \draw[dashed, (-)] (0.2,1.5)--(0.2,4.5);
  \draw[dotted, (-)] (0.2,4.5)--(0.2,6);

  \draw[dotted, (-)] (2.8,0)--(2.8,1.5);
  \draw[dashed, (-)] (2.8,1.5)--(2.8,4.5);
  \draw[dotted, (-)] (2.8,4.5)--(2.8,6);

  \end{tiny}
\end{scope}

\begin{scope}[shift={(6.75,0)}]
  \begin{tiny}
  \tikzstyle{every node}=[inner sep=2pt,fill=white]
  
  \node at (1.5,-0.25) {$\bot$};
  \draw[thick,color=red,(-)] (0,0)--(3,0);
  \node at (1.5,6.25) {$\top$};
  \draw[thick,color=red,(-)] (0,6)--(3,6);

  \node at (1.5,1.75) {$a$};
  \draw[thick,color=red,(-)] (0.5,1.5)--(2.5,1.5);

  \node at (1.5,4.75) {$b$};
  \draw[thick,color=red,(-)] (0.5,4.5)--(2.5,4.5);
      
  \node at (1.25,5.5) {$x_1$};
  \draw[thick, -)] (0.0,5.25)--(2.5, 5.25);

  \node at (1.25,1) {$x_2$};
  \draw[thick, -)] (0.0,0.75)--(2.5, 0.75);
  
  \node at (1.75,3.95) {$y_1$};
  \draw[thick, (-] (0.5,3.7)--(3, 3.7);

  \node at (1.75,2.55) {$y_2$};
  \draw[thick, (-] (0.5,2.3)--(3, 2.3);

  \end{tiny}
\end{scope}

\begin{scope}[shift={(10.25,0)}, xscale=0.75]
  \begin{tiny}

  \draw[fill=light-gray, color=light-gray] (0,0) rectangle (3,6);
  
  \draw[thick,color=red,(-)] (0,0)--(3,0);
  \draw[thick,color=red,(-)] (0,6)--(3,6);

  \node at (1.5,3) {$\text{XOR}$};

  \node at (0.5,5.5) {$x_1$};
  \draw[very thick, -] (0.0,5.25)--(1, 5.25);

  \node at (0.5,0.75) {$x_2$};
  \draw[very thick, -] (0,0.5)--(1, 0.5);

  \node at (2.5,3.95) {$y_1$};
  \draw[very thick, -] (2,3.7)--(3, 3.7);

  \node at (2.5,2.55) {$y_2$};
  \draw[very thick, -] (2,2.3)--(3, 2.3);
  
  \draw[dotted, (-)] (0.2,0)--(0.2,1.5);
  \draw[dashed, (-)] (0.2,1.5)--(0.2,4.5);
  \draw[dotted, (-)] (0.2,4.5)--(0.2,6);

  \draw[dotted, (-)] (2.8,0)--(2.8,1.5);
  \draw[dashed, (-)] (2.8,1.5)--(2.8,4.5);
  \draw[dotted, (-)] (2.8,4.5)--(2.8,6);

  \end{tiny}
\end{scope}
\end{scope}
\end{tikzpicture}
\caption{The XOR gadget depicted by its two possible representations and their schemes.}
\label{fig:eps_xor}
\end{center}
\end{figure*}
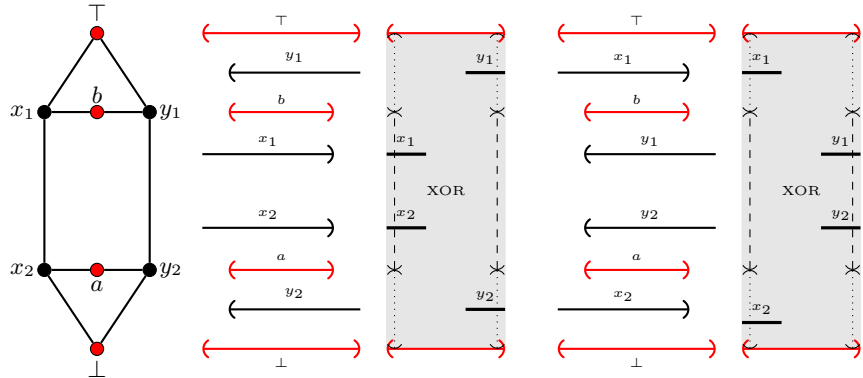

\noindent \textbf{XOR Gadget.}
  Figure~\ref{fig:eps_xor} presents a XOR gate.
  It checks that the inputs $x_1$ and $x_2$ have different boolean values.
  It also produces outputs $y_1$ and $y_2$.
  The partial representation is extendable if and only if $x_1 = \lnot x_2 = \lnot y_1 = y_2$.

  To see that, observe that the visibility gap between $b$ and $\top$ needs to be obstructed and $b$ has only two neighbors $x_1$ and $y_1$.
  Assume that $y_1$ blocks the visibility between $b$ and $\top$.
  Now $x_1$ needs to block the visibility between $b$ and the other bars.
  Thus, $x_1$ must be placed below $b$.
  The only other neighbor of $x_1$ is $x_2$ and thus, it needs to go above $a$.
  The last unrepresented vertex is $y_2$ and it needs to obstruct the visibility gap between $a$ and $\bot$.

  The other possibility is that $x_1$ blocks visibility between $b$ and $\top$.
  The analysis of this case is symmetric to the previous one and gives the second possible valuation of the variables.

  Finally, note that the visibility between every two bars does not change across the two representations and the corresponding edges of $G$ are well defined.

\begin{figure*}
\begin{center}
\begin{tikzpicture}[>=latex]
\definecolor{light-gray}{gray}{0.90}

\begin{scope}[xscale=0.7, yscale=0.7]

\begin{scope}[shift={(0,0)}]
  \begin{tiny}
  \draw[fill=light-gray, color=light-gray] (0,0) rectangle (2.95,6);

  \node at (1.5,3) {$\text{XOR}$};

  \node at (0.5,3.75) {$x_1$};
  \draw[very thick, -] (0.0,3.5)--(1, 3.5);

  \node at (0.5,2.25) {$x_2$};
  \draw[very thick, -] (0,2.5)--(1, 2.5);

  \node at (2.4,5.25) {$y_1$};
  \draw[very thick, -] (2,5.5)--(3, 5.5);

  \node at (2.4,0.75) {$y_2$};
  \draw[very thick, -] (2,0.5)--(3, 0.5);
  
  \draw[dotted, (-)] (0.2,0)--(0.2,1.5);
  \draw[dashed, (-)] (0.2,1.5)--(0.2,4.5);
  \draw[dotted, (-)] (0.2,4.5)--(0.2,6);

  \draw[dotted, (-)] (2.7,0)--(2.7,1.5);
  \draw[dashed, (-)] (2.7,1.5)--(2.7,4.5);
  \draw[dotted, (-)] (2.7,4.5)--(2.7,6);

  \end{tiny}
\end{scope}

\begin{scope}[shift={(3,3)}]
\draw[fill=light-gray, color=light-gray] (0.05,3) rectangle (3,-3);
  \draw[thick,color=red,(-)] (-3,-3)--(3,-3);
  \draw[thick,color=red,(-)] (-3,3)--(3,3);

  \begin{tiny}
  
  \draw[thick,color=red,(-)] (0,0)--(3,0);

  \node at (1.5,1.5) {$\overline{\text{NOT}}$};
  
  \draw[dotted, (-)] (0.2,2)--(0.2,3);
  \draw[dotted, (-)] (0.2,0.0)--(0.2,1);

  \draw[dotted, (-)] (2.8,2)--(2.8,3);
  \draw[dotted, (-)] (2.8,0.0)--(2.8,1);
  
  \node at (0.6,2.25) {$x$};
  \draw[very thick, -] (0.0,2.5)--(1.0, 2.5);
  
  \node at (2.5,0.75) {$y$};
  \draw[very thick, -] (2.0,0.5)--(3, 0.5);

  \end{tiny}
\end{scope}

\begin{scope}[shift={(3,0)}]

  \begin{tiny}

  \node at (1.5,1.5) {$\underline{\text{NOT}}$};
  
  \draw[dotted, (-)] (0.2,2)--(0.2,3);
  \draw[dotted, (-)] (0.2,0.0)--(0.2,1);

  \draw[dotted, (-)] (2.8,2)--(2.8,3);
  \draw[dotted, (-)] (2.8,0.0)--(2.8,1);
  
  \node at (0.6,0.75) {$x$};
  \draw[very thick, -] (0.0,0.5)--(1.0, 0.5);
  
  \node at (2.5,2.25) {$y$};
  \draw[very thick, -] (2.0,2.5)--(3, 2.5);

  \end{tiny}
\end{scope}

%check_xor_scheme
\begin{scope}[shift={(8,0)}]
  \begin{tiny}
  \draw[fill=light-gray, color=light-gray] (0,0) rectangle (3,6);
  
  \draw[thick,color=red,(-)] (0,0)--(3,0);
  \draw[thick,color=red,(-)] (0,6)--(3,6);
    
  \node at (2.5,3.75) {$y_1$};
  \draw[very thick, -] (2.0,3.5)--(3, 3.5);

  \node at (2.5,2.25) {$y_2$};
  \draw[very thick, -] (2,2.5)--(3, 2.5);

  \node at (1.5,3) {$\text{CXOR}$};

  \draw[very thick, -] (0,3.5)--(1, 3.5);
  \node at (0.5,3.75) {$x_1$};
  
  \draw[very thick, -] (0,2.5)--(1, 2.5);
  \node at (0.5,2.25) {$x_2$};
  
  \draw[dotted, (-)] (0.2,0)--(0.2,1.5);
  \draw[dashed, (-)] (0.2,1.5)--(0.2,4.5);
  \draw[dotted, (-)] (0.2,4.5)--(0.2,6);

  \draw[dotted, (-)] (2.8,0)--(2.8,1);
  \draw[dotted, (-)] (2.8,2)--(2.8,3);
  \draw[dotted, (-)] (2.8,3)--(2.8,4);
  \draw[dotted, (-)] (2.8,5)--(2.8,6);

  \end{tiny}
\end{scope}

%check_xor_scheme
\begin{scope}[shift={(12,0)}]
  \begin{tiny}
  \draw[fill=light-gray, color=light-gray] (0,0) rectangle (3,6);
  
  \draw[thick,color=red,(-)] (0,0)--(3,0);

  \draw[thick,color=red,(-)] (0,6)--(3,6);
    
  \node at (2.5,5.25) {$y_1$};
  \draw[very thick,-] (2.0,5.5)--(3, 5.5);

  \node at (2.5,0.75) {$y_2$};
  \draw[very thick, -] (2,0.5)--(3, 0.5);

  \node at (1.5,3) {$\text{CXOR}$};

  \node at (0.5,5.25) {$x_1$};
  \draw[very thick, -] (0,5.5)--(1, 5.5);

  \node at (0.5,0.75) {$x_2$};
  \draw[very thick, -] (0,0.5)--(1, 0.5);
  
  \draw[dotted, (-)] (0.2,0)--(0.2,1.5);
  \draw[dashed, (-)] (0.2,1.5)--(0.2,4.5);
  \draw[dotted, (-)] (0.2,4.5)--(0.2,6);

  \draw[dotted, (-)] (2.8,0)--(2.8,1);
  \draw[dotted, (-)] (2.8,2)--(2.8,3);
  \draw[dotted, (-)] (2.8,3)--(2.8,4);
  \draw[dotted, (-)] (2.8,5)--(2.8,6);

  \end{tiny}
\end{scope}

\end{scope}

\end{tikzpicture}
\caption{The CXOR gadget depicted by the schemes for its two possible representations.}
\label{fig:eps_cxor}
\end{center}
\end{figure*}
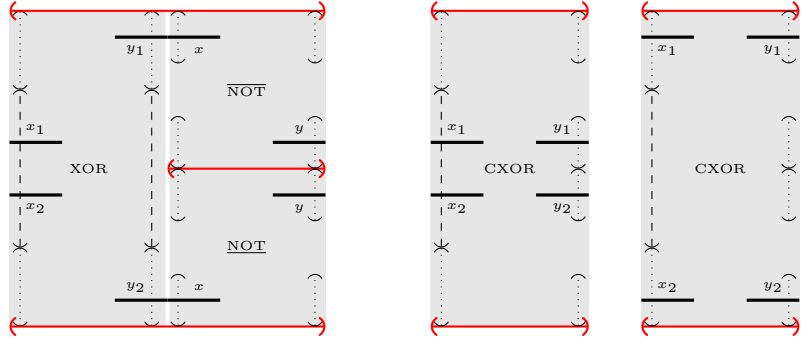

\noindent \textbf{CXOR Gadget.}
  Figure~\ref{fig:eps_cxor} presents a CXOR circuit.
  It checks that the inputs $x_1$ and $x_2$ have different boolean values and produces copies $y_1$ and $y_2$ of the inputs.
  The CXOR construction combines the XOR gate and two NOT gates in order to obtain a circuit that checks that $x_1 = \lnot x_2$, $y_1 = x_1$ and $y_2 = x_2$.

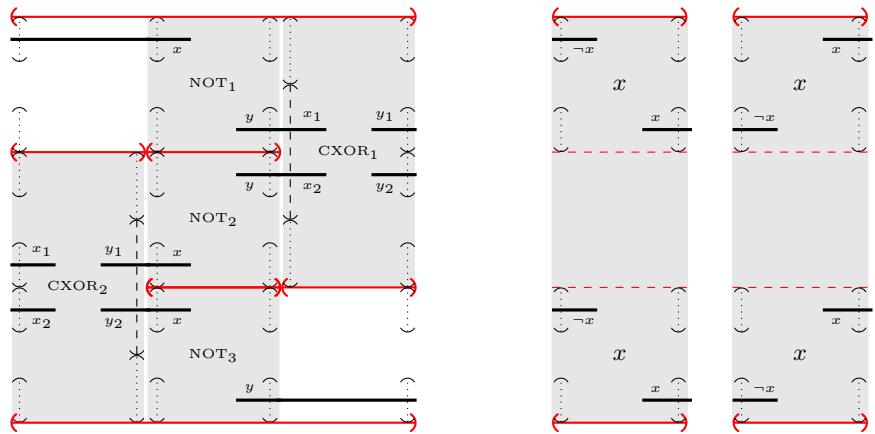
\begin{figure*}
\begin{center}
\begin{tikzpicture}[>=latex]
\definecolor{light-gray}{gray}{0.90}

\begin{scope}[xscale=0.6, yscale=0.6]

\begin{scope}[shift={(0,6)}]

  \begin{tiny}
  \draw[fill=light-gray, color=light-gray] (0.05,3) rectangle (2.95,-6);
  
  \draw[thick,color=red,(-)] (0,0)--(3,0);
  
  \node at (1.5,1.5) {$\text{NOT}_1$};
  
  \draw[dotted, (-)] (0.25,2)--(0.25,3);
  \draw[dotted, (-)] (0.25,0.0)--(0.25,1);

  \draw[dotted, (-)] (2.75,2)--(2.75,3);
  \draw[dotted, (-)] (2.75,0.0)--(2.75,1);
  
  \node at (0.7,2.25) {$x$};
  \draw[very thick, -] (0.0,2.5)--(1.0, 2.5);
  
  \node at (2.3,0.75) {$y$};
  \draw[very thick, -] (2.0,0.5)--(3, 0.5);

  \end{tiny}
\end{scope}

\begin{scope}[shift={(0,3)}]
  \begin{tiny}
  
  \draw[thick,color=red,(-)] (0,0)--(3,0);
  
  \node at (1.5,1.5) {$\text{NOT}_2$};
  
  \draw[dotted, (-)] (0.25,2)--(0.25,3);
  \draw[dotted, (-)] (0.25,0.0)--(0.25,1);

  \draw[dotted, (-)] (2.75,2)--(2.75,3);
  \draw[dotted, (-)] (2.75,0.0)--(2.75,1);
  
  \node at (0.7,0.75) {$x$};
  \draw[very thick, -] (0.0,0.5)--(1.0, 0.5);
  
  \node at (2.3,2.25) {$y$};
  \draw[very thick, -] (2.0,2.5)--(3, 2.5);

  \end{tiny}
\end{scope}

\begin{scope}[shift={(0,0)}]
  \begin{tiny}
  
  \draw[thick,color=red,(-)] (0,3)--(3,3);
  
  \node at (1.5,1.5) {$\text{NOT}_3$};
  
  \draw[dotted, (-)] (0.25,2)--(0.25,3);
  \draw[dotted, (-)] (0.25,0.0)--(0.25,1);

  \draw[dotted, (-)] (2.75,2)--(2.75,3);
  \draw[dotted, (-)] (2.75,0.0)--(2.75,1);
  
  \node at (0.7,2.25) {$x$};
  \draw[very thick, -] (0.0,2.5)--(1.0, 2.5);
  
  \node at (2.3,0.75) {$y$};
  \draw[very thick, -] (2.0,0.5)--(3, 0.5);

  \end{tiny}

\end{scope}

%XOR1
\begin{scope}[shift={(3,3)}]
  \begin{tiny}
  \draw[fill=light-gray, color=light-gray] (0.05,0) rectangle (2.95,6);
  
  \draw[thick,color=red,(-)] (0,0)--(3,0);    

  \node at (1.5,3) {$\text{CXOR}_1$};
  
  \node at (2.3,3.8) {$y_1$};
  \draw[very thick, -] (2.0,3.5)--(3, 3.5);

  \node at (2.3,2.2) {$y_2$};
  \draw[very thick, -] (2,2.5)--(3, 2.5);

  \draw[very thick, -] (-0.1,3.5)--(1, 3.5);
  \node at (0.7,3.8) {$x_1$};

  \draw[very thick, -] (-0.1,2.5)--(1, 2.5);
  \node at (0.7,2.2) {$x_2$};
  
  \draw[dotted, (-)] (0.2,0)--(0.2,1.5);
  \draw[dashed, (-)] (0.2,1.5)--(0.2,4.5);
  \draw[dotted, (-)] (0.2,4.5)--(0.2,6);

  \draw[dotted, (-)] (2.8,0)--(2.8,1);
  \draw[dotted, (-)] (2.8,2)--(2.8,3);
  \draw[dotted, (-)] (2.8,3)--(2.8,4);
  \draw[dotted, (-)] (2.8,5)--(2.8,6);

  \end{tiny}
\end{scope}

%XOR2
\begin{scope}[shift={(-3,0)}]
  \begin{tiny}
  \draw[fill=light-gray, color=light-gray] (0.05,0) rectangle (2.95,6);
  
  \draw[thick,color=red,(-)] (0,6)--(3,6);    

  \node at (1.5,3) {$\text{CXOR}_2$};
  
  \node at (2.3,3.8) {$y_1$};
  \draw[very thick, -] (2.0,3.5)--(3.1, 3.5);

  \node at (2.3,2.2) {$y_2$};
  \draw[very thick, -] (2,2.5)--(3.1, 2.5);

  \draw[very thick, -] (0,3.5)--(1, 3.5);
  \node at (0.7,3.8) {$x_1$};

  \draw[very thick, -] (0,2.5)--(1, 2.5);
  \node at (0.7,2.2) {$x_2$};
  
  \draw[dotted, (-)] (2.8,0)--(2.8,1.5);
  \draw[dashed, (-)] (2.8,1.5)--(2.8,4.5);
  \draw[dotted, (-)] (2.8,4.5)--(2.8,6);

  \draw[dotted, (-)] (0.2,0)--(0.2,1);
  \draw[dotted, (-)] (0.2,2)--(0.2,3);
  \draw[dotted, (-)] (0.2,3)--(0.2,4);
  \draw[dotted, (-)] (0.2,5)--(0.2,6);

  \end{tiny}
\end{scope}

%empty
\begin{scope}[shift={(-3,6)}]
  \draw[thick,color=red,(-)] (0,3)--(9,3);
  \draw[very thick, -] (0,2.5)--(3.1, 2.5);

  \draw[dotted, (-)] (0.2,2)--(0.2,3);
  \draw[dotted, (-)] (0.2,0.0)--(0.2,1);

\end{scope}
%empty
\begin{scope}[shift={(3,0)}]
  \draw[thick,color=red,(-)] (-6,0)--(3,0);
  \draw[very thick, -] (-0.1,0.5)--(3, 0.5);

  \draw[dotted, (-)] (2.8,2)--(2.8,3);
  \draw[dotted, (-)] (2.8,0.0)--(2.8,1);

\end{scope}

%x=FALSE
\begin{scope}[shift={(9,0)}]
\begin{scope}[shift={(0,6)}]
  \begin{tiny}  
  \draw[fill=light-gray, color=light-gray] (0,3) rectangle (3,-6);  
  \draw[thick,color=red,(-)] (0,3)--(3,3); 
  
  \draw[dotted, (-)] (0.2,2)--(0.2,3);
  \draw[dotted, (-)] (0.2,0.0)--(0.2,1);

  \draw[dotted, (-)] (2.8,2)--(2.8,3);
  \draw[dotted, (-)] (2.8,0.0)--(2.8,1);
  
  \begin{footnotesize}
  \node at (1.5,1.5) {$x$};
  \end{footnotesize}
  
  \node at (0.7,2.25) {$\lnot x$};
  \draw[very thick, -] (0.0,2.5)--(1.0, 2.5);
  
  \node at (2.3,0.75) {$x$};
  \draw[very thick, -] (2.0,0.5)--(3.1, 0.5);

  \end{tiny}
\end{scope}

\begin{scope}[shift={(0,3)}]

  \begin{tiny}
  \tikzstyle{every node}=[inner sep=2pt,fill=white]
  
    \draw[dashed,color=red,-] (0,3)--(3,3);
    \draw[dashed,color=red,-] (0,0)--(3,0);

  \end{tiny}
\end{scope}

\begin{scope}[shift={(0,0)}]
  \begin{tiny}  
  \draw[thick,color=red,(-)] (0,0)--(3,0); 
  
  \draw[dotted, (-)] (0.2,2)--(0.2,3);
  \draw[dotted, (-)] (0.2,0.0)--(0.2,1);

  \draw[dotted, (-)] (2.8,2)--(2.8,3);
  \draw[dotted, (-)] (2.8,0.0)--(2.8,1);
  
  \begin{footnotesize}
  \node at (1.5,1.5) {$x$};
  \end{footnotesize}
  
  \node at (0.7,2.25) {$\lnot x$};
  \draw[very thick, -] (0.0,2.5)--(1.0, 2.5);
  
  \node at (2.3,0.75) {$x$};
  \draw[very thick, -] (2.0,0.5)--(3.1, 0.5);

  \end{tiny}
\end{scope}

\end{scope}

%x=TRUE
\begin{scope}[shift={(13,0)}]
\begin{scope}[shift={(0,6)}]
  \begin{tiny}  
  \draw[fill=light-gray, color=light-gray] (0,3) rectangle (3,-6);  
  \draw[thick,color=red,(-)] (0,3)--(3,3); 
  
  \draw[dotted, (-)] (0.2,2)--(0.2,3);
  \draw[dotted, (-)] (0.2,0.0)--(0.2,1);

  \draw[dotted, (-)] (2.8,2)--(2.8,3);
  \draw[dotted, (-)] (2.8,0.0)--(2.8,1);
  
  \begin{footnotesize}
  \node at (1.5,1.5) {$x$};
  \end{footnotesize}
  
  \node at (0.7,0.75) {$\lnot x$};
  \draw[very thick, -] (0.0,0.5)--(1.0, 0.5);
  
  \node at (2.3,2.25) {$x$};
  \draw[very thick, -] (2.0,2.5)--(3.1, 2.5);

  \end{tiny}
\end{scope}

\begin{scope}[shift={(0,3)}]

  \begin{tiny}
  \tikzstyle{every node}=[inner sep=2pt,fill=white]
  
    \draw[dashed,color=red,-] (0,3)--(3,3);
    \draw[dashed,color=red,-] (0,0)--(3,0);

  \end{tiny}
\end{scope}

\begin{scope}[shift={(0,0)}]
  \begin{tiny}  
  \draw[thick,color=red,(-)] (0,0)--(3,0); 
  
  \draw[dotted, (-)] (0.2,2)--(0.2,3);
  \draw[dotted, (-)] (0.2,0.0)--(0.2,1);

  \draw[dotted, (-)] (2.8,2)--(2.8,3);
  \draw[dotted, (-)] (2.8,0.0)--(2.8,1);
  
  \begin{footnotesize}
  \node at (1.5,1.5) {$x$};
  \end{footnotesize}
  
  \node at (0.7,0.75) {$\lnot x$};
  \draw[very thick, -] (0.0,0.5)--(1.0, 0.5);
  
  \node at (2.3,2.25) {$x$};
  \draw[very thick, -] (2.0,2.5)--(3.1, 2.5);

  \end{tiny}
\end{scope}

\end{scope}

\end{scope}
\end{tikzpicture}
\caption{On the left: the gadget for the variable $x$ with two output slots 
and one of its possible representation for the negative value of $x$. 
On the right: schemes of the gadget for the negative and positive value of $x$, respectively.}
\label{fig:eps_var}
\end{center}
\end{figure*}

\noindent \textbf{Variable Gadget.}
  Using NOT gates and CXOR circuits it is easy to construct a variable gadget.
  Figure~\ref{fig:eps_var} presents a gadget that gives two wires with value $x$ to the right side, and two wires with value $\lnot x$ to the left side.
  If we need $k$ copies of a variable, we simply stack $2k-1$ NOT gates one on another, and add CXOR gates to check the consistency of the outputs produced by every second NOT gate.

  All we need to finish the construction of our building blocks is a clause gadget that checks that at least one of three wires connected to it transmits a positive value.

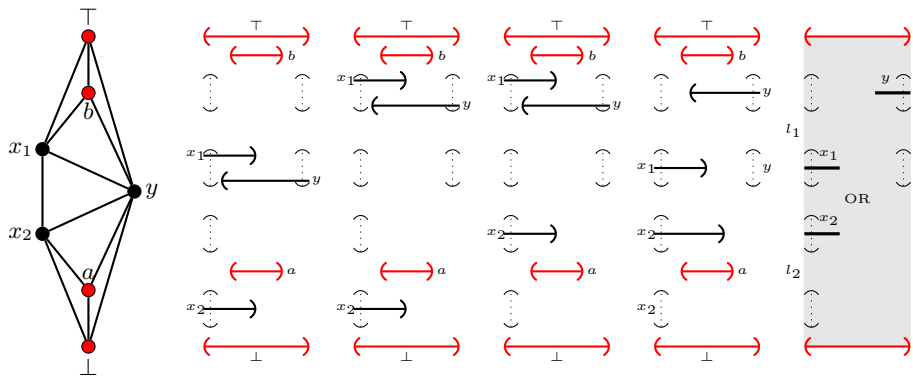
\begin{figure*}
\begin{center}
\begin{tikzpicture}[>=latex]
\definecolor{light-gray}{gray}{0.90}

\begin{scope}[xscale=0.47, yscale=0.5]

   \begin{scope}[shift={(-8.25,0)}]
    \tikzstyle{every node}=[circle,minimum size=5pt,inner sep=0pt,draw,fill]
    \begin{footnotesize}
    \node[label=above:$\top$,fill=red] (top) at (5.0,8.25) {};
    \node[label=below:$b$,fill=red] (b) at (5,6.75) {};
    \node[label=left:$x_1$,fill=black] (x1) at (3.7,5.25) {};
    \node[label=right:$y$,fill=black] (y) at (6.3,4.125) {};
    \node[label=left:$x_2$,fill=black] (x2) at (3.7,3) {};
    \node[label=above:$a$,fill=red] (a) at (5,1.5) {};
    \node[label=below:$\bot$,fill=red] (bot) at (5.0,0) {};
    \draw[thick,-] (bot)--(a);
    \draw[thick,-] (bot)--(x2);
    \draw[thick,-] (bot)--(y);
    \draw[thick,-] (a)--(y);
    \draw[thick,-] (y)--(x1);
    \draw[thick,-] (x1)--(x2);
    \draw[thick,-] (y)--(x2);
    \draw[thick,-] (a)--(x2);
    \draw[thick,-] (b)--(y);
    \draw[thick,-] (b)--(x1);
    \draw[thick,-] (top)--(x1);
    \draw[thick,-] (top)--(y);
    \draw[thick,-] (top)--(b);

    \end{footnotesize}
  \end{scope}

\begin{scope}[shift={(0.0,0)}]
\begin{tiny}
\begin{scope}[shift={(0,-1.25)}]
  
  \node at (1.5,9.8) {$\top$};
  \draw[thick,color=red,(-)] (0,9.5)--(3,9.5);

  \node at (2.5,9) {$b$};
  \draw[thick,color=red,(-)] (0.75,9)--(2.25,9);

%y1

%  \node at (3.2,8) {$y$};
  \draw[dotted, (-)] (2.8,7.5)--(2.8,8.5);

%y0
  \node at (3.2,5.66) {$y$};
  \draw[thick, (-] (0.5,5.66)--(3, 5.66);
  \draw[dotted, (-)] (2.8,5.5)--(2.8,6.5);
  
%x_11
%  \node at (-0.2,8) {$x_1$};
  \draw[dotted, (-)] (0.2,7.5)--(0.2,8.5);

%x_10
  \node at (-0.2,6.33) {$x_1$};
  \draw[thick, -)] (0,6.33)--(1.5, 6.33);
  \draw[dotted, (-)] (0.2,5.5)--(0.2,6.5);
\end{scope}  
%x_20
%  \node at (-0.2,3) {$x_2$};
  \draw[dotted, (-)] (0.2,2.5)--(0.2,3.5);

  \draw[thick,color=red,(-)] (0.75,2)--(2.25,2);
  \node at (2.5,2) {$a$};

%x_21
  \node at (-0.2,1) {$x_2$};
  \draw[thick, -)] (0,1)--(1.5, 1);
  \draw[dotted, (-)] (0.2,0.5)--(0.2,1.5);

  \node at (1.5,-0.25) {$\bot$};
  \draw[thick,color=red,(-)] (0,0)--(3,0);
      
  \end{tiny}
\end{scope}

%10
\begin{scope}[shift={(4.25,0)}]
\begin{tiny}
\begin{scope}[shift={(0,-1.25)}]
  
  \node at (1.5,9.8) {$\top$};
  \draw[thick,color=red,(-)] (0,9.5)--(3,9.5);

  \node at (2.5,9) {$b$};
  \draw[thick,color=red,(-)] (0.75,9)--(2.25,9);

%y1

  \node at (3.2,7.66) {$y$};
  \draw[thick, (-] (0.5,7.66)--(3, 7.66);
  \draw[dotted, (-)] (2.8,7.5)--(2.8,8.5);

%y0
%  \node at (3.2,6) {$y$};
  \draw[dotted, (-)] (2.8,5.5)--(2.8,6.5);
  
%x_11
  \node at (-0.2,8.33) {$x_1$};
  \draw[thick, -)] (0,8.33)--(1.5, 8.33);
  \draw[dotted, (-)] (0.2,7.5)--(0.2,8.5);

%x_10
%  \node at (-0.2,6) {$x_1$};
  \draw[dotted, (-)] (0.2,5.5)--(0.2,6.5);

\end{scope}
  
%x_20
%  \node at (-0.2,3) {$x_2$};
  \draw[dotted, (-)] (0.2,2.5)--(0.2,3.5);

  \draw[thick,color=red,(-)] (0.75,2)--(2.25,2);
  \node at (2.5,2) {$a$};

%x_21
  \node at (-0.2,1) {$x_2$};
  \draw[thick, -)] (0,1)--(1.5, 1);
  \draw[dotted, (-)] (0.2,0.5)--(0.2,1.5);

  \node at (1.5,-0.25) {$\bot$};
  \draw[thick,color=red,(-)] (0,0)--(3,0);
      
  \end{tiny}
\end{scope}

%11
\begin{scope}[shift={(8.5,0)}]
\begin{tiny}
\begin{scope}[shift={(0,-1.25)}]
  
  \node at (1.5,9.8) {$\top$};
  \draw[thick,color=red,(-)] (0,9.5)--(3,9.5);

  \node at (2.5,9) {$b$};
  \draw[thick,color=red,(-)] (0.75,9)--(2.25,9);

%y1

  \node at (3.2,7.66) {$y$};
  \draw[thick, (-] (0.5,7.66)--(3, 7.66);
  \draw[dotted, (-)] (2.8,7.5)--(2.8,8.5);

%y0
%  \node at (3.2,6) {$y$};
  \draw[dotted, (-)] (2.8,5.5)--(2.8,6.5);
  
%x_11
  \node at (-0.2,8.33) {$x_1$};
  \draw[thick, -)] (0,8.33)--(1.5, 8.33);
  \draw[dotted, (-)] (0.2,7.5)--(0.2,8.5);

%x_10
%  \node at (-0.2,6) {$x_1$};
  \draw[dotted, (-)] (0.2,5.5)--(0.2,6.5);
\end{scope}

%x_20
  \node at (-0.2,3) {$x_2$};
  \draw[thick, -)] (0,3)--(1.5, 3);
  \draw[dotted, (-)] (0.2,2.5)--(0.2,3.5);

  \draw[thick,color=red,(-)] (0.75,2)--(2.25,2);
  \node at (2.5,2) {$a$};

%x_21
  \draw[dotted, (-)] (0.2,0.5)--(0.2,1.5);

  \node at (1.5,-0.25) {$\bot$};
  \draw[thick,color=red,(-)] (0,0)--(3,0);
      
  \end{tiny}
\end{scope}

%01
\begin{scope}[shift={(12.75,0)}]
\begin{tiny}
\begin{scope}[shift={(0,-1.25)}]
  
  \node at (1.5,9.8) {$\top$};
  \draw[thick,color=red,(-)] (0,9.5)--(3,9.5);

  \node at (2.5,9) {$b$};
  \draw[thick,color=red,(-)] (0.75,9)--(2.25,9);

%y1

  \node at (3.2,8) {$y$};
  \draw[thick, (-] (1,8)--(3, 8);
  \draw[dotted, (-)] (2.8,7.5)--(2.8,8.5);

%y0
  \node at (3.2,6) {$y$};
  \draw[dotted, (-)] (2.8,5.5)--(2.8,6.5);
  
%x_11
%  \node at (-0.2,8) {$x_1$};
  \draw[dotted, (-)] (0.2,7.5)--(0.2,8.5);

%x_10
  \node at (-0.2,6) {$x_1$};
  \draw[thick, -)] (0,6)--(1.5, 6);
  \draw[dotted, (-)] (0.2,5.5)--(0.2,6.5);

\end{scope}
  
%x_20
  \node at (-0.2,3) {$x_2$};
  \draw[thick, -)] (0,3)--(2, 3);
  \draw[dotted, (-)] (0.2,2.5)--(0.2,3.5);

  \draw[thick,color=red,(-)] (0.75,2)--(2.25,2);
  \node at (2.5,2) {$a$};

%x_21
  \node at (-0.2,1) {$x_2$};
  \draw[dotted, (-)] (0.2,0.5)--(0.2,1.5);

  \node at (1.5,-0.25) {$\bot$};
  \draw[thick,color=red,(-)] (0,0)--(3,0);
      
  \end{tiny}
\end{scope}

\begin{scope}[shift={(17,0)}]
\begin{tiny}
\begin{scope}[shift={(0,-1.25)}]
  \draw[fill=light-gray, color=light-gray] (0,1.25) rectangle (3,9.5);
  
  \draw[thick,color=red,(-)] (0,9.5)--(3,9.5);

%y1

  \node at (2.3,8.35) {$y$};
  \draw[very thick, -] (2,8)--(3, 8);
  \draw[dotted, (-)] (2.8,7.5)--(2.8,8.5);
  %\node at (4,7) {$l_1 \vee l_2$};

%y0
%  \node at (3.2,6) {$y$};
  \draw[dotted, (-)] (2.8,5.5)--(2.8,6.5);
  
%x_11
%  \node at (-0.2,8) {$x_1$};
  \draw[dotted, (-)] (0.2,7.5)--(0.2,8.5);
  \node at (-0.3,7) {$l_1$};

%x_10
  \node at (0.7,6.35) {$x_1$};
  \draw[very thick, -] (0,6)--(1, 6);
  \draw[dotted, (-)] (0.2,5.5)--(0.2,6.5);

\end{scope}
  \node at (1.5,3.95) {$\text{OR}$};
  
%x_20
  \node at (0.7,3.35) {$x_2$};
  \draw[very thick, -] (0,3)--(1, 3);
  \draw[dotted, (-)] (0.2,2.5)--(0.2,3.5);

%x_21
  \node at (-0.3,2) {$l_2$};
  \draw[dotted, (-)] (0.2,0.5)--(0.2,1.5);
  
  \draw[thick,color=red,(-)] (0,0)--(3,0);
      
  \end{tiny}
\end{scope}

\end{scope}
\end{tikzpicture}
\caption{The OR Gadget.}
\label{fig:eps_or}
\end{center}
\end{figure*}

\noindent \textbf{OR Gadget and Clause Gadget.}
  Figure~\ref{fig:eps_or} presents an OR gate that has two inputs $x_1$ and $x_2$ and one output $y$.
  The output value can be positive only if at least one of the inputs is positive.
  Case analysis shows that in each of these three scenarios $y$ can be represented in the higher of its regions.
  However, if both $x_1$ and $x_2$ are negative then $y$ must be represented in the lower of its regions.
  Also, the visibility between vertices does not change across the scenarios.

\begin{figure*}
\begin{center}
\begin{tikzpicture}[>=latex]
\definecolor{light-gray}{gray}{0.90}

\begin{scope}[xscale=0.7, yscale=0.7]

\begin{scope}[shift={(0,0)}]
  \begin{tiny}

%OR1
  \draw[fill=light-gray, color=light-gray] (0,4) rectangle (2.95,11);

%l1l_2-out
%  \node at (2.0,9.25) {$l_1 \vee l_2$};
  \draw[dotted, (-)] (2.7,9.5)--(2.7,10.5);
  \draw[dotted, (-)] (2.7,8)--(2.7,9);
  \node at (2.25,8.75) {$y$};

  \node at (1.5,7.5) {$\text{OR}_1$};
%l1
  \node at (-0.2,9.25) {$l_1$};
  \draw[dotted, (-)] (0.2,9.5)--(0.2,10.5);
  \draw[dotted, (-)] (0.2,8)--(0.2,9);
  \node at (0.75,8.75) {$x_1$};
  \draw[very thick, -] (0,8.5)--(1, 8.5);

%l2
  \node at (-0.2,5.75) {$l_2$};
  \draw[dotted, (-)] (0.2,6)--(0.2,7);
  \draw[dotted, (-)] (0.2,4.5)--(0.2,5.5);
  \node at (0.75,5.25) {$x_2$};
  \draw[very thick, -] (0,5)--(1, 5);
  \draw[ thick, red, -] (0,4)--(3,4);

%OR1-END
%l3

  \node at (-0.2,1.75) {$l_3$};
  
  \draw[dotted, (-)] (0.2,2.0)--(0.2,3);
  \draw[dotted, (-)] (0.2,0.5)--(0.2,1.5);
  
%OR2
%$l_1 l2 in

  \draw[fill=light-gray, color=light-gray] (3.05,0) rectangle (6,11);
%l1l2l3-out
  \node at (5.25,10.25) {$y$};
  \draw[very thick, -)] (5,10)--(7.5, 10);
  \draw[dotted, (-)] (5.8,9.5)--(5.8,10.5);
  \draw[dotted, (-)] (5.8,8)--(5.8,9);

% l1l2-in 
  \node at (3.75,8.75) {$x_1$};
  \draw[very thick, -] (2,8.5)--(4, 8.5);
  \draw[dotted, (-)] (3.2,9.5)--(3.2,10.5);
  \draw[dotted, (-)] (3.2,8)--(3.2,9);
  \node at (4.5,5.5) {$\text{OR}_2$};

%l3
  \draw[dotted, (-)] (3.3,2.0)--(3.3,3);
  \draw[dotted, (-)] (3.3,0.5)--(3.3,1.5);
  \node at (3.75,2.75) {$x_2$};
  \draw[very thick, -] (0,2.5)--(4, 2.5);

%OR2-END

  \draw[ thick, red, (-)] (6.5,9.5)--(7.5, 9.5);
  \node at (7.7,9.5) {$b$};
  \draw[ thick, red, (-)] (6.5,10.5)--(7.5, 10.5);
  \node at (7.7,10.5) {$a$};

  \draw[ thick, red, (-)] (0,0)--(8, 0);
  \draw[ thick, red, (-)] (0,11)--(8, 11);
    
  \end{tiny}
\end{scope}

\begin{scope}[shift={(12,0)}]
  \begin{tiny}

%OR1
  \draw[fill=light-gray, color=light-gray] (0,0) rectangle (3,11);

  \node at (1.65,5.7) {$l_1 \vee l_2 \vee l_3$};
%l1
  \node at (-0.2,9.25) {$l_1$};
  \draw[dotted, (-)] (0.2,9.5)--(0.2,10.5);
  \draw[dotted, (-)] (0.2,8)--(0.2,9);
  \draw[very thick, -] (0,8.5)--(1, 8.5);

%l2
  \node at (-0.2,5.75) {$l_2$};
  \draw[dotted, (-)] (0.2,6)--(0.2,7);
  \draw[dotted, (-)] (0.2,4.5)--(0.2,5.5);
  \draw[very thick, -] (0,5)--(1, 5);

%l3
  \node at (-0.2,1.75) {$l_3$};
  
  \draw[dotted, (-)] (0.2,2.0)--(0.2,3);
  \draw[dotted, (-)] (0.2,0.5)--(0.2,1.5);
  \draw[very thick, -] (0,2.5)--(1, 2.5);
  \draw[ thick, red, (-)] (0,0)--(3, 0);
  \draw[ thick, red, (-)] (0,11)--(3,11);
\end{tiny}  
\end{scope}
\end{scope}

\end{tikzpicture}
\caption{The construction for a clause $l_1 \vee l_2 \vee l_3$ and its representation for the negative value of $l_1$ and $l_2$ and the positive value of $l_3$.}
\label{fig:eps_clause}
\end{center}
\end{figure*}
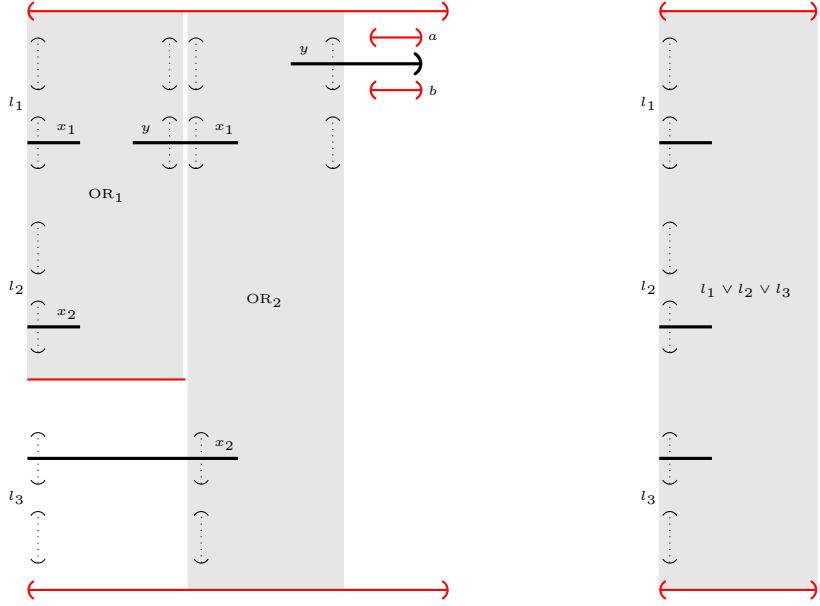

  Combining two OR gates and two bars that ensure that the output of the second gate is positive, we get a clause gadget presented in Figure~\ref{fig:eps_clause}.
%  Figure~\ref{fig:eps_sat} shows an example construction of the whole circuit.

% \input fig_eps_sat.tex

  Given an instance $\phi$ of $\pmthreesat$ (together with a rectilinear planar representation of $\phi$),
  we show how we construct the graph $G$ with a partial representation $\psi'$.
  We rotate the rectilinear representation by $90$ degrees.
  Now, we replace vertical segments representing variables of $\phi$ with variable gadgets and vertical segments representing clauses of $\phi$ with clause gadgets.
  Finally, for each occurrence of variable $x$ in clause $C$, we replace the horizontal connection between the segment of $x$ and the segment of $C$ by a wire connecting the appropriate gadgets.
  The properties of our gadgets assert that $\phi$ is satisfiable iff $\psi'$ is extendable to a bar visibility representation of $G$.
\end{proof}

\subsection{Grid representations}
\label{app:hardness_grid}

In this section we consider the following problem:

\smallskip
\noindent \underline{Grid Bar Visibility Representation Extension for Directed Graphs}:\smallskip\\
\noindent{\bf Input:}
 $(G, \psi')$; $G$ is a digraph; $\psi'$ is a map assigning bars to a $V' \subseteq V(G)$.\\
\noindent{\bf Question:}
Does $G$ admit a grid bar visibility representation $\psi$ with $\psi|V' = \psi'$?\smallskip\\
In what follows we show that the above problem is \NP-complete.
The proof is generic and can be easily modified to work for other grid representations including undirected, rectangular directed/undirected, and even other models of visibility.

\begin{theorem}\label{th:var-emb-hardness-grid}
Grid Bar Visibility Representation Extension problem is \NP-complete.
\end{theorem}

\begin{proof}
We use the \threepart{} problem to show \NP-hardness.
Consider an instance $w,$ $a_1,$ $\ldots,$ $a_{3m}$ of the \threepart{} problem.
Let $W = \sum^{3m}_{j=1} a_j$, \ie{}, $W = mw$.
From this instance of \threepart{}, we create a graph $G$ and a partial representation $\psi'$ that assigns bars to a subset $V'$ of $V(G)$.
These are constructed as follows and depicted in Figure~\ref{fig:3part}.

\begin{figure*}
\begin{tikzpicture}[>=latex]
\definecolor{light-gray}{gray}{0.60}
%\tikzstyle{every node}=[circle,fill=black]
\begin{scope}[xscale=0.47, yscale=0.5]
\begin{tiny}
    {\thinmuskip=0.5mu
      \medmuskip=0.5mu plus 0.5mu minus 0.5mu
      \thickmuskip=0.5mu plus 0.5mu minus 0.5mu

\begin{scope}[shift={(0,0)}]%11x4
  \node (s) at (4,0) {$s$};
  \node (t) at (4,4) {$t$};
  \node (u0) at (0,2) {$u_0$};
  \node (u1) at (1,2) {$u_1$};
  \node (udots) at (2,2) {$\cdots$};
  \node (um) at (3,2) {$u_m$};
  \draw (s) edge[->] (u0);
  \draw (u0) edge[->] (t);
  \draw (s) edge[->] (u1);
  \draw (u1) edge[->] (t);
  \draw (s) edge[->] (um);
  \draw (um) edge[->] (t);
  \node[inner sep=0pt,outer sep=0pt,minimum size=0pt] (h1l) at (4.25,2) {};
  \node[inner sep=0pt,outer sep=0pt,minimum size=0pt] (h1r) at (5.75,2) {};
  \node (h1t) at (5,3) {$t_1$};
  \node (h1n) at (5,2) {$H_1$};
  \node (h1b) at (5,1) {$s_1$};
  \draw (h1l) -- (h1t) -- (h1r) -- (h1b) -- (h1l);
  \draw (s) edge[->] (h1b);
  \draw (h1t) edge[->] (t);
  \node[inner sep=0pt,outer sep=0pt,minimum size=0pt] (h2l) at (6.25,2) {};
  \node[inner sep=0pt,outer sep=0pt,minimum size=0pt] (h2r) at (7.75,2) {};
  \node (h2t) at (7,3) {$t_2$};
  \node (h2n) at (7,2) {$H_2$};
  \node (h2b) at (7,1) {$s_2$};
  \draw (h2l) -- (h2t) -- (h2r) -- (h2b) -- (h2l);
  \draw (s) edge[->] (h2b);
  \draw (h2t) edge[->] (t);
  \node (hdots) at (8.5,2) {$\cdots$};
  \node[inner sep=0pt,outer sep=0pt,minimum size=0pt] (hml) at (9.25,2) {};
  \node[inner sep=0pt,outer sep=0pt,minimum size=0pt] (hmr) at (10.75,2) {};
  \node (hmt) at (10,3) {$t_{3m}$};
  \node (hmn) at (10,2) {$H_{3m}$};
  \node (hmb) at (10,1) {$s_{3m}$};
  \draw (hml) -- (hmt) -- (hmr) -- (hmb) -- (hml);
  \draw (s) edge[->] (hmb);
  \draw (hmt) edge[->] (t);
\end{scope}

\begin{scope}[shift={(3.5,-6)}]%4x4
  \node[] (s) at (2,0) {$s_i$};
  \node[] (t) at (2,4) {$t_i$};
  \node[circle,fill=black,inner sep=1pt,outer sep=1pt] (u0) at (0,2) {};
  \node[circle,fill=black,inner sep=1pt,outer sep=1pt] (u1) at (1,2) {};
  \node (udots) at (2,2) {$\cdots$};
  \node[circle,fill=black,inner sep=1pt,outer sep=1pt] (um) at (3,2) {};
  \node[circle,fill=black,inner sep=1pt,outer sep=1pt] (un) at (4,2) {};
  \draw (s) edge[->] (u0);
  \draw (u0) edge[->] (t);
  \draw (s) edge[->] (u1);
  \draw (u1) edge[->] (t);
  \draw (s) edge[->] (um);
  \draw (um) edge[->] (t);
  \draw (s) edge[->] (un);
  \draw (un) edge[->] (t);
\end{scope}

\begin{scope}[shift={(13,0)}]%12x4
  \begin{scope}[color=light-gray]
  \draw (-0.5,0)--(12,0);
  \draw (-0.5,1)--(12,1);
  \draw (-0.5,2)--(12,2);
  \draw (-0.5,3)--(12,3);
  \draw (-0.5,4)--(12,4);
  \node at (-0.75,0) {$0$};
  \node at (-0.75,1) {$1$};
  \node at (-0.75,2) {$2$};
  \node at (-0.75,3) {$3$};
  \node at (-0.75,4) {$4$};

  \draw (0,-0.5)--(0,4.5);
  \draw (2.5,-0.5)--(2.5,4.5);
  \draw (4.5,-0.5)--(4.5,4.5);
  \draw (7,-0.5)--(7,4.5);
  \draw (9,-0.5)--(9,4.5);
  \draw (11.5,-0.5)--(11.5,4.5);
  \node at (0,-0.9) {$0$};
  \node at (2.5,-0.9) {$1$};
  \node at (4.5,-0.9) {$iw+i$};
  \node at (7,-0.9) {$iw+i+1$};
  \node at (9,-0.9) {$W+m$};
  \node at (11.5,-0.9) {$W+m+1$};
  \end{scope}

  \draw[thick, (-)] (0,0)--(11.5,0);
  \node at (5.75,0.4) {$\psi'(s)$};
  \draw[thick, (-)] (0,4)--(11.5,4);
  \node at (5.75,4.4) {$\psi'(t)$};
  \draw[thick, (-)] (0,2)--(2.5,2);
  \node at (1.25,2.4) {$\psi'(u_0)$};
  \node at (3.5,2.4) {$\cdots$};
  \draw[thick, (-)] (4.5,2)--(7,2);
  \node at (5.75,2.4) {$\psi'(u_i)$};
  \node at (8,2.4) {$\cdots$};
  \draw[thick, (-)] (9,2)--(11.5,2);
  \node at (10.25,2.4) {$\psi'(u_m)$};

\end{scope}

\begin{scope}[shift={(14.75,-6)}]%8.5x3

  \begin{scope}[color=light-gray]
  \draw (-0.5,1)--(8.5,1);
  \draw (-0.5,2)--(8.5,2);
  \draw (-0.5,3)--(8.5,3);
  \node at (-0.75,1) {$1$};
  \node at (-0.75,2) {$2$};
  \node at (-0.75,3) {$3$};
  \draw (0,0.5)--(0,3.5);
  \draw (1.5,0.5)--(1.5,3.5);
  \draw (3,0.5)--(3,3.5);
  \draw (5,0.5)--(5,3.5);
  \draw (6.5,0.5)--(6.5,3.5);
  \draw (8,0.5)--(8,3.5);
  \node at (0,0.1) {$0$};
  \node at (1.5,0.1) {$1$};
  \node at (3,0.1) {$2$};
  \node at (5,0.1) {$a_i-2$};
  \node at (6.5,0.1) {$a_i-1$};
  \node at (8,0.1) {$a_i$};
  \end{scope}
  
  \draw[thick, (-)] (0,1)--(8,1);
  \node at (4,1.4) {$\psi(s_i)$};
  \draw[thick, (-)] (0,3)--(8,3);
  \node at (4,3.4) {$\psi(t_i)$};
  \draw[thick, (-)] (0,2)--(1.5,2);
  \draw[thick, (-)] (1.5,2)--(3,2);
  \node at (4,2.4) {$\cdots$};
  \draw[thick, (-)] (5,2)--(6.5,2);
  \draw[thick, (-)] (6.5,2)--(8,2);

\end{scope}

    }
\end{tiny}
\end{scope}
\end{tikzpicture}
\caption{The graph $G$, partial representation $\psi'$, graph $H_i$, and the representation $\psi$ of $H_i$ with minimum width.}
\label{fig:3part}
\end{figure*}
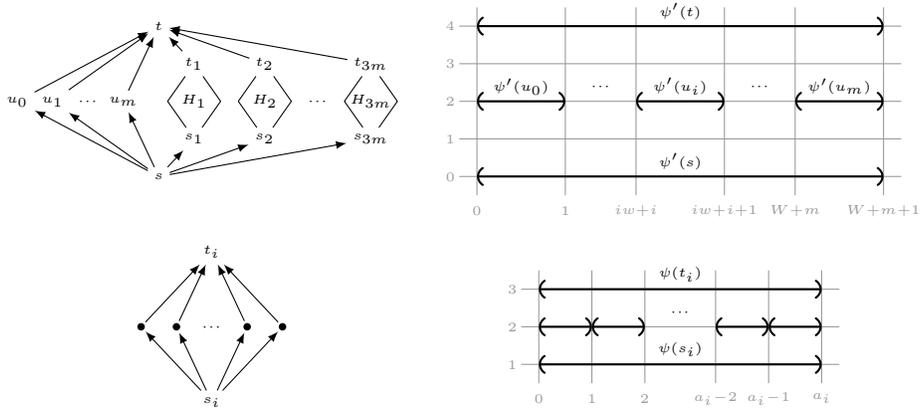

The graph $G$ is constructed as follows.
We start with $G'$ which is a $K_{2,m+1}$ with source $s$ and sink $t$ as the two vertices of degree $m+1$ and the other vertices are labeled $u_0, \ldots, u_{m}$.
Now, for each $i=1,\ldots,3m$, create an $st$-graph $H_i$ which is a $K_{2,a_i}$ with source $s_i$ and sink $t_i$ as its two vertices of degree $a_i$.
We remark that the width of any visibility representation of $H_i$ in the integer grid is at least $a_i$.
Finally, the graph $G$ is obtained by attaching each $H_i$ to $G'$ by adding the edges $(s,s_i)$ and $(t_i,t)$.

For the fixed bars, we let $V' = \{s,t,u_1, \ldots, u_{m+1}\}$ and we define a bar $\psi'(v)$ for each element $v \in V'$.
$$\hfill y_{\psi'}(v)=\left\{\begin{array}{ll} 0 & v=s \\ 2 & v=u_i \\ 4 & v=t \\ \end{array}\right. \hfill \begin{array}{lcl} l_{\psi'}(v)&=&\left\{\begin{array}{ll} 0 & v=s,t \\ iw+i & v=u_i \\ \end{array}\right. \\ r_{\psi'}(v)&=&\left\{\begin{array}{ll} W+m+1 & v=s,t \\ iw+i+1 & v=u_i \\ \end{array}\right. \end{array} \hfill$$
Observe that the distance between the right-end of $\psi'(u_{i})$ and the left-end of $\psi'(u_{i+1})$ is exactly $w$.

We now claim that there is a solution for input $(G,\psi')$ if and only if the \threepart{} instance $\{w,a_1, \ldots, a_{3m}\}$ has a solution.
First, we consider a collection of triples $T_1, \ldots, T_m$ which form a solution of the \threepart{} problem.
We extend $\psi'$ to a visibility representation of $G$ where, for each triple $T_j = \{a_{j_1}, a_{j_2}, a_{j_3}\}$, we place visibility representations of $H_{j_1}$, $H_{j_2}$, $H_{j_3}$ in sequence left-to-right and between the fixed bars $\psi'(u_{j-1})$ and $\psi'(u_j)$.
Notice that this is possible since $a_{j_1} + a_{j_2} + a_{j_3} = w$ and the distance from the right-end of $\psi'(u_{j-1})$ to the left-end of $\psi'(u_j)$ is $w$.

To show that any visibility representation $\psi$ of $G$ which extends $\psi'$ provides a solution to the \threepart{} problem we make the following observations.
First, due to the the placement of the bars $\psi'(s), \psi'(t), \psi'(u_0)$ and $\psi'(u_m)$, the outer face of the resulting embedding is $s,u_0,t,u_m$.
In particular, the bars of every $H_i$ occur strictly within the rectangle enclosed by $\psi'(s)$ and $\psi'(t)$.
Thus, each $H_i$ must also be drawn between some pair of bars $\psi'(u_{j-1})$, $\psi'(u_j)$ (for some $j \in \{1,\ldots, m-1\}$).
Second, due to each $a_i$ being between $\frac{w}{4}$ and $\frac{w}{2}$ and the width of a visibility representation of $H_i$ being at least $a_i$, at most three $H_i$'s `fit' between the fixed bars $\psi'(u_{j-1})$ and $\psi'(u_{j})$.
Thus, since there are $3m$ $H_i$'s, every $\psi'(u_{j-1})$, $\psi'(u_{j})$ has exactly three $H_i$'s between them.
Moreover, if $H_{i_1}, H_{i_2}$ and $H_{i_3}$ are placed between $\psi'(u_{j-1})$ and $\psi'(u_j)$, then $a_{i_1} + a_{i_2} + a_{i_3} \leq w$.
Thus, in $\psi$, the gaps between each pair $\psi'(u_{j-1})$, $\psi'(u_j)$ must contain precisely three $H_i$'s whose sum of corresponding $a_i$'s is $w$, \ie{}, the gaps correspond to the triples of a solution of the \threepart{} problem.
\end{proof}

\fi

\end{document}